\documentclass[letterpaper, abstracton, DIV11, 12pt, titlepage=false]{scrartcl}

\usepackage[ruled]{algorithm2e}

\SetAlFnt{\small}
\SetAlCapFnt{\small}
\SetAlCapNameFnt{\small}
\SetAlCapHSkip{0pt}
\IncMargin{-\parindent}
\let\chapter\undefined

\usepackage{enumerate}
\usepackage[round]{natbib}
\usepackage[english]{babel}
\usepackage{amsmath}
\usepackage{amsthm}
\usepackage{microtype}
\usepackage{pifont}
\usepackage{dsfont}
\usepackage[usenames,dvipsnames,svgnames,table]{xcolor}
\usepackage{tikz}
\usetikzlibrary{matrix,positioning}
\usepackage{numprint}
\npdecimalsign{.}
\usepackage{pgfplots}
\usepgfplotslibrary{groupplots}
\pgfplotsset{compat=newest}
\usepackage{bm}
\usepackage{mathabx}
\usepackage{setspace}
\usepackage[colorlinks=true,bookmarks=true, pdftex, citecolor = blue, linkcolor = blue,urlcolor = blue]{hyperref}
\onehalfspace

\DisableLigatures[f]{encoding = *, family = * }

\changenotsign

\def\bf{\normalfont\bfseries}
\changenotsign

\newcommand{\OPT}{\textsc{Opt}}
\newcommand{\PF}{\textsc{Pf}}
\newcommand{\FindOpt}{\textsc{FindOpt}}
\newcommand{\FindBounds}{\textsc{FindBounds}}
\newcommand{\FindLower}{\textsc{FindLower}}
\newcommand{\deficit}{\delta}

\newcommand{\variable}[1]{{\mathtt{#1}}}
\newcommand{\NMO}{\text{$(N,M,\deficit)$}}
\definecolor{LightGray}{gray}{0.85}
\definecolor{MedGray}{gray}{0.75}
\newcommand{\desifct}{{d}}

\DeclareMathOperator*{\argmin}{\arg\!\min}

\newtheorem{claim}{Claim}
\newtheorem{lp}{Linear Program}
\theoremstyle{plain}
\newtheorem{theorem}{Theorem}
\newtheorem{lemma}{Lemma}
\newtheorem{proposition}{Proposition}
\newtheorem{corollary}{Corollary}

\newtheorem{fact}{Fact}

\theoremstyle{definition}
\newtheorem{definition}{Definition}

\newtheorem{example}{Example}
\theoremstyle{remark}
\newtheorem{remark}{Remark}

\newcommand{\ThmASPEquivLinConStatement}{Given a setting $(N,M)$, a bound $\varepsilon \in [0,1]$, and a mechanism $\varphi$, the following are equivalent:
\begin{enumerate}
	\setlength{\itemsep}{0pt}
	\item \label{item:ThmASPEquivLinCon:eASP} $\varphi$ is $\varepsilon$-approximately strategyproof in $(N,M)$.
	\item \label{item:ThmASPEquivLinCon:fin_lin_con} For any agent $i\in N$, any preference profile $(P_i,P_{-i}) \in \mathcal{P}^N$, any misreport $P_i' \in \mathcal{P}$, and any rank $r \in \{1,\ldots,m\}$, we have
\begin{equation}
	\sum_{j \in M : \text{rank}_{P_i}(j) \leq r} \varphi_{j}(P_i',P_{-i}) - \varphi_{j}(P_i,P_{-i}) \leq \varepsilon.
\label{eq:finite_linear_incentive_constraint}
\end{equation}
\end{enumerate}}
\newcommand{\ThmASPEquivLinConProof}{%
Fix an agent $i$, a preference profile $(P_i,P_{-i})$, and a misreport $P_i'$.
The admissible set of utility functions for agent $i$ is $U_{P_i}$, i.e., all the utilities $u_i : M \rightarrow [0,1]$ for which $u_i(j) \geq u_i(j')$ whenever $P_i: j \succeq j'$.
Let $B^{\{0,1\}}(P_i)$ denote the set of \emph{binary utilities associated with $P_i$}, i.e.,
\begin{equation}
	B^{\{0,1\}}(P_i) = \left\{ u : M \rightarrow \{0,1\}~\left|~ u(j) \geq u(j') \text{ whenever } P_i: j \succeq j' \right. \right\}.
\end{equation}

We first show the direction ``$\Leftarrow$,'' i.e., the condition \ref{item:ThmASPEquivLinCon:fin_lin_con} implies $\varepsilon$-approximate strategyproofness (\ref{item:ThmASPEquivLinCon:eASP}).
Let $u \in B^{\{0,1\}}(P_i)$, then the incentive constraint (\ref{EQ:INCENTIVE_CONSTRAINTS_APPROXIMATE_STRATEGYPROOFNESS}) for this particular utility function has the form
\begin{eqnarray}
		\varepsilon\left(u,\left(P_i,P_{-i}\right),P_i',\varphi\right) & = & \sum_{j \in M} u(j) \cdot \left( \varphi_j(P_i',P_{-i}) - \varphi_j(P_i,P_{-i}) \right) \\
			& = & \sum_{j \in M : u(j) = 1} \varphi_j(P_i',P_{-i}) - \varphi_j(P_i,P_{-i}) \\
			& = & \sum_{j \in M : \text{rank}_{P_i}(j) \leq r} \varphi_j(P_i',P_{-i}) - \varphi_j(P_i,P_{-i})
\end{eqnarray}
for some $k \in \{1,\ldots,m\}$.
By (\ref{eq:finite_linear_incentive_constraint}) from the condition \ref{item:ThmASPEquivLinCon:fin_lin_con}, this term is always upper bounded by $\varepsilon$.
By Lemma \ref{lem:type_in_conv_hull_of_binary_utilities}, $U_{P_i} \subseteq \textit{Conv}\left(B^{\{0,1\}}(P_i)\right)$, which means that any $u_i \in U_{P_i}$ that represents the preference order $P_i$ can be written as a convex combination of utility functions in $B^{\{0,1\}}(P_i)$, i.e.,
\begin{equation}
	u_i = \sum_{l = 1}^L \alpha_l u^l
\end{equation}
for $u^l \in B^{\{0,1\}}(P_i)$ and $\alpha_l \geq 0$ for all $l\in \{1,\ldots,L\}$ and $\sum_{l = 1}^L \alpha_l = 1$.
By linearity of the incentive constraint (\ref{EQ:INCENTIVE_CONSTRAINTS_APPROXIMATE_STRATEGYPROOFNESS}) we get that
\begin{eqnarray}
	\varepsilon\left(u_i,\left(P_i,P_{-i}\right),P_i',\varphi\right) & = & \sum_{l=1}^L \alpha_l \varepsilon\left(u^l,\left(P_i,P_{-i}\right),P_i',\varphi\right) \\
		&\leq& \sum_{l=1}^L \alpha_l \varepsilon = \varepsilon.
\end{eqnarray}
This proves the direction ``$\Leftarrow$.''

Next, we prove the direction ``$\Rightarrow$''.
Towards contradiction, assume that the constraint (\ref{eq:finite_linear_incentive_constraint}) is violated for some $k \in \{1,\ldots,m\}$, i.e.,
\begin{equation}
	\sum_{j \in M : \text{rank}_{P_i}(j) \leq k} \varphi_{j}(P_i',P_{-i}) - \varphi_{j}(P_i,P_{-i}) = \varepsilon + \delta
\end{equation}
with $\delta > 0$.
Let $u \in B^{\{0,1\}}(P_i)$ be the binary utility function with
\begin{equation}
	u(j) = \left\{ \begin{array}{ll} 1, & \text{ if } \text{rank}_{P_i}(j) \leq k, \\ 0, & \text{ else}. \end{array} \right.
\end{equation}
Then
\begin{equation}
	\varepsilon\left(u,\left(P_i,P_{-i}\right),P_i',\varphi\right) = \varepsilon + \delta.
\end{equation}
Choose any utility function $u' \in U_{P_i}$ and let $\beta = \frac{\delta/2}{\varepsilon+\delta+1}$.
The utility function constructed by $\tilde{u} = (1-\beta)u + \beta u'$ represents $P_i$ and we have
\begin{eqnarray}
	\varepsilon \left(\tilde{u},\left(P_i,P_{-i}\right),P_i',\varphi\right) & = & (1-\beta)\varepsilon\left(u,\left(P_i,P_{-i}\right),P_i',\varphi\right) + \beta \varepsilon\left(u',\left(P_i,P_{-i}\right),P_i',\varphi\right) \\
		& \geq & (1-\beta)(\varepsilon + \delta) - \beta \\
		& = & -\beta(\varepsilon+\delta+1) + (\varepsilon+\delta) = \varepsilon + \delta/2,
\end{eqnarray}
since the change in utility from manipulation is lower bounded by $-1$.
Thus, the $\varepsilon$-approximate strategyproofness constraint is violated (for the utility function $\tilde{u}$, a contradiction.
\begin{lemma}
\label{lem:type_in_conv_hull_of_binary_utilities}
For any preference order $P \in \mathcal{P}$ define the set of \emph{binary utilities associated with $P$} by
\begin{equation}
	B^{\{0,1\}}(P) = \left\{ u : M \rightarrow \{0,1\}~\left|~ u(j) \geq u(j') \text{ whenever } P:j \succeq j' \right. \right\}.
\end{equation}
Then $U_P \subseteq \textit{Conv}\left(B^{\{0,1\}}(P)\right)$, where $\textit{Conv}(\cdot)$ denotes the convex hull of a set.
\end{lemma}
\begin{proof}[Proof of Lemma \ref{lem:type_in_conv_hull_of_binary_utilities}]
First, suppose that the preference order $P$ is strict, i.e., $P: j \succ j'$ or $P: j' \succ j$ for all $j \neq j'$, and without loss of generality,
\begin{equation}
	P: j_1 \succ j_2 \succ \ldots \succ j_m.
\end{equation}
In this case, $B^{\{0,1\}}(P)$ consists of all the functions $u^r:M\rightarrow \{0,1\}$ with
\begin{equation}
	u^k(j_r) = \left\{ \begin{array}{ll} 1, & \text{ if }r \leq k, \\ 0, & \text{ else}.\end{array} \right.
\end{equation}
With $k \in {0,\ldots,m}$.
Let $\Delta u(0) = 1-u(j_1)$, $\Delta u(k) = u(j_k) - u(j_{k+1})$ for all $k \in \{1,\ldots,m-1\}$, and $\Delta u(m) = u(j_m)$. Then represent $u$ by
\begin{equation}
	u(j_r) = \sum_{k = r}^{m} \Delta u(k).
\end{equation}
Now we construct the utility function
\begin{equation}
	\tilde{u} = \sum_{k = 0}^m \Delta u(k) \cdot u^k.
\end{equation}
First note that $\sum_{k=0}^m \Delta u(k) = 1$ and $\Delta u(k) \geq 0 $ for all $k$, so that $\tilde{u}$ is a convex combination of elements of $B^{\{0,1\}}(P)$.
Furthermore, for any $j_r \in M$, we have that
\begin{eqnarray}
	\tilde{u}(j_r) & = & \sum_{k = 0}^m \Delta u(k) \cdot u^k(j_r) \\
		& = & \sum_{k = r}^m \Delta u(k) = u (j_r), \\
\end{eqnarray}
i.e., $u = \tilde{u}$.
This establishes the Lemma for strict preference orders.
 
For arbitrary preference orders (i.e., with indifferences) the proof can be easily extended by combining all alternatives about which an agent with preference order $P$ is indifferent into a single virtual alternative.
Then we apply the proof for the strict case.
The utility functions in $B^{\{0,1\}}(P)$ will be exactly those that put equal value on the groups of alternatives between which an agent with preference order $P$ is indifferent.
This concludes the proof of the Lemma.
\end{proof}
This concludes the proof of Theorem \ref{THM:ASP_EQUIV_LINEAR_CONSTRAINTS}.}

\newcommand{\PropOptimalMechanismsNotEmptyStatement}{
Given a problem $(N,M,\deficit)$ and a manipulability bound $\varepsilon \in [0,1]$, there exists at least one mechanism that is optimal at $\varepsilon$.%
}
\newcommand{\PropOptimalMechanismsNotEmptyProof}{
A strategyproof mechanism always exists (e.g., the constant mechanism), and any strategyproof mechanism is also $\varepsilon$-approximately strategyproof. 
Thus, the set of candidates for $\OPT(\varepsilon)$ is never empty. 
Since the deficit of any mechanism is upper bounded by $1$, we get $\deficit(\varepsilon) \leq 1$ for all $\varepsilon \in [0,1]$.
Next, for some $\varepsilon \in [0,1]$ let
\begin{equation}
\deficit(\varepsilon) = \inf \left\{ \deficit(\varphi)~|~\varphi\ \varepsilon\text{-approximately strategyproof} \right\}.
\end{equation}
By definition (of the infimum), there exists a sequence of mechanisms $(\varphi_k)_{k \geq 1}$ such that $\deficit(\varphi^k) \rightarrow \deficit(\varepsilon)$ as $k \rightarrow \infty$.
Since all $\varphi_k$ are mechanisms, they are uniformly bounded functions from the finite set $\mathds{P}^N$ of preference profiles to the compact set $\Delta(M)$.
Therefore, some sub-sequence of $(\varphi_k)_{k \geq 1}$ converges uniformly to a limit $\tilde{\varphi}$, which must itself be a mechanism. 

By Theorem \ref{THM:ASP_EQUIV_LINEAR_CONSTRAINTS}, $\varepsilon$-approximate strategyproofness is equivalent to a finite set of weak, linear inequalities. 
Since these are all satisfied by the elements of the sequence $(\varphi_k)_{k \geq 1}$, they are also satisfied by $\tilde{\varphi}$, i.e., $\tilde{\varphi}$ is an $\varepsilon$-approximately strategyproof mechanism. 
The deficit of a mechanism $\varphi$ is either a maximum (worst-case deficit) or a weighted average (ex-ante deficit) of a finite set of values, but in any case, it is a uniformly continuous projection of $\varphi$ onto the compact interval $[0,1]$. 
Thus, by uniform convergence, the deficit of $\deficit(\tilde{\varphi})$ must be the limit of the sequence of deficits $(\deficit(\varphi_k))_{k \geq 1}$, i.e., $\deficit(\varepsilon)$, which yields $\tilde{\varphi} \in \OPT(\varepsilon)$.}

\newcommand{\ThmHybridGuaranteesStatement}{
Given a problem $(N,M,\deficit)$, mechanisms $\varphi,\psi$, and $\beta \in [0,1]$, we have
\begin{eqnarray}
	\varepsilon(h_{\beta}) & \leq &  (1-\beta) \varepsilon(\varphi) + \beta \varepsilon(\psi), \label{eq:guarantees_manip_bound} \\
	\deficit(h_{\beta}) & \leq &  (1-\beta) \deficit(\varphi) + \beta \deficit(\psi). \label{eq:guarantees_deficit_bound}
\end{eqnarray}
}
\newcommand{\ThmHybridGuaranteesProof}{
The manipulability of a mechanism is determined by the maximum manipulability across all agents $i\in N$, preference profiles $(P_i,P_{-i})\in \mathcal{P}^N$, and misreports $P_i'\in \mathcal{P}$.
Using the triangle inequality for this $\max$-operator, we get that
\begin{eqnarray}
	\varepsilon(h_{\beta}) & = & \max \varepsilon\left(u,\left(P_i,P_{-i}\right),P_i',h_{\beta}\right) \\
		& = & \max \sum_{j \in M} u_i(j) \cdot \left( (h_{\beta})_j(P_i',P_{-i}) - (h_{\beta})_j(P_i,P_{-i}) \right) \\
		& \leq & (1-\beta) \max \sum_{j \in M} u_i(j) \cdot \left( \varphi_j(P_i',P_{-i}) - \varphi_j(P_i,P_{-i}) \right) \\
		& & + \beta \max \sum_{j \in M} u_i(j) \cdot \left( \psi_j(P_i';P_{-i}) - \psi_j(P_i,P_{-i}) \right) \\
		& = & (1-\beta) \varepsilon(\varphi) + \beta \varepsilon(\psi).
\end{eqnarray}
For worst-case deficit, observe that the extended desideratum function for random outcomes is linear in the first argument, i.e., $\desifct((1-\beta) x + \beta y, \bm P) = (1-\beta) \desifct(x,\bm P) + \beta \desifct(y,\bm P)$ for any outcomes $x,y$ preference profile $\bm P$. 
Thus, 
\begin{equation}
	\deficit(h_{\beta}) = \max \deficit\left(h_{\beta}(\bm P), \bm P\right) \leq (1-\beta) \varepsilon(\varphi) + \beta \varepsilon(\psi)
\end{equation}
follows analogously. 
Since the ex-ante deficit is simply a single linear function of the mechanism, we get exact equality, i.e., $	\deficit(h_{\beta}) = (1-\beta) \varepsilon(\varphi) + \beta \varepsilon(\psi)$.
}

\newcommand{\CorMonotonicConvexStatement}{%
Given a problem \NMO, the mapping $\varepsilon \mapsto \deficit(\varepsilon)$ 
is monotonically decreasing and convex. }
\newcommand{\CorMonotonicConvexProof}{%
Monotonicity follows (almost) by definition: 
observe a mechanism $\varphi \in \OPT(\varepsilon)$ is also a candidate for $\OPT(\varepsilon')$ for any larger manipulability bound $\varepsilon' \geq \varepsilon$. 
The only reason for $\varphi$ to not be optimal at $\varepsilon'$ is that some other $\varepsilon'$-approximately strategyproof mechanism has strictly lower deficit. 
Thus, the mapping $\varepsilon\mapsto\deficit(\varepsilon)$ is weakly monotonic and decreasing. 

To see convexity, assume towards contradiction that the mapping is not convex. 
Then there must exist bounds $\varepsilon, \varepsilon' \in [0,1]$ and $\beta \in [0,1]$, such that for $\varepsilon_{\beta} = (1-\beta) \varepsilon + \beta \varepsilon'$, we have 
\begin{equation}
	\deficit(\varepsilon_{\beta}) > (1-\beta) \deficit(\varepsilon) + \beta \deficit(\varepsilon'). 
\end{equation}
Let $\varphi\in\OPT(\varepsilon)$ and $\varphi'\in\OPT(\varepsilon')$ and consider the hybrid $h_{\beta} = (1 - \beta) \varphi + \beta \varphi'$. 
By Theorem \ref{THM:HYBRID_GUARANTEES}, this hybrid has a manipulability of at most $\varepsilon_{\beta}$ and a deficit of at most $(1-\beta) \deficit(\varepsilon) + \beta \deficit(\varepsilon')$. 
Thus, 
\begin{equation}
	\deficit(\varepsilon_{\beta}) \leq \deficit(h_{\beta}) \leq (1-\beta) \deficit(\varepsilon) + \beta \deficit(\varepsilon') < \deficit(\varepsilon_{\beta}), 
\end{equation}
a contradiction.
}

\newcommand{\ThmFrontierStructureStatement}{%
Given a problem $(N,M,\deficit)$, there exists a finite set of \emph{supporting manipulability bounds}
\begin{equation}
	\varepsilon_0 = 0 < \varepsilon_1 < \ldots < \varepsilon_{K-1} < \varepsilon_K = \bar{\varepsilon},
\end{equation}
such that for any $k \in \{1,\ldots,K\}$ and any $\varepsilon \in [\varepsilon_{k-1},\varepsilon_{k}]$ with $\beta = \frac{\varepsilon - \varepsilon_{k-1}}{\varepsilon_k - \varepsilon_{k-1}}$ we have that 
\begin{eqnarray}
	&\OPT(\varepsilon) = 
	(1-\beta) \OPT(\varepsilon_{k-1}) + \beta \OPT(\varepsilon_k), &\\
	&\deficit(\varepsilon) = 
	(1-\beta) \deficit(\varepsilon_{k-1}) + \beta \deficit(\varepsilon_k).& 
\end{eqnarray}}
\newcommand{\ThmFrontierStructureProof}{%
From Section \ref{SEC:OPTIMAL} we know that for each $\varepsilon \in [0,\bar{\varepsilon}]$ we can write $\OPT(\varepsilon)$ as the set of solutions to a linear program, i.e.,
\begin{eqnarray}
	\OPT(\varepsilon) = & &  \argmin_{x} \left\langle v, x \right\rangle \\
		\text{s.t.} & & D x \leq d, \\
			& & A x \leq \mathbf{\varepsilon},
\end{eqnarray}
where $D$ and $A$ are matrices, $v$ and $d$ are vectors, $\mathbf{\varepsilon}$ is a vector with all entries equal to $\varepsilon$, and $x$ is a vector of variables of dimension $L$.
Observe that $\varepsilon$ enters the constraints only as the upper bound in a number of linear inequalities.
The proof utilizes this characterization of $\OPT(\varepsilon)$.

Before we proceed with the proof of Theorem \ref{THM:FRONTIER_STRUCTURE}, we require a number of definitions.
We denote by $F_{\varepsilon}$ the set of \emph{feasible points at $\varepsilon$}, i.e.,
\begin{equation}
	F_{\varepsilon} = \left\{ x~|~Dx \leq d, Ax \leq \varepsilon \right\},
\end{equation}
and we denote by $S_{\varepsilon}$ the set of \emph{solutions at $\varepsilon$}, i.e.,
\begin{equation}
	S_{\varepsilon} = \argmin_{x \in F_{\varepsilon}} \left\langle v , x \right\rangle.
\end{equation}
A constraint is a row $C_l$ of either the matrix $A$ or the matrix $D$ with the corresponding bound $c_l$ equal to respective entry of $d$ or $\mathbf{\varepsilon}$.
A feasible point $x \in F_{\varepsilon}$ is an \emph{extreme point of $F_{\varepsilon}$} if there exist $L$ independent constraints $C_1, \ldots, C_L$ such that
\begin{equation}
	C_l x = c_l
\end{equation}
for all $l \in \{ 1,\ldots,L\}$, i.e., the constraints are satisfied with equality at $x$.
$x$ is then said to be an extreme point of $F_{\varepsilon}$ \emph{with respect to $(C_1,\ldots,C_L)$}.
We say that the set of constraints $(C_1, \ldots, C_L)$ is \emph{restrictive at $\varepsilon$} if they are independent and there exists an extreme point in $F_{\varepsilon}$ with respect to these constraints.
Let
\begin{equation}
	\mathcal{R}(\varepsilon) = \left\{ (C_1,\ldots,C_L)~|~(C_1, \ldots, C_L) \text{ is restrictive at } \varepsilon\right\}
\end{equation}
be the set of all sets of constraints that are restrictive at $\varepsilon$.
A set of restrictive constraints $\bm C \in \mathcal{R}(\varepsilon)$ is \emph{binding at $\varepsilon$} if the extreme point $x \in F_{\varepsilon}$ where the constraints of $\bm C$ are satisfied with equality is a solution, i.e., $x \in S_{\varepsilon}$.
Let
\begin{equation}
	\mathcal{B}(\varepsilon) = \left\{ \bm C \in \mathcal{R}(\varepsilon)~|~\bm C \text{ is binding at } \varepsilon\right\}
\end{equation}
be the set of all sets of constraints that are binding at $\varepsilon$.
We denote by $E(F_{\varepsilon})$ and $E(S_{\varepsilon})$ the extreme points of $F_{\varepsilon}$ and $S_{\varepsilon}$, respectively.

Observe that since $F_{\varepsilon}$ and $S_{\varepsilon}$ are polytopes and bounded by finitely many hyperplanes, the extreme points $E(F_{\varepsilon})$ and $E(S_{\varepsilon})$ form minimal $\mathcal{V}$-representations of $F_{\varepsilon}$ and $S_{\varepsilon}$ (see, e.g., p.51ff in \citep{Gruenbaum2003ConvexPolytopesTEXTBOOK}).
Thus, $\textit{Conv}(E(F_{\varepsilon})) = F_{\varepsilon}$ and $\textit{Conv}(E(S_{\varepsilon})) = S_{\varepsilon}$.
Furthermore, since $S_{\varepsilon} \subseteq F_{\varepsilon}$, any extreme point of $S_{\varepsilon}$ is also an extreme point of $F_{\varepsilon}$, i.e., $E(S_{\varepsilon}) \subseteq E(F_{\varepsilon})$.
Finally, each extreme point is uniquely determined by the set of constraints with respect to which it is extreme, i.e., if there exists an extreme point with respect to a set of constraints $\bm C \in \mathcal{R}(\varepsilon)$, then this point is unique.

\begin{claim}
\label{claim:path_efficient_frontier.1}
For $\varepsilon_0, \varepsilon_1 \in [0,\bar{\varepsilon}]$ with $\varepsilon_0 < \varepsilon_1$, if $x_0 \in F_{\varepsilon_0}$ and $x_1 \in F_{\varepsilon_1}$, then for any $\gamma \in [0,1]$ and $\varepsilon = (1-\gamma) \varepsilon_0 + \gamma\varepsilon_1$ we have that
\begin{equation}
	x = (1-\gamma) x_0 + \gamma x_1 \in F_{\varepsilon}.
\end{equation}
\end{claim}
\begin{proof} By assumption, $D_k x_0 \leq d_k$ and $D_k x_1 \leq d_k$ for all $k$.
Thus,
\begin{equation}
	D_k((1-\gamma) x_0 + \gamma x_1) = (1-\gamma) D_k x_0 + \gamma D_k x_1 \leq d_k.
\end{equation}
Furthermore, $A_k x_0 \leq \varepsilon_0$ and $A_k x_1 \leq \varepsilon_1$ for all $k$, which implies
\begin{equation}
	A_k((1-\gamma) x_0 + \gamma x_1) = (1-\gamma) A_k x_0 + \gamma A_k x_1 \leq (1-\gamma) \varepsilon_0 + \gamma \varepsilon_1  = \varepsilon.
\end{equation}
\end{proof}

\begin{claim}
\label{claim:path_efficient_frontier.2}
For $\varepsilon_0, \varepsilon_1 \in [0,\bar{\varepsilon}]$ with $\varepsilon_0 < \varepsilon_1$, if $\bm C \in \mathcal{R}(\varepsilon_0)$ and $\bm C\in \mathcal{R}(\varepsilon_1)$, then for any $\gamma\in [0,1]$ and $\varepsilon = (1-\gamma) \varepsilon_0 + \gamma \varepsilon_1$ we have that $\bm C \in \mathcal{R}\left(\varepsilon\right)$, and the $\gamma$-convex combination of the extreme points at $\varepsilon_0$ and $\varepsilon_1$ with respect to $\bm C$ is the unique extreme point at $\varepsilon$ with respect to $\bm C$.
\end{claim}
\begin{proof} Since $\bm C \in \mathcal{R}(\varepsilon_0)$ and $\bm C\in \mathcal{R}(\varepsilon_1)$, there exist unique extreme points $x_0 \in E(F_{\varepsilon_0})$ and $x_1 \in E(F_{\varepsilon_0})$ with respect to $\bm C$.
By Claim \ref{claim:path_efficient_frontier.1}, the point $x = (1-\gamma) x_0 + \gamma x_1$ is feasible at $\varepsilon$.
For any $l \in \{1,\ldots,L\}$ if $C_l = D_k$ for some $k$ we have that
\begin{equation}
	C_l x = (1-\gamma) D_k x_0 + \gamma D_k x_1 = (1-\gamma) d_k + \gamma d_k = d_k = c_l,
\end{equation}
i.e., $x$ satisfies the constraint $C_l$ with equality.
If $C_l = A_k$, then the constraint is also satisfied with equality, since
\begin{equation}
	C_l x = (1-\gamma) A_k x_0 + \gamma A_k x_1 = (1-\gamma) \varepsilon_0 + \gamma \varepsilon_1 = \varepsilon.
\end{equation}
Consequently, $x$ is the unique extreme point at $\varepsilon$ with respect to $\bm C$.
This in turn implies that $\bm C$ is restrictive at $\varepsilon$.
\end{proof}

\begin{claim}
\label{claim:path_efficient_frontier.3}
There exists a finite decomposition
\begin{equation}
	0 = \varepsilon_0 < \varepsilon_1 < \ldots < \varepsilon_K = \bar{\varepsilon}
\end{equation}
of the interval $[0,\bar{\varepsilon}]$, such that on each interval $[\varepsilon_{k-1},\varepsilon_{k}]$ we have that $\mathcal{R}(\varepsilon) = \mathcal{R}(\varepsilon')$ for all $\varepsilon,\varepsilon' \in [\varepsilon_{k-1},\varepsilon_{k}]$.
\end{claim}
\begin{proof}
By Claim \ref{claim:path_efficient_frontier.2}, if some set of $L$ independent constraints $\bm C$ is restrictive at some $\varepsilon\in [0,\bar{\varepsilon}]$, then the set of $\varepsilon' \in [0,\bar{\varepsilon}]$ where $\bm C$ is also restrictive must be compact interval $[\varepsilon^-_{\bm C},\varepsilon^+_{\bm C}] \subseteq [0,\bar{\varepsilon}]$ with $\varepsilon \in [\varepsilon^-_{\bm C},\varepsilon^+_{\bm C}]$.
Since there is a finite number of constraints, there is also a finite number of constraint sets $\bm C$.
Consider the set
\begin{equation}
	\{\varepsilon_0, \ldots,\varepsilon_K\} = \bigcup_{\bm C\text{ set of }L\text{ indep. constraints}}\{ \varepsilon_{\bm C}^-,\varepsilon_{\bm C}^+ \}.
\end{equation}
Observe that by construction, a set of $L$ independent constraints ${\bm C}$ becomes restrictive or un-restrictive only at one of the finitely many $\varepsilon_{k-1}$.
This proves the claim.
\end{proof}

\begin{claim}
\label{claim:path_efficient_frontier.4}
On each interval $[\varepsilon_{k-1},\varepsilon_{k}]$ from Claim \ref{claim:path_efficient_frontier.3} and for any $\varepsilon \in [\varepsilon_{k-1},\varepsilon_{k}]$, if ${\bm C} \in \mathcal{R}(\varepsilon)$, then ${\bm C}\in \mathcal{R}(\varepsilon_{k-1}) \cap \mathcal{R}(\varepsilon_{k})$.
\end{claim}
\begin{proof} Assume towards contradiction that ${\bm C} \in \mathcal{R}(\varepsilon)$, but ${\bm C} \notin \mathcal{R}(\varepsilon_{k-1})$.
Then there exists an $\varepsilon' \in (\varepsilon_{k-1},\varepsilon) \subseteq (\varepsilon_{k-1},\varepsilon_{k}) $, where ${\bm C}$ become restrictive for the first time.
Then $\varepsilon' \in \{\varepsilon_0,\varepsilon_{k-1}\}$, and therefore $[\varepsilon_{k-1}, \varepsilon_{k}]$ cannot be one of the intervals in the decomposition.
Instead, it would be split by $\varepsilon'$, a contradiction.
\end{proof}

\begin{claim}
\label{claim:path_efficient_frontier.5}
On each interval $[\varepsilon_{k-1},\varepsilon_{k}]$ from Claim \ref{claim:path_efficient_frontier.3} and for any $\gamma \in [0,1]$ with $\varepsilon = (1-\gamma) \varepsilon_{k-1} + \gamma \varepsilon_{k}$ we have that
\begin{equation}
	F_{\varepsilon} = (1-\gamma) F_{\varepsilon_{k-1}} + \gamma F_{\varepsilon_{k}},
\end{equation}
i.e., the set of feasible points at $\varepsilon$ is equal to the $\gamma$-convex combination of feasible points at $\varepsilon_{k-1}$ and $\varepsilon_{k}$.
\end{claim}
\begin{proof}
By Claim \ref{claim:path_efficient_frontier.1} we have
\begin{equation}
	F_{\varepsilon} \supseteq (1-\gamma) F_{\varepsilon_{k-1}} + \gamma F_{\varepsilon_{k}}.
\end{equation}
By Claim \ref{claim:path_efficient_frontier.4} the extreme points of $F_{\varepsilon}$ are the $\gamma$-convex combinations of extreme points of $F_{\varepsilon_{k-1}}$ and $F_{\varepsilon_{k}}$.
Since $F_{\varepsilon} = \textit{Conv}(E(F_{\varepsilon}))$, this implies
\begin{equation}
	F_{\varepsilon} \subseteq (1-\gamma) F_{\varepsilon_{k-1}} + \gamma F_{\varepsilon_{k}},
\end{equation}
which concludes the proof of the claim.
\end{proof}

\begin{claim}
\label{claim:path_efficient_frontier.6}
On each interval $[\varepsilon_{k-1},\varepsilon_{k}]$ from Claim \ref{claim:path_efficient_frontier.3}, if ${\bm C} \in \mathcal{B}(\varepsilon)$ for some $\varepsilon \in (\varepsilon_{k-1},\varepsilon_{k})$, then ${\bm C} \in \mathcal{B}(\varepsilon_{k-1}) \cap \mathcal{B}(\varepsilon_{k})$.
Furthermore, the extreme point of $S_{\varepsilon}$ with respect to ${\bm C}$ is the $\gamma$-convex combination of the extreme points of $S_{\varepsilon_{k-1}}$ and $S_{\varepsilon_{k}}$ with respect to ${\bm C}$ with $\gamma = \frac{\varepsilon - \varepsilon_{k-1}}{\varepsilon_{k} - \varepsilon_{k-1}}$.
\end{claim}
\begin{proof} Since ${\bm C} \in \mathcal{B}(\varepsilon)$, there exists a unique extreme point $x \in E(S_{\varepsilon})$ such that $x$ is extreme at $\varepsilon$ with respect to ${\bm C}$.
	By Claim \ref{claim:path_efficient_frontier.5}, we can represent $x = (1-\gamma) x_0 + \gamma x_1$ with $x_0 \in F_{\varepsilon_{k-1}}$, $x_1 \in F_{\varepsilon_{k}}$.
	By Claim \ref{claim:path_efficient_frontier.2}, $x_0$ and $x_1$ are extreme points of $F_{\varepsilon_{k-1}}$ and $F_{\varepsilon_{k}}$, respectively.
	
	Suppose towards contradiction that $x_0 \notin S_{\varepsilon_{k-1}}$.
	Then there exists $x_0' \in S_{\varepsilon_{k-1}}$ such that
	\begin{equation}
		\left\langle v, x_0' \right\rangle < \left\langle v, x_0 \right\rangle.
	\end{equation}
	By Claim \ref{claim:path_efficient_frontier.1}, $x' = (1-\gamma) x_0' + \gamma x_1$ is in $F_{\varepsilon}$ and
	\begin{eqnarray}
		\left\langle v, x' \right\rangle & = & (1-\gamma) \left\langle v, x_0' \right\rangle + \gamma \left\langle v, x_1 \right\rangle \\
			& < & (1-\gamma) \left\langle v, x_0 \right\rangle + \gamma \left\langle v, x_1 \right\rangle \\
			& = & \left\langle v, x \right\rangle,
	\end{eqnarray}
	i.e., $x'$ is feasible at $\varepsilon$ and has lower objective than $x$.
	This contradicts the assumption that $x \in S_{\varepsilon}$.
	A similar argument yields $x_1 \in S_{\varepsilon_{k}}$, which concludes the proof of the claim.
\end{proof}

\begin{claim}
\label{claim:path_efficient_frontier.7}
On each interval $[\varepsilon_{k-1},\varepsilon_{k}]$ from Claim \ref{claim:path_efficient_frontier.3} and any $\gamma \in [0,1]$ with $\varepsilon = (1-\gamma) \varepsilon_{k-1} + \gamma \varepsilon_{k}$ we have that
\begin{equation}
	S_{\varepsilon} = (1-\gamma) S_{\varepsilon_{k-1}} + \gamma S_{\varepsilon_{k}},
\end{equation}
i.e., the set of solutions at $\varepsilon$ is equal to the $\gamma$-convex combination of solutions at $\varepsilon_{k-1}$ and $\varepsilon_{k}$.
\end{claim}
\begin{proof} By Claim \ref{claim:path_efficient_frontier.1} we have
\begin{equation}
	S_{\varepsilon} \supseteq (1-\gamma) S_{\varepsilon_{k-1}} + \gamma S_{\varepsilon_{k}} .
\end{equation}
By Claim \ref{claim:path_efficient_frontier.6} the extreme points of $S_{\varepsilon}$ are the $\gamma$-convex combinations of extreme points of $S_{\varepsilon_{k-1}}$ and $S_{\varepsilon_{k}}$.
Since $S_{\varepsilon} = \textit{Conv}(E(S_{\varepsilon}))$, this implies
\begin{equation}
	S_{\varepsilon} \subseteq (1-\gamma) S_{\varepsilon_{k-1}} + \gamma S_{\varepsilon_{k}},
\end{equation}
which concludes the proof of the claim.
\end{proof}
Claim \ref{claim:path_efficient_frontier.7} is the main step in the proof of Theorem \ref{THM:FRONTIER_STRUCTURE}: every solution $x \in S_{\varepsilon}$ corresponds to some optimal mechanism $\varphi \in \OPT(\varepsilon)$.
Furthermore, by the nature of the representation of mechanisms in the linear program, the convex combination of two solutions corresponds to the hybrid mechanism of the two mechanisms. 
Thus, 
\begin{equation}
  S_{\varepsilon} = (1-\gamma) S_{\varepsilon_{k-1}} + \gamma S_{\varepsilon_{k}},
\end{equation}
implies
\begin{equation}
  \OPT(\varepsilon) = (1-\gamma) \OPT(\varepsilon_{k-1}) + \gamma \OPT(\varepsilon_{k}).
\end{equation}
Since the objective value $\variable{d}$ is a variable in the solution, we get 
\begin{equation}
  \delta(\varepsilon) = (1-\gamma) \delta(\varepsilon_{k-1}) + \gamma \delta(\varepsilon_{k}).
\end{equation}}

\newcommand{\PropComputableStatement}{%
Given a problem $(N,M,\deficit)$, the algorithms \FindLower\ and \FindBounds\ require at most $4K + \log_2\left(1/\varepsilon_{1}\right) - 1$ executions of the linear program \FindOpt\ to determine all supporting manipulability bounds of the Pareto frontier, where $K$ is the number of supporting manipulability bounds and $\varepsilon_1$ is the smallest non-trivial supporting manipulability bound.}
\newcommand{\PropComputableProof}{
We defined $\textit{signature}(\varepsilon)$ (or $\textit{sign}(\varepsilon)$ for short) as a function that uses the linear program \FindOpt\ to determine the signature of some mechanism that is optimal at $\varepsilon$, i.e., $(\varepsilon,\deficit(\varepsilon))$. 
Algorithm \FindLower\ initially calls to the function $\textit{sign}$ 3 times, for $0$, $1$, and $1/2$. 
Now suppose that $\varepsilon_1 \in \left[\frac{1}{2^{n}},\frac{1}{2^{n-1}}\right)$ for some $n \geq 1$. 
Then $\textit{sign}$ will be called for $\frac{1}{4},\frac{1}{8},\ldots,\frac{1}{2^{n+1}}$ until $\underline{\varepsilon}=\frac{1}{2^{n+1}}$ is returned. 
Note that while $\underline{\varepsilon}=\frac{1}{2^{n}}$ would have been sufficiently small, the algorithm needs to try $\underline{\varepsilon}=\frac{1}{2^{n+1}}$ to verify this. 
Thus, it takes $n = \log_2(1/\varepsilon_1)$ calls to $\textit{sign}$. 

The remainder of the proof is concerned with the algorithm \FindBounds.
Any \emph{segment} consists of two points $\left(\textit{sign}(\varepsilon),\textit{sign}(\varepsilon')\right)$. 
Initially, there are two \emph{outer} segments $\left(\textit{sign}(0),\textit{sign}(\underline{\varepsilon})\right)$ and $\left(\textit{sign}(1),\textit{sign}(2)\right)$, which are needed to help the algorithm get started.
The algorithm maintains a decomposition of the interval $[\underline{\varepsilon},1]$, which initially consists of a single \emph{unverified} segment.
In every execution of the \texttt{while}-loop, \FindBounds\ selects an unverified segment $\left(\textit{sign}(\varepsilon^-),\textit{sign}(\varepsilon^+)\right)$.
Then it uses the segments to the left and right $\left(\textit{sign}(\varepsilon^{--}),\textit{sign}(\varepsilon^-)\right)$ and $\left(\textit{sign}(\varepsilon^+),\textit{sign}(\varepsilon^{++})\right)$ to ``guess'' the position of a new supporting manipulability bound between $\varepsilon^-$ and $\varepsilon^+$.
This guess $\mathtt{e}$ is the $\varepsilon$-value of the point of intersection of the affine hulls of the two segments, i.e.,
\begin{equation}
	\mathtt{e} = \left(\textit{affine.hull}\left(\left\{\textit{sign}(\varepsilon^{--}),\textit{sign}(\varepsilon^-)\right\}\right) \cap \textit{affine.hull}\left(\left\{\textit{sign}(\varepsilon^{+}),\textit{sign}(\varepsilon^{++})\right\}\right) \right)_{\varepsilon},
\end{equation}
where $\textit{affine.hull}$ denotes the affine hull.
This value is unique and lies inside the open interval $(\varepsilon^-,\varepsilon^+)$ (by Claim \ref{claim:find_hinge.1}).
Now $\mathtt{P} = \textit{sign}(\mathtt{e})$ is computed using the linear program \FindOpt, and one of three cases can occur:
\begin{enumerate}
	\item \label{PropFindBounds.case_1} $\mathtt{P}$ may be equal to the unique point of intersection
		\begin{equation}
			\textit{affine.hull}\left(\left\{\textit{sign}(\varepsilon^{--}),\textit{sign}(\varepsilon^-)\right\}\right) \cap \textit{affine.hull}\left(\left\{\textit{sign}(\varepsilon^{+}),\textit{sign}(\varepsilon^{++})\right\}\right).
		\end{equation}
		In this case, $\mathtt{e}$ is a supporting manipulability bound. 
		Furthermore, the segments 
		$\left(\textit{sign}(\varepsilon^{--}),\textit{sign}(\varepsilon^-)\right)$,
		$\left(\textit{sign}(\varepsilon^{-}),\mathtt{P}\right)$,
		$\left(\mathtt{P},\textit{sign}(\varepsilon^{+})\right)$, and 
		$\left(\textit{sign}(\varepsilon^{+}),\textit{sign}(\varepsilon^{++})\right)$.
		are all part of the signature of the Pareto frontier (by Claim \ref{claim:find_hinge.3}), and there are no other supporting manipulability bound in the interval $(\varepsilon^{--},\varepsilon^{++})$.
		\FindBounds\ marks the four segments as \emph{verified} and includes $\mathtt{e}$ in the collection of supporting manipulability bounds.
	\item \label{PropFindBounds.case_2}$\mathtt{P}$ lies in the affine hull of the segment $\left(\textit{sign}(\varepsilon^{-}),\textit{sign}(\varepsilon^+)\right)$.
		Then by Claim \ref{claim:find_hinge.2},
		\begin{equation}
			\textit{Conv}\left(\left\{\textit{sign}(\varepsilon^-),\mathtt{P}\right\}\right) \cup \textit{Conv}\left(\left\{\mathtt{P},\textit{sign}(\varepsilon^+)\right\}\right)
		\end{equation}
		is part of the signature of the Pareto frontier, and there are no supporting manipulability bounds inside the interval $\left(\varepsilon^-,\varepsilon^+\right)$.
		\FindBounds\ marks the segment $\left(\textit{sign}(\varepsilon^-),\textit{sign}(\varepsilon^+)\right)$ as verified.
	\item \label{PropFindBounds.case_3} In any other case, \FindBounds\ splits the segment $\left(\textit{sign}(\varepsilon^-),\textit{sign}(\varepsilon^+)\right)$ by creating two new unverified segments
		\begin{equation}
			\left(\textit{sign}(\varepsilon^-),\mathtt{P}\right)\text{ and } \left(\mathtt{P},\textit{sign}(\varepsilon^+)\right).
		\end{equation}
\end{enumerate}

We first show correctness of \FindBounds, then completeness:
\begin{description}
	\item[Correctness] \FindBounds\ stops running when there are no more unverified segments.
		Assume towards contradiction that there exists a supporting manipulability bound $\varepsilon \in (\underline{\varepsilon},1)$ that has not been identified.
		Then this supporting manipulability bound lies in some segment $[\varepsilon^-,\varepsilon^+]$ that was verified.
		
		If the verification happened in a case \ref{PropFindBounds.case_1}, Claim \ref{claim:find_hinge.3} ensures that there is no other supporting manipulability bound, i.e., the supporting manipulability bound would have been added to the collection during the analysis of the interval $(\varepsilon^{--},\varepsilon^{++})$.
		
		If the verification happened in a case \ref{PropFindBounds.case_2}, Claim \ref{claim:find_hinge.2} ensures that $\varepsilon \notin (\varepsilon^-,\varepsilon^+)$, so that $\varepsilon = \varepsilon^-$ (without loss of generality we can assume $\varepsilon=\varepsilon^-$, the case $\varepsilon= \varepsilon^+$ is analogous).
		The segment $\left(\textit{sign}(\varepsilon^{--}),\textit{sign}(\varepsilon^{-})\right)$ was not a verified segment at this time, otherwise this would be a case \ref{PropFindBounds.case_1}.
		Thus, at some future step some segment $\left(\textit{sign}(\tilde{\varepsilon}),\textit{sign}(\varepsilon^{-})\right)$ with a right end-point in $\textit{sign}(\varepsilon^{-})$ was verified.
		But at this step $\textit{sign}(\varepsilon^{-})$ was on the intersection of the affine hulls of the segments $\left(\textit{sign}(\tilde{\tilde{\varepsilon}}),\textit{sign}(\tilde{\varepsilon})\right)$ and $\left(\textit{sign}(\varepsilon^-),\textit{sign}(\varepsilon^{+})\right)$.
		This creates a case \ref{PropFindBounds.case_1}, and thus $\varepsilon$ was identified as a supporting manipulability bound in this step. 		
	\item[Completeness] It remains to be shown that \FindBounds\ stops at some point.
		By Claim \ref{claim:find_hinge.4}, for ever two adjacent supporting manipulability bounds $\varepsilon_k,\varepsilon_{k+1}$, \FindBounds\ computes at most three points $\textit{sign}(\varepsilon')$, $\textit{sign}(\varepsilon'')$, $\textit{sign}(\varepsilon''')$ with $\varepsilon',\varepsilon'',\varepsilon''' \in (\varepsilon_k,\varepsilon_{k+1})$.
		Since there is a finite number of supporting manipulability bounds, \FindBounds\ loops at most $3+1=4$ times per interval, which establishes completeness.
		The run-time bound follows by observing that there exist $K-2$ intervals between the smallest non-trivial bound $\varepsilon_1$ and the largest bound $\bar{\varepsilon}$. 
		However, we may need to check the interval $(\hat{\varepsilon},1)$ if $\hat{\varepsilon} <1$. 
		Thus, using $\underline{\varepsilon}$ from \FindLower\, we require at most $4K-4$ executions of \FindOpt\ to run \FindBounds. 
\end{description}

\begin{claim}
\label{claim:find_hinge.1}
$\mathtt{e} \in (\varepsilon^-,\varepsilon^+)$
\end{claim}
\begin{proof}
By convexity of $\varepsilon \mapsto \deficit(\varepsilon)$, we get that $\mathtt{e} \in [\varepsilon^-,\varepsilon^+]$.

Now suppose that $\mathtt{e} = \varepsilon^-$.
Then $\varepsilon^- \in \textit{affine.hull}\left(\left\{\textit{sign}(\varepsilon^{+}),\textit{sign}(\varepsilon^{++})\right\}\right)$.
Since $\varepsilon^- < \varepsilon^{+}$, the segments $\left(\textit{sign}(\varepsilon^-),\textit{sign}(\varepsilon^{+})\right)$ and $\left(\textit{sign}(\varepsilon^+),\textit{sign}(\varepsilon^{++})\right)$ would have been verified in a previous step.
But this is a contradiction to the assumption that $\left(\textit{sign}(\varepsilon^-),\textit{sign}(\varepsilon^{+})\right)$ was an unverified segment.
\end{proof}

\begin{claim}
\label{claim:find_hinge.2}
For $0 \leq \varepsilon_{-1} < \varepsilon_{0} < \varepsilon_1 \leq 1$, if
\begin{equation}
	\textit{sign}(\varepsilon_0) \in \textit{Conv}\left(\left\{ \textit{sign}(\varepsilon_{-1}),\textit{sign}(\varepsilon_1)  \right\}\right),
\end{equation}
then
\begin{equation}
	\textit{Conv}\left(\left\{ \textit{sign}(\varepsilon_{-1}),\textit{sign}(\varepsilon_1)  \right\}\right)
\end{equation}
is part of the signature of the Pareto frontier.
\end{claim}
\begin{proof}
Assume towards contradiction that $\textit{Conv}\left(\left\{\textit{sign}(\varepsilon_{-1}),\textit{sign}(\varepsilon_{1})\right\}\right)$ is not part of the signature of the Pareto frontier.
Then by convexity of $\varepsilon \mapsto \deficit(\varepsilon)$ there exists $\gamma \in \left(0, 1\right)$ with
\begin{equation}
\deficit(\varphi_{(1-\gamma) \varepsilon_{-1} + \gamma \varepsilon_1 }) < (1-\gamma) \deficit(\varphi_{\varepsilon_{-1}}) + \gamma \deficit(\varphi_{\varepsilon_1}).
\end{equation}
If $\varepsilon' = (1-\gamma) \varepsilon_{-1} + \gamma \varepsilon_1 > \varepsilon_0$, then for $\gamma' = \frac{\varepsilon_0 - \varepsilon_{-1}}{\varepsilon' - \varepsilon_{-1}}$ we get
\begin{equation}
(1-\gamma') \deficit(\varphi_{\varepsilon_{-1}}) + \gamma'\deficit(\varphi_{\varepsilon'}) 
	< \deficit(\varphi_{\varepsilon_0}) 
	= \deficit(\varphi_{(1-\gamma') \varepsilon_{-1} + \gamma' \varepsilon'}),
\end{equation}
a contradiction to convexity of $\varepsilon \mapsto \deficit(\varepsilon)$.
The argument for $\varepsilon' < \varepsilon_0$ is analogous. \end{proof}

\begin{claim}
\label{claim:find_hinge.3}
For $0 \leq \varepsilon_{-2} < \varepsilon_{-1} < \varepsilon_0 < \varepsilon_1 < \varepsilon_2 \leq 1$, if
\begin{equation}
\left\{\textit{sign}(\varepsilon)\right\} =
	\text{affine.hull}\left(\left\{\textit{sign}(\varepsilon_{-2}),\textit{sign}(\varepsilon_{-1})\right\}\right) \cap
	\text{affine.hull}\left(\left\{\textit{sign}(\varepsilon_1),\textit{sign}(\varepsilon_2)\right\}\right),
\end{equation}
then
\begin{equation}
	\textit{Conv}\left(\left\{\textit{sign}(\varepsilon_{0}),\textit{sign}(\varepsilon_{\pm 2})\right\}\right)
\end{equation}
is part of the signature of the Pareto frontier.
 
\end{claim}
\begin{proof}
The claim follows by applying Claim \ref{claim:find_hinge.2} twice.
\end{proof}

\begin{claim}
\label{claim:find_hinge.4}
For any two adjacent supporting manipulability bounds $\varepsilon_k,\varepsilon_{k+1}$, \FindBounds\ computes at most three points $\textit{sign}(\varepsilon')$, $\textit{sign}(\varepsilon'')$, $\textit{sign}(\varepsilon''')$ with $\varepsilon',\varepsilon'',\varepsilon''' \in (\varepsilon_k,\varepsilon_{k+1})$.
\end{claim}
\begin{proof} Suppose that $\textit{sign}(\varepsilon')$, $\textit{sign}(\varepsilon'')$, and $\textit{sign}(\varepsilon''')$ are computed in this order.
If $\varepsilon''' < \min(\varepsilon',\varepsilon'')$, then $\varepsilon'$ must be a supporting manipulability bound by convexity of $\varepsilon \mapsto \deficit(\varepsilon)$, which is a contradiction.
The same holds if $\varepsilon''' > \max(\varepsilon',\varepsilon'')$.
If $\varepsilon''' \in (\min(\varepsilon',\varepsilon''),\max(\varepsilon',\varepsilon''))$, the segment
\begin{equation}
	\left(\textit{sign}(\min(\varepsilon',\varepsilon'')),\textit{sign}(\max(\varepsilon',\varepsilon''))\right)
\end{equation}
is verified (case (\ref{PropFindBounds.case_2})).
Another guess $\varepsilon''''$ that lies within $[\varepsilon_k,\varepsilon_{k+1}]$ involve the segment $\left(\textit{sign}(\varepsilon'),\textit{sign}(\varepsilon'')\right)$.
Thus, by convexity of $\varepsilon \mapsto \deficit(\varepsilon)$, $\varepsilon''''$ is a supporting manipulability bound.
\end{proof}
This concludes the proof of Proposition \ref{PROP:RUNTIME}
}

\newcommand{\PropPluralityParetoFrontierStatement}{%
In a problem $(N,M,\deficit^{\textsc{Plu}})$ with $n = 3$ agents, $m=3$ alternatives, strict preferences, and where $\deficit^{\textsc{Plu}}$ is the worst-case $\desifct^{\textsc{Plu}}$-deficit, the following hold:
\begin{enumerate}
	\setlength{\itemsep}{0pt}
	\item The Pareto frontier has two supporting manipulability bounds $0$ and $1/3$. 
	\item Random Dictatorship is a representative of $\OPT(0)$. 
	\item Uniform Plurality is a representative of $\OPT(1/3)$.
\end{enumerate}
}
\newcommand{\PropPluralityParetoFrontierProof}{%
We first prove that Random Dictatorship is optimal at $\varepsilon_0 = 0$. 
Since Random Dictatorship is by construction a lottery of unilateral, strategyproof mechanisms, it is obviously strategyproof.
At any preference profile where all agents have the same first choice, Random Dictatorship will select this alternative and achieve zero deficit.
At any preference profile where all agents have different first choices, all alternatives have the same $\desifct^{\textsc{Plu}}$-value and therefore, any outcome has zero deficit.
Finally, consider the case where two agents agree on a first choice, $a$ say, but the third agent has a different first choice, $b$ say.
In this case, selecting $a$ would yield a maximal $\desifct^{\textsc{Plu}}$-value of $\frac{2}{3}$.
However, Random Dictatorship will only select alternative $a$ with probability $\frac{2}{3}$ and $b$ otherwise.
This leads to an outcome with $\desifct^{\textsc{Plu}}$-value of $\frac{2}{3} \cdot \frac{2}{3} + \frac{1}{3} \cdot \frac{1}{3} = \frac{5}{9}$, and therefore, the deficit of Random Dictatorship is $\frac{1}{9}$.

It remains to be shown that any strategyproof mechanism will have a deficit of at least $\frac{1}{9}$.
Following the discussion of the second takeaway from Theorem \ref{THM:HYBRID_GUARANTEES} in Section \ref{SEC:HYBRIDS}, we can restrict our attention to mechanisms $\varphi$ that are strategyproof, anonymous, and neutral. 
By Theorem \ref{THM:SYMMETRIC_DECOMP}, $\varphi$ has a symmetric decomposition, i.e., it can be represented as a lottery over neutral, strategyproof unilaterals and anonymous, strategyproof duples. 
Suppose, it contains a some anonymous, strategyproof duple $\text{dup}_{a,b}$.
By the characterization of strategyproofness via swap monotonicity, upper invariance, and lower invariance (Theorem 1 in \citep{MennleSeuken2017PSP_WP}), it follows that the outcome of $\text{dup}_{a,b}$ can only depend on the relative rankings of $a$ and $b$, so that $\text{dup}_{a,b}$ has the form
\begin{equation}
	\text{dup}_{a,b}(\bm P) = \left\{
		\begin{array}{ll}
			\left( p_3,1-p_3,0 \right), & \text{ if }P_i: a \succ b\text{ for all agents } i, \\
			\left( p_2,1-p_2,0 \right), & \text{ if }P_i: a \succ b\text{ for two agents } i, \\
			\left( p_1,1-p_1,0 \right), & \text{ if }P_i: a \succ b\text{ for one agents } i, \\
			\left( p_0,1-p_0,0 \right), & \text{ if }P_i: a \succ b\text{ for one agents } i,
	\end{array}  \right.
\end{equation}
where the vector $(p,1-p,0) = (\varphi_a(\bm P),\varphi_b(\bm P),\varphi_c(\bm P))$ denotes the outcome and $p_3 \geq p_2 \geq \frac{1}{2}$ and $p_0 \leq p_1 \leq \frac{1}{2}$.
Again by symmetry of the symmetric decomposition, it must also contain the anonymous duple $\text{dup}_{a,b}^{\varpi}$ for any permutation of the alternatives $\varpi: M \rightarrow M$.
Consider the preference profile
\begin{eqnarray}
	P_1: & & a \succ b \succ c, \\
	P_2: & & a \succ c \succ b, \\
	P_3: & & b \succ c \succ a.
\end{eqnarray}
The following table shows what outcomes of the different duples $\{\text{dup}^{\varpi}_{a,b}~|~\varpi:M \rightarrow M \text{ permutation}\}$ at this preference profile.
\begin{center}
\begin{tabular}{|c|c|c|c|}
	\hline
	\textbf{Duple} & $\bm a$ & $\bm b$ & $\bm c$ \\
	\hline \hline
	$\text{dup}_{a,b}$ & $p_2$ & $1-p_2$ & $0$  \\
	\hline
	$\text{dup}_{a,c}$ & $p_2$ & $0$ & $1-p_2$  \\
	\hline
	$\text{dup}_{b,c}$ & $0$ & $p_2$ & $1-p_2$  \\
	\hline
	$\text{dup}_{b,a}$ & $1-p_1$ & $p_1$ & $0$  \\
	\hline
	$\text{dup}_{c,a}$ & $1-p_1$ & $0$ & $p_1$  \\
	\hline
	$\text{dup}_{c,b}$ & $0$ & $1-p_1$ & $p_1$  \\
	\hline
\end{tabular}
\end{center}
A uniform lottery over these duples assigns probability $2\cdot\frac{1}{6} = \frac{1}{3}$ to alternative $b$. 
Since $p_1 \geq p_2$, this mechanism selects an outcome with $\desifct^{\textsc{Plu}}$-value at most $ \frac{2}{3} \cdot \frac{2}{3} + \frac{1}{3} \cdot \frac{1}{3} = \frac{5}{9}$. 
Since the best alternative $a$ has $\desifct^{\textsc{Plu}}$-value $\frac{2}{3}$, the mechanism must have a deficit of at least $\frac{1}{9}$ at this particular preference profile.
This is the same deficit that Random Dictatorship has at this profile, which means that including any strategyproof duples in the symmetric decomposition will not improve the deficit of $\varphi(\bm P)$.

Suppose now that the symmetric decomposition of $\varphi$ contains a neutral, strategyproof unilateral $\text{uni}_i$.
As before, it follows that $u_i$ must pick an outcome $(p_1,p_2,1-p_1-p_2)$ where $p_1 \geq p_2 \geq 1- p1-p_2$.
Again by symmetry of the symmetric decomposition, it must also contain the neutral unilateral $\text{uni}_i^{\pi}$ for any permutation of the agents $\pi: N \rightarrow N$.
Consider the same preference profile as before, with
\begin{eqnarray}
	P_1: & & a \succ b \succ c, \\
	P_2: & & a \succ c \succ b, \\
	P_3: & & b \succ c \succ a.
\end{eqnarray}
The following table shows what outcomes the different unilaterals $\{\text{uni}_i^{\pi}~|~\pi:N\rightarrow N\text{ permutation}\}$ will select.
\begin{center}
\begin{tabular}{|c|c|c|c|}
	\hline
	\textbf{Unilateral} & $\bm a$ & $\bm b$ & $\bm c$ \\
	\hline \hline
	$\text{uni}_1$ & $p_1$ & $p_2$ & $1-p_1-p_2$  \\
	\hline
	$\text{uni}_2$ & $p_1$ & $1-p_1-p_2$ & $p_2$  \\
	\hline
	$\text{uni}_3$ & $1-p_1-p_2$ & $p_1$ & $p_2$  \\
	\hline
\end{tabular}
\end{center}
A uniform lottery over these unilaterals assigns probability $\frac{1}{3}$ to alternative $b$, and consequently, this mechanism has a deficit of at least $\frac{1}{9}$ at this profile.
This is the same deficit that Random Dictatorship has at this profile, which means that including any strategyproof unilaterals in the symmetric decomposition will not improve the deficit of $\varphi(\bm P)$.

Since for all mechanisms the worst-case deficit was attained at the same preference profile, there exists no strategyproof, anonymous, neutral mechanism that has a lower deficit at $\bm P$ than $\frac{1}{9}$, and therefore, Random Dictatorship has minimal deficit among all strategyproof mechanisms.

\medskip
Next, we show that the manipulability of Uniform Plurality is minimal among all $\desifct^{\textsc{Plu}}$-maximizing mechanisms.  
Without loss of generality, we restrict attention to anonymous, neutral, and $\desifct^{\textsc{Plu}}$-maximizing mechanisms. 
At any preference profile where an alternative is ranked first by two agents or more, this alternative is implemented with certainty.
Thus, if two agents have the same first choice, the third agent has no opportunity to manipulate, because it cannot change the outcome.
The two agents with the same first choice are already receiving their favorite outcome, which makes manipulation useless for them as well.

Thus, any manipulability will arise at a preference profile where all agents have different first choices. 
Consider the preference profile $\bm P$ with
\begin{eqnarray}
	P_1 : & & a \succ b \succ c,\\
	P_2 : & & b \succ c \succ a,\\
	P_3 : & & c \succ a \succ b.
\end{eqnarray}
Renaming the alternatives is equivalent to renaming the agents.
Thus, an anonymous and neutral mechanisms has to treat all alternatives equally and must therefore select each alternative with probability $\frac{1}{3}$.
Now suppose that agent 1 is almost indifferent between $a$ and $b$, but strongly dislikes $c$, i.e., 
its utility function is ``close to'' the binary utility $u(a) = u(b) = 1, u(c) = 0$.
Then, by swapping $a$ and $b$, agent 1 can enforce the implementation $b$ with certainty. 
Its gain from this manipulation is close to
\begin{equation}
	1 \cdot u(b) - \frac{1}{3}\left(u(a) + u(b) + u(c)\right) = \frac{1}{3}.
\end{equation}
Thus, any $\desifct^{\textsc{Plu}}$-maximizing mechanism will have manipulability $\varepsilon(\varphi) \geq \frac{1}{3}$.

Now consider the Uniform Plurality mechanism.
At the above preference profile agents cannot change the outcome, unless they change their first choice.
By anonymity and neutrality, it suffices to show that agent 1 cannot do any better than $\frac{1}{3}$ by manipulating.
However, the only other possible misreport that has any effect on the outcome is to bring $c$ forward and enforce it as the outcome, which would yield no benefit for agent 1. 
Consequently, Uniform Plurality has minimal manipulability of $\frac{1}{3}$. 

\medskip
Finally, we show that there are no additional supporting manipulability bounds besides $0$ and $1/3$. 
So far we have that Random Dictatorship is in $\OPT(0)$ with a deficit of $\deficit(0) = \frac{1}{9}$.
Furthermore, Uniform Plurality is in $\OPT(1/3)$ and no $\desifct^{\textsc{Plu}}$-maximizing mechanism has strictly lower manipulability.
To complete the proof, we will show that for $\varepsilon=\frac{1}{6}$, all optimal mechanisms have deficit $\frac{1}{18}$.
By convexity of the mapping $\varepsilon\mapsto\deficit(\varepsilon)$, the signature of the Pareto frontier must therefore be a straight line between $(0,1/9)$ and $(1/3,0)$. 
Considering the performance guarantees for hybrid mechanisms (Theorem \ref{THM:HYBRID_GUARANTEES}), this implies optimality of the hybrids of Random Dictatorship and Uniform Plurality. 

Suppose that $\varphi$ is $\frac{1}{6}$-approximately strategyproof, anonymous, and neutral.
It follows from anonymity and neutrality that at the preference profile
\begin{eqnarray}
	P_1: & & a \succ b \succ c, \\
	P_2: & & b \succ c \succ a, \\
	P_3: & & c \succ a \succ b,
\end{eqnarray}
the outcome must be $\left(\frac{1}{3},\frac{1}{3},\frac{1}{3}\right)$ for $a$, $b$, $c$, respectively.
If agent $1$ changes its report to
\begin{eqnarray}
	P_1': & & b \succ a \succ c,
\end{eqnarray}
the outcome changes to $\left(\frac{1}{3}-\underline{\varepsilon}_a,\frac{1}{3}+\underline{\varepsilon}_a+\underline{\varepsilon}_c,\frac{1}{3}-\underline{\varepsilon}_c\right)$ for some values $\underline{\varepsilon}_a \leq \frac{1}{3},\underline{\varepsilon}_c \leq \frac{1}{3}$. 
Suppose now that agent 1 has a utility close to indifference between $a$ and $b$, i.e., close to $u(a) = u(b) = 1, u(c) = 0$. 
Then the utility gain for agent 1 from misreporting $P_1'$ is (arbitrarily close to) 
\begin{equation}
	1 \cdot (-\underline{\varepsilon}_a) + 1 \cdot (\underline{\varepsilon}_a + \underline{\varepsilon}_c) + 0 \cdot (-\underline{\varepsilon}_c) \approx \underline{\varepsilon}_c. 
\end{equation}
If the mechanism is $\frac{1}{6}$-approximately strategyproof, then $\underline{\varepsilon}_c \leq \frac{1}{6}$. 
The $\desifct^{\textsc{Plu}}$-value of $\varphi(\bm P')$ at $\bm P' = (P_1',P_2,P_3)$ is
\begin{eqnarray}
	\left\langle \varphi(\bm P'), \desifct^{\textsc{Plu}}(\cdot,\bm P') \right\rangle 
		& = &  \varphi_a(\bm P') \cdot 0 + \varphi_b(\bm P') \cdot \frac{2}{3} + \varphi_c(\bm P') \cdot \frac{1}{3} \\
		& = & \frac{2}{3} \left( \frac{1}{3} + \underline{\varepsilon}_a + \underline{\varepsilon}_c \right) + \frac{1}{3} \left( \frac{1}{3} - \underline{\varepsilon}_c\right) \\
		& = & \frac{6}{18} + \frac{4}{6}\underline{\varepsilon}_a + \frac{1}{3} \underline{\varepsilon}_c \\
		& \leq & \frac{6}{18} + \frac{4}{18} + \frac{1}{18} = \frac{11}{18}. 		
\end{eqnarray}
However, the only $\desifct^{\textsc{Plu}}$-maximizing alternative at $\bm P'$ is $b$ with $\desifct^{\textsc{Plu}}(b,\bm P') = \frac{2}{3} = \frac{12}{18}$. 
Thus, any $\frac{1}{6}$-approximately strategyproof mechanism must incur a deficit of at least $\frac{1}{18}$.%
}

\newcommand{\PropVetoParetoFrontierStatement}{%
In a problem $(N,M,\deficit^{\textsc{Veto}})$ with $n = 3$ agents, $m=3$ alternatives, strict preferences, and where $\deficit^{\textsc{Veto}}$ is the worst-case $\desifct^{\text{Veto}}$-deficit, the following hold:
\begin{enumerate}
	\setlength{\itemsep}{0pt}
	\item The Pareto frontier has four supporting manipulability bounds $0$, $1/21$, $1/12$, $1/2$. 
	\item Random Duple is a representative of $\OPT(0)$. 
	\item Uniform Veto is a representative of $\OPT(1/2)$.
	\item Hybrids of Random Duple and Uniform Veto not optimal for $\beta \in (0,1/2)$. 
\end{enumerate}
}
\newcommand{\PropVetoParetoFrontierProof}{
First, we prove that Random Duple is optimal at $\varepsilon_0 = 0$. 
Since it is a lottery over strategyproof duple mechanisms, it is obviously strategyproof.

At any preference profile where all agents agree on the last choice, Random Duple selects one of the other alternatives, each of which gives maximal $\desifct^{\textsc{Veto}}$-value.
At any preference profile where all agents have different last choices, any outcome has zero deficit.
Finally, consider a preference profile with
\begin{eqnarray}
	P_1,P_2: & & \ldots \succ c, \\
	P_3: & & \ldots \succ b.
\end{eqnarray}
The $\desifct^{\textsc{Veto}}$-value of $a$ is $1$ and the $\desifct^{\textsc{Veto}}$-value of $b$ is $\frac{2}{3}$, so that the maximum $\desifct^{\textsc{Veto}}$-value is 1. 
The worst case for Random Duple is that agents 1 and 2 rank $b$ first, in which case $a$ will be selected with probability $\frac{1}{3}$ and $b$ with probability $\frac{2}{3}$.
Thus, the deficit of Random Duple at $\bm P$ is 
\begin{equation}
	1 - \frac{1}{3}\cdot 1 - \frac{2}{3}\cdot \frac{2}{3} = \frac{2}{9}.
\end{equation}
In particular, this deficit is attained by Random Duple at the preference profile $\bm P$ with 
\begin{eqnarray}
	P_1,P_2: & & b \succ a \succ c, \\
	P_3: & & c \succ a \succ b. 
\end{eqnarray}

It remains to be proven whether any strategyproof, anonymous, neutral mechanism $\varphi$ can achieve a lower deficit (where anonymity and neutrality are without loss of generality, similar to the proof of Proposition \ref{PROP:PLURALITY_PARETO_FRONTIER}).
Consider the preference profile $\bm P$ and let $\text{uni}_i$ be a strategyproof and neutral unilateral component in the symmetric decomposition of $\varphi$.
$\text{uni}_i$ must pick an outcome $(p_1,p_2,1-p_1-p_2)$ with $p_1 \geq p_2 \geq 1-p_1-p_2$, where $p_k$ denotes the probability of agent $i$'s $k$th choice.
The symmetric decomposition implies that $\text{uni}_1$, $\text{uni}_2$, and $\text{uni}_3$ are equally likely to be chosen.
Analogous to the proof of Proposition \ref{PROP:PLURALITY_PARETO_FRONTIER}, we get
\begin{center}
\begin{tabular}{|c|c|c|c|}
	\hline
	\textbf{Unilateral} & $\bm a$ & $\bm b$ & $\bm c$ \\
	\hline \hline
	$\text{uni}_1$ or $\text{uni}_2$ & $p_1$ & $p_2$ & $1-p_1-p_2$  \\
	\hline
	$\text{uni}_3$ & $p_2$ & $1-p_1-p2$ & $p_1$  \\
	\hline
\end{tabular}
\end{center}
Thus, alternative $b$ is selected with probably at least $\frac{1}{3}$, which means that the deficit of $\varphi$ at the preference profile $\bm P$ is not reduced by including any unilaterals in the symmetric decomposition.

Similarly, if $\text{dup}_{a,b}$ is a strategyproof, anonymous duple in the symmetric decomposition of $\varphi$, it has the form
\begin{equation}
	\text{dup}_{a,b}(\bm P) = \left\{
		\begin{array}{ll}
			\left( p_3,1-p_3,0 \right), & \text{ if }P_i: a \succ b\text{ for all agents } i, \\
			\left( p_2,1-p_2,0 \right), & \text{ if }P_i: a \succ b\text{ for two agents } i, \\
			\left( p_1,1-p_1,0 \right), & \text{ if }P_i: a \succ b\text{ for one agents } i, \\
			\left( p_0,1-p_0,0 \right), & \text{ if }P_i: a \succ b\text{ for one agents } i,
	\end{array}  \right.
\end{equation}
where $p_3 \geq p_2 \geq \frac{1}{2}$ and $p_0 \leq p_1 \leq \frac{1}{2}$.
Again by symmetry of the symmetric decomposition, it must also contain the anonymous duple $\text{dup}_{a,b}^{\varpi}$ for any permutation of the alternatives $\varpi: M \rightarrow M$.
The following table shows what outcomes the different duples $\{\text{dup}_{a,b}^{\varpi}~|~\varpi:M\rightarrow M\text{ permutation}\}$ will select at $\bm P$.
\begin{center}
\begin{tabular}{|c|c|c|c|}
	\hline
	\textbf{Duple} & $\bm a$ & $\bm b$ & $\bm c$ \\
	\hline \hline
	$\text{dup}_{a,b}$ & $p_1$ & $1-p_1$ & $0$  \\
	\hline
	$\text{dup}_{a,c}$ & $p_2$ & $0$ & $1-p_2$  \\
	\hline
	$\text{dup}_{b,c}$ & $0$ & $p_2$ & $1-p_2$  \\
	\hline
	$\text{dup}_{b,a}$ & $1-p_2$ & $p_2$ & $0$  \\
	\hline
	$\text{dup}_{c,a}$ & $1-p_1$ & $0$ & $p_1$  \\
	\hline
	$\text{dup}_{c,b}$ & $0$ & $1-p_1$ & $p_1$  \\
	\hline
\end{tabular}
\end{center}
Thus, since $p_2 \geq p_1$, $b$ is selected with probability of at least $\frac{1}{3}$, and therefore, including any other duples in the symmetric decomposition of $\varphi$ will not improve the deficit at $\bm P$.
Consequently, there exists no strategyproof mechanism with a lower deficit than Random Duple.

\medskip
Next, we show that Uniform Veto has minimal manipulability of among all $\desifct^{\textsc{Veto}}$-maximizing mechanisms.  
First, observe that Uniform Veto is $\frac{1}{2}$-approximately strategyproof.
To see this, consider the preference profile $\bm P$
\begin{eqnarray}
	P_1: & & a \succ b \succ c, \\
	P_2,P_3: & & \ldots \succ c.
\end{eqnarray}
Uniform Veto selects $a$ and $b$ with probability $\frac{1}{2}$ each.
To manipulate, agent 1 can rank $b$ last and obtain $a$ with certainty.
Its gain from this manipulation would be
\begin{equation}
	u_1(a) - \frac{1}{2}\left(u_1(a)+u_1(b)\right) = \frac{1}{2}\left(u_1(a)-u_1(b)\right),
\end{equation}
which is at most $\frac{1}{2}$ for $u_1(a) = 1$ and $u_1(b)$ close to $0$. 
Thus, the manipulability of Uniform Veto is at least $\frac{1}{2}$. 

Suppose now that all agents have different last choices.
In this case, Uniform Veto selects any of the alternatives with probability $\frac{1}{3}$.
By ranking another alternative last, an agent could only ensure the implementation of its third choice with certainty, which is not a beneficial manipulation.

Finally, suppose that two agents have the same last choice, while a third agent has a different last choice.
We have the following cases from the perspective of agent 1.
\begin{itemize}
	\item Case I:
		\begin{eqnarray}
			P_1: & & a \succ b \succ c, \\
			P_2,P_3: & & \ldots \succ a.
		\end{eqnarray}
		In this case, Uniform Veto implements $b$ with certainty.
		Agent 1 can only enforce $c$ by ranking $b$ last, or rank $a$ last and obtain $b$ and $c$ with probabilities $\frac{1}{2}$ each.
		Neither of these moves will make agent 1 better off.
	\item Case II:
		\begin{eqnarray}
			P_1: & & a \succ b \succ c, \\
			P_2,P_3: & & \ldots \succ b.
		\end{eqnarray}
		In this case, Uniform Veto implements $a$ with certainty, which is already agent 1's first choice.
	\item Case III:
		\begin{eqnarray}
			P_1: & & a \succ b \succ c,\\
			P_2: & & \ldots \succ b, \\
			P_3: & & \ldots \succ c.
		\end{eqnarray}
		In this case, Uniform Veto implements $a$ with certainty, which is already agent 1's first choice.
	\item Case IV:
		\begin{eqnarray}
			P_1: & & a \succ b \succ c,\\
			P_2: & & \ldots \succ a, \\
			P_3: & & \ldots \succ c.
		\end{eqnarray}
		In this case, Uniform Veto implements $b$ with certainty.
		By ranking $b$ last, agent 1 could obtain a probability of $\frac{1}{3}$ for each alternative instead.
		Its gain from this manipulation is
		\begin{equation}
			\frac{1}{3}\left(u_1(a) + u_1(b) + u_1(c)\right) - u_1(b) \leq \frac{1}{3} - \frac{2}{3}u_1(b) \leq \frac{1}{3}.
		\end{equation}
\end{itemize}
Further observe that by renaming agents and alternatives, the above cases I through IV cover all possible constellations with 2 different alternatives ranked as last choices from the perspective of any agent. 
Thus, Uniform Veto is $\frac{1}{2}$-approximately strategyproof.

Having shown that the manipulability of Uniform Veto is exactly $\frac{1}{2}$, we now must show that any other $\desifct^{\textsc{Veto}}$-maximizing mechanism $\varphi$ has manipulability $\varepsilon(\varphi)\geq \frac{1}{2}$. 
We assume without loss of generality that $\varphi$ is also anonymous and neutral, and we consider the preference profile
\begin{eqnarray}
	P_1,P_2: & & a \succ b \succ c, \\
	P_3: & & b \succ a \succ c.
\end{eqnarray}
At this profile, $\varphi$ has to select $a$ with some probability $p_a$ and $b$ with probability $1-p_a$.
If $p_a \geq \frac{1}{2}$, agent 3 can rank $a$ last and enforce selection of $b$, the only remaining alternative with full $\desifct^{\textsc{Veto}}$-value.
If $u_3(b) = 1$, $u_3(c) = 0$, and $u_3(a)$ is close to 0, its gain will be
\begin{equation}
	1 - (1-p_a) = p_a \geq \frac{1}{2}.
\end{equation}
If $p_a < \frac{1}{2}$, agent 1 can enforce $a$ by ranking $b$ last and obtain a gain of
\begin{equation}
	1 - p_a > \frac{1}{2}
\end{equation}
with a similar utility function. 
Thus, any $\desifct^{\textsc{Veto}}$-maximizing mechanism has manipulability of at least $\frac{1}{2}$.

\medskip
Last, we show that the hybrids of Random Duples and Uniform Veto do not lie on the Pareto frontier, except for the extreme cases. 
In fact, the signatures of their hybrids form a straight line: 
observe that the deficit of Random Duple is attained at the preference profile $\bm P$ with
\begin{eqnarray}
	P_1,P_2,P_3: & & a \succ b \succ c.
\end{eqnarray}
Since Uniform Veto is $\desifct^{\textsc{Veto}}$-maximizing, its deficit is zero at all preference profiles.
Consequently, by linearity of the $\desifct^{\textsc{Veto}}$-value $\desifct(x,\bm P)$ in the outcome $x$, the deficit of any hybrid $h_{\beta}$ of Random Duple and Uniform Veto is determined by the deficit at $\bm P$.
Furthermore, the manipulability of Uniform Veto is highest at the same preference profile if agent 1 swaps $b$ and $c$ to enforce $a$ and has a utility close to $u_1(a) = 1$, $u_1(c)= 0$, and $u_1(b)$ close to $0$.
This misreport leaves the outcome of Random Duple unchanged.
Therefore, the manipulability of any hybrid will also be determined by this preference profile and this potential misreport.
By linearity of the incentive constraints from Theorem \ref{THM:ASP_EQUIV_LINEAR_CONSTRAINTS} it is evident that the signatures of the hybrids of Random Duple and Uniform Veto for a straight line between the signatures $(0,2/9)$ and $(1/2,0)$ of the respective component mechanism.
Consequently, if the Pareto frontier is not linear, then these hybrids will not be on the Pareto frontier for any $\beta \neq 0,1$.
\begin{table}%
\begin{tabular}{|c|c|c|l|}
	\hline
	\textbf{Algorithm} & $\bm \varepsilon$ & $\bm{\delta({\varepsilon})}$ & \textbf{Comment} \\
	\hline
	\hline
	\FindLower & $0$				 					& $2/9$ & found s.m.b. at $\varepsilon = 0$\\
	\hline
	\FindLower & $1/2$				& $0$ & \\
	\hline
	\FindLower & $1/4$				& $1/12$ & \\
	\hline
	\FindLower & $1/8$				& $1/8$ & \\
	\hline
	\FindLower & $1/16$			& $65/432$ & \\
	\hline
	\FindLower & $1/32$			& $13/72$ & \\
	\hline
	\FindLower & $1/64$			& $29/144$ & found $\underline{\varepsilon} = 1/64 < \varepsilon_1$ \\
	\hline
	\hline
	\FindBounds & $1/6$			& $1/9$ & \\
	\hline
	\FindBounds & $1/18$			& $25/162$ & \\
	\hline
	\FindBounds & $3/47$			& $28/187$ & \\
	\hline
	\FindBounds & $1/12$			& $5/36$ & found s.m.b. at $\varepsilon = 1/12$ and $\varepsilon = 1/2$ \\
	\hline
	\FindBounds & $1/21$			& $10/63$ & found s.m.b. at $\varepsilon = 1/21$\\
	\hline
\end{tabular}
\caption{Executions of the linear program \FindOpt\ when using \FindLower\ and \FindBounds\ to determine the Pareto frontier.}
\label{tbl:find_hinges_dup_veto}
\end{table}%

To find the supporting manipulability bounds of the Pareto frontier, we used the algorithm \FindLower\ to determine a lower bound for the smallest non-zero supporting manipulability bound and then applied \FindBounds\ with this value of $\underline{\varepsilon}$.
Table \ref{tbl:find_hinges_dup_veto} gives the signatures on the Pareto frontier that were determined using the \textit{signature}-function in the order in which they were computed.%

Table \ref{tbl:optimal_mechanisms} gives two mechanisms that are optimal at $\varepsilon_1=1/21$ and $\varepsilon_2=1/12$, respectively. 
For preference profiles that are not listed, rename the agents and alternatives to obtain one of the listed preference profiles, and select the respective outcome (renaming the alternatives again). 
\begin{table}
\begin{tabular}{|c|c|c||c|c|c||c|c|c|}
	\hline
	\multicolumn{3}{|c||}{\textbf{Preference Profile}} & \multicolumn{3}{|c||}{$\bm{\varepsilon_1 = 1/21}$} & \multicolumn{3}{|c|}{$\bm{\varepsilon_2= 1/12}$} \\
	\hline
	$\bm{P_1}$ & $\bm{P_2}$ & $\bm{P_3}$ & \hspace{1.2em}$\bm{a}$\hspace{1em} & \hspace{1em}$\bm{b}$\hspace{1em} & \hspace{1em}$\bm{c}$\hspace{1em} & \hspace{1em}$\bm{a}$\hspace{1em} & \hspace{1em}$\bm{b}$\hspace{1em} & \hspace{1em}$\bm{c}$\hspace{1em} \\
	\hline
	\hline
	$a \succ b \succ c$ & $ a \succ b \succ c$ & $ a \succ b \succ c$ & $  11/21 $ & $   3/7  $ & $   1/21 $ & $   7/12 $ & $   1/3  $ & $   1/12 $\\
	\hline
	$a \succ b \succ c$ & $ a \succ b \succ c$ & $ a \succ c \succ b$ & $   4/7  $ & $   3/7  $ & $ 0$ & $   2/3  $ & $   1/4  $ & $   1/12 $\\
	\hline
	$a \succ b \succ c$ & $ a \succ b \succ c$ & $ b \succ a \succ c$ & $  11/21 $ & $  10/21 $ & $ 0$ & $   1/2  $ & $   1/2  $ & $ 0$\\
	\hline
	$a \succ b \succ c$ & $ a \succ b \succ c$ & $ b \succ c \succ a$ & $  10/21 $ & $  11/21 $ & $ 0$ & $   5/12 $ & $   7/12 $ & $ 0$\\
	\hline
	$a \succ b \succ c$ & $ a \succ b \succ c$ & $ c \succ a \succ b$ & $   4/7  $ & $   8/21 $ & $   1/21 $ & $   2/3  $ & $   1/3  $ & $ 0$\\
	\hline
	$a \succ b \succ c$ & $ a \succ b \succ c$ & $ c \succ b \succ a$ & $  10/21 $ & $  11/21 $ & $ 0$ & $   5/12 $ & $   7/12 $ & $ 0$\\
	\hline
	$a \succ b \succ c$ & $ a \succ c \succ b$ & $ a \succ c \succ b$ & $   4/7  $ & $ 0$ & $   3/7  $ & $   2/3  $ & $   1/12 $ & $   1/4  $\\
	\hline
	$a \succ b \succ c$ & $ a \succ c \succ b$ & $ b \succ a \succ c$ & $   4/7  $ & $   3/7  $ & $ 0$ & $   7/12 $ & $   5/12 $ & $ 0$\\
	\hline
	$a \succ b \succ c$ & $ a \succ c \succ b$ & $ b \succ c \succ a$ & $   3/7  $ & $   8/21 $ & $   4/21 $ & $   1/3  $ & $   5/12 $ & $   1/4  $\\
	\hline
	$a \succ b \succ c$ & $ a \succ c \succ b$ & $ c \succ a \succ b$ & $   4/7  $ & $ 0$ & $   3/7  $ & $   7/12 $ & $ 0$ & $   5/12 $\\
	\hline
	$a \succ b \succ c$ & $ a \succ c \succ b$ & $ c \succ b \succ a$ & $   3/7  $ & $   4/21 $ & $   8/21 $ & $   1/3  $ & $   1/4  $ & $   5/12 $\\
	\hline
	$a \succ b \succ c$ & $ b \succ a \succ c$ & $ b \succ a \succ c$ & $  10/21 $ & $  11/21 $ & $ 0$ & $   1/2  $ & $   1/2  $ & $ 0$\\
	\hline
	$a \succ b \succ c$ & $ b \succ a \succ c$ & $ b \succ c \succ a$ & $   3/7  $ & $   4/7  $ & $ 0$ & $   5/12 $ & $   7/12 $ & $ 0$\\
	\hline
	$a \succ b \succ c$ & $ b \succ a \succ c$ & $ c \succ a \succ b$ & $   4/7  $ & $   8/21 $ & $   1/21 $ & $   7/12 $ & $   5/12 $ & $ 0$\\
	\hline
	$a \succ b \succ c$ & $ b \succ a \succ c$ & $ c \succ b \succ a$ & $   8/21 $ & $   4/7  $ & $   1/21 $ & $   5/12 $ & $   7/12 $ & $ 0$\\
	\hline
	$a \succ b \succ c$ & $ b \succ c \succ a$ & $ b \succ c \succ a$ & $   1/21 $ & $   4/7  $ & $   8/21 $ & $ 0$ & $   2/3  $ & $   1/3  $\\
	\hline
	$a \succ b \succ c$ & $ b \succ c \succ a$ & $ c \succ a \succ b$ & $   1/3  $ & $   1/3  $ & $   1/3  $ & $   1/3  $ & $   1/3  $ & $   1/3  $\\
	\hline
	$a \succ b \succ c$ & $ b \succ c \succ a$ & $ c \succ b \succ a$ & $   1/21 $ & $   4/7  $ & $   8/21 $ & $ 0$ & $   7/12 $ & $   5/12 $\\
	\hline
	$a \succ b \succ c$ & $ c \succ a \succ b$ & $ c \succ a \succ b$ & $  11/21 $ & $ 0$ & $  10/21 $ & $   7/12 $ & $ 0$ & $   5/12 $\\
	\hline
	$a \succ b \succ c$ & $ c \succ a \succ b$ & $ c \succ b \succ a$ & $   8/21 $ & $   4/21 $ & $   3/7  $ & $   5/12 $ & $   1/4  $ & $   1/3  $\\
	\hline
	$a \succ b \succ c$ & $ c \succ b \succ a$ & $ c \succ b \succ a$ & $ 0$ & $  11/21 $ & $  10/21 $ & $ 0$ & $   7/12 $ & $   5/12 $\\
	\hline
\end{tabular}
\caption{Optimal mechanisms at $\varepsilon_1 = 1/21$ and $\varepsilon_2=1/12$ (extended to other preference profiles via the anonymity and neutrality extension).}
\label{tbl:optimal_mechanisms}
\end{table}}

\begin{document}

{\setstretch{1}
\title{%
\LARGE{%
The Pareto Frontier for \\ Random Mechanisms}\thanks{\scriptsize{
Department of Informatics, University of Zurich, Switzerland,
Email: \{mennle, seuken\}@ifi.uzh.ch.
Updates available from https://www.ifi.uzh.ch/ce/publications/EFF.pdf. 
We would like to thank (in alphabetical order)
Daniel Abaecherli, 
Haris Aziz, 
Gianluca Brero, 
Albin Erlanson, 
Bettina Klaus, 
Christian Kroer, 
Dmitry Moor, 
Ariel Procaccia, 
Tuomas Sandholm, 
Arunava Sen, 
Steffen Schuldenzucker, 
William Thomson, 
and 
Utku \"Unver 
for helpful comments on this work.
Furthermore, we are thankful for the feedback we received from participants at 
COST COMSOC Meeting 2016 (Istanbul, Turkey) and multiple anonymous referees at EC'15 and EC'16. 
Any errors remain our own.
A 1-page abstract based on this paper has been published in the conference proceedings of EC'16 \citep{MennleSeuken2016EFF_EC_extended_abstract}.
Part of this research was supported by the SNSF (Swiss National Science Foundation) under grant \#156836.}}}
\author{%
Timo Mennle \\ University of Zurich
\and Sven Seuken \\ University of Zurich }
\date{First version: February 13, 2015 \\
This version: \today}
\maketitle

\begin{abstract}
We study the trade-offs between strategyproofness and other desiderata, such as efficiency or fairness, that often arise in the design of random ordinal mechanisms. 
We use $\varepsilon$-approximate strategyproofness to define \emph{manipulability}, a measure to quantify the incentive properties of non-strategyproof mechanisms, 
and we introduce the \emph{deficit}, a measure to quantify the performance of mechanisms with respect to another desideratum. 
When this desideratum is incompatible with strategyproofness, mechanisms that trade off manipulability and deficit optimally form the \emph{Pareto frontier}.
Our main contribution is a structural characterization of this Pareto frontier, and we present algorithms that exploit this structure to compute it.
To illustrate its shape, we apply our results for two different desiderata, namely Plurality and Veto scoring, in settings with $3$ alternatives and up to $18$ agents.
\end{abstract}

}

\newpage

\section{Introduction}
\label{SEC:INTRODUCTION}
In many situations, a group of individuals has to make a collective decision by selecting an alternative: 
\emph{who should be the next president? 
who gets a seat at which public school? or where to build a new stadium?} 
\emph{Mechanisms} are systematic procedures to make such decisions. 
Formally, a mechanism collects the individuals' (or \emph{agents'}) preferences and selects an alternative based on this input. 
While the goal 
is to select an alternative that is desirable for society as a whole, it may be possible for individual agents to manipulate the outcome to their own advantage by being insincere about their preferences. 
However, if these agents are lying about their preferences, the mechanism will have difficulty in determining an outcome that is desirable with respect to the true preferences. 
Therefore, incentives for truthtelling are a major concern in mechanism design. 
In this paper, we consider \emph{ordinal} mechanisms with \emph{random outcomes}. 
Such mechanisms collect preference orders and select lotteries over alternatives. 
We formulate our results for the full domain where agents can have any weak or strict preferences, but our results continue to hold on many domain restrictions, including the domain of strict preferences, the assignment domain, and the two-sided matching domain. 
\subsection{The Curse of Strategyproofness}
\label{SEC:INTRODUCTION:STRATEGYPROOFNESS}
\noindent
\emph{Strategyproof} mechanisms make truthful reporting a dominant strategy for all agents and it is therefore the ``gold standard'' among the incentive concepts. 
However, the seminal impossibility results by \citet{Gibbard1973ManipulationOfVotingSchemes} and \citet{Satterthwaite1975StrategyproofnessAndArrowsConditions} established that if there are at least three alternatives and all strict preference orders are possible, then the only unanimous and deterministic mechanisms that are strategyproof are \emph{dictatorships}. 
\citet{Gibbard1977ManipulationOfVotingWithChance} extended these insights to mechanisms that involve randomization and showed that all strategyproof random mechanisms are probability mixtures of strategyproof unilateral and strategyproof duple mechanisms. 
These results greatly restrict the design space of strategyproof mechanisms. 
In particular, many common desiderata are incompatible with strategyproofness, such as 
Condorcet consistency, 
stability of matchings, 
or 
egalitarian fairness. 

When a desideratum is incompatible with strategyproofness, designing ``ideal'' mechanisms is impossible. 
For example, suppose that our goal is to select an alternative that is the first choice of as many agents as possible. 
The \emph{Plurality} mechanism selects an alternative that is the first choice for a maximal number of agents. 
Thus, it achieves our goal perfectly. 
At the same time, Plurality and any other mechanism that achieves this goal perfectly must be manipulable. 
In contrast, the \emph{Random Dictatorship} mechanism selects the first choice of a random agent. 
It is strategyproof but there is a non-trivial probability that it selects ``wrong'' alternatives. 
If Plurality is ``too manipulable'' and Random Dictatorship is ``wrong too often,'' then trade-offs are necessary. 
In this paper, we study mechanisms that make \emph{optimal trade-offs} between (non-)manipulability and another desideratum. 
Such mechanisms form the \emph{Pareto frontier} as they achieve the desideratum as well as possible, subject to a given limit on manipulability. 
\subsection{Measuring Manipulability and Deficit}
\label{SEC:INTRODUCTION:MEASURES}
In order to understand these trade-offs formally, we need measures for the performance of mechanisms in terms of incentives and in terms of the desideratum. 
\paragraph{Approximate Strategyproofness and Manipulability}
A strategyproof mechanism does not allow agents to obtain a strictly positive gain from misreporting. 
\emph{Approximate strategyproofness} is a relaxation of strategyproofness: 
instead of requiring the gain from misreporting to be non-positive, $\varepsilon$-approximate strategyproofness imposes a small (albeit positive) upper bound $\varepsilon$ on this gain. 
The economic intuition is that if the potential gain is small, then agents might 
stick to truthful reporting (e.g., due to cognitive cost). 
To obtain a notion of approximate strategyproofness for ordinal mechanisms, 
we follow earlier work and consider agents with von Neumann-Morgenstern utility functions that are bounded between 0 and 1 \citep{BirrellPass2011ApproximatelyStrategyproofVoting,Carroll2013AQuantitativeApproachToIncentives}. 
This allows the formulation of a parametric measure for incentives: 
the \emph{manipulability} $\varepsilon(\varphi)$ of a mechanism $\varphi$ is the smallest bound $\varepsilon$ for which $\varphi$ is $\varepsilon$-approximately strategyproof.
\paragraph{Desideratum Functions and Deficit}
By allowing higher manipulability, the design space of mechanisms grows. 
This raises the question how the new freedom can be harnessed to improve the performance of mechanisms with respect to other desiderata. 
To measure this performance we introduce the \emph{deficit}. 

We first illustrate its construction via an example: 
consider the desideratum to select an alternative that is the first choice of as many agents as possible. 
For any alternative $a$ it is natural to think of ``the share of agents for whom $a$ is the first choice'' as the \emph{value} that society derives from selecting $a$.  
For some alternative ($b$, say) this value is maximal. 
Then the \emph{deficit of $a$} is the difference between the share of agents with first choice $b$ and the share of agents with first choice $a$. 
A Dictatorship mechanism may select alternatives with strictly positive deficit. 
We define the \emph{deficit} of this mechanism to be the highest deficit of any alternatives that it selects across all possible preference profiles. 

In general, any notion of deficit is constructed in three steps: 
first, we express the desideratum via a \emph{desideratum function}, which specifies the value to society from selecting a given alternative at a given preference profile.
Since we consider random mechanisms, we extend desideratum functions to random outcomes by taking expectations.%
\footnote{%
Many important desiderata can be expressed via desideratum functions, including binary desiderata such as
\emph{unanimity}, 
\emph{Condorcet consistency}, 
\emph{egalitarian fairness},
\emph{Pareto optimality}, 
\emph{$v$-rank efficiency} of assignments, 
\emph{stability} of matchings, 
any desideratum specified via a \emph{target mechanism} (or a \emph{target correspondence}), 
or any \emph{logical combination} of these. 
Moreover, it is possible to express quantitative desiderata, such as 
\emph{maximizing positional score} in voting, 
\emph{maximizing $v$-rank value} of assignments, 
or \emph{minimizing the number of blocking pairs} in matching. 
We discuss the generality and limitations of desideratum functions in Section \ref{APP:GENERAL}.} 
Second, we define the deficit of outcomes based on the desideratum function. 
Third, we define the deficit of mechanisms based on the deficit of outcomes.%
\footnote{%
In the example we considered \emph{absolute} differences to construct the deficit of outcomes and we considered the \emph{worst-case} deficit to construct the deficit of mechanisms. 
Alternatively, we can consider \emph{relative} differences and we can also define an \emph{ex-ante} deficit of mechanisms. 
Many meaningful notions of deficit can be constructed in this way (see Appendix \ref{APP:DEFICITS}); 
the results in this paper hold for any of them.%
}
For any mechanism $\varphi$, we denote its deficit by $\deficit(\varphi)$.
\subsection{Optimal Mechanisms and the Pareto Frontier}
\label{SEC:INTRODUCTION:FRONTIER}
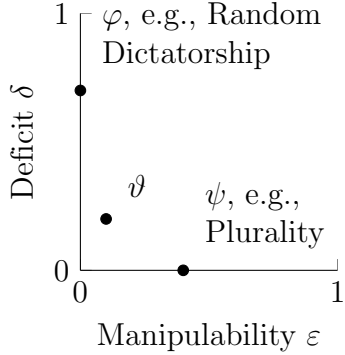
\begin{figure}
\begin{center}
\begin{tikzpicture}[scale=1] 
\begin{axis}[ 
	xlabel={Manipulability $\varepsilon$}, 
	ylabel={Deficit $\deficit$},
	ymin=0,
	ymax=1,
	xmin=0,
	xmax=1, 
	height=5cm,
	width=5cm,
	xtick={0,1},
	ytick={0,1},
	axis lines*=left, 
	clip=false,
	] 
\addplot+[
	mark=*,
	mark options={solid,black},
	only marks,
	]
	coordinates {
	(0.4,0) 
	(0.1,0.2) 
	(0,0.7) 
	}; 
	\node[label={[align=left]45:{$\varphi$, e.g., Random \\ Dictatorship}}] at (axis cs:0,.7) {};
	\node[label={[align=left]45:{$\vartheta$}}] at (axis cs:0.1,.2) {};
	\node[label={[align=left]45:{$\psi$, e.g., \\ Plurality}}] at (axis cs:.4,0) {};
\end{axis} 
\end{tikzpicture}
\end{center}
\caption{Example signatures of mechanisms in a manipulability-deficit-plot.}
\label{FIG:INTRO:EXAMPLE_SIGN}
\end{figure}
Together, manipulability $\varepsilon(\varphi)$ and deficit $\deficit(\varphi)$ yield a way to compare different mechanisms in terms of incentives and their performance with respect to some desideratum. 
Specifically, the \emph{signature} of a mechanism $\varphi$ is the point $(\varepsilon(\varphi),\deficit(\varphi))$ in the unit square $[0,1]\times [0,1]$.
Figure \ref{FIG:INTRO:EXAMPLE_SIGN} illustrates this comparison. 
Ideally, a mechanism would be strategyproof \emph{and} would always select the most desirable alternatives. 
This corresponds to a signature in the origin $(0,0)$. 
However, for desiderata that are incompatible with strategyproofness, designing ideal mechanisms is not possible. 
Instead, there exist strategyproof mechanisms which have a non-trivial deficit, such as Random Dictatorship; 
and there exist value maximizing mechanisms which have non-trivial manipulability, such as Plurality (if the goal is to select an alternative that is the first choice of as many agents as possible). 
Finally, there may exist mechanisms that have intermediate signatures, such as $\vartheta$ in Figure \ref{FIG:INTRO:EXAMPLE_SIGN}. 
Choosing between these mechanisms means to make trade-offs.
\paragraph{Finding Optimal Mechanisms}
Naturally, we want to make \emph{optimal} trade-offs. 
A mechanism is \emph{optimal at manipulability bound $\varepsilon$} if 
it has manipulability of at most $\varepsilon$ and 
it has the lowest deficit among all such mechanisms. 
For a given $\varepsilon$ we denote by $\OPT(\varepsilon)$ the set of all mechanisms that are optimal at $\varepsilon$. 
Given an optimal mechanism, it is not possible to reduce the deficit without increasing manipulability at the same time. 
In this sense, the set of all optimal mechanisms constitutes the \emph{Pareto frontier}. 
Our first result yields a finite set of linear constraints that is equivalent to $\varepsilon$-approximate strategyproofness. 
This equivalence allows us to formulate the linear program \FindOpt, whose solutions uniquely identify the optimal mechanisms at $\varepsilon$.
\paragraph{Trade-offs via Hybrid Mechanisms}
Given two mechanisms, \emph{mixing} them suggests itself as a natural approach to create new mechanisms with intermediate signatures. 
Formally, the \emph{$\beta$-hybrid} of two mechanisms $\varphi$ and $\psi$ is the $\beta$-convex combination of the two mechanisms. 
Such a hybrid can be implemented by first collecting the agents' preference reports, then randomly deciding to use $\psi$ or $\varphi$ with probabilities $\beta$ and $1-\beta$, respectively. 
If $\varphi$ has lower manipulability and $\psi$ has lower deficit, then one would expect their hybrid to inherit a share of both properties. 
Our second result in this paper formalizes this intuition: 
we prove that the signature of a $\beta$-hybrid is always \emph{weakly preferable} on both dimensions to the $\beta$-convex combination of the signatures of the two original mechanisms. 
This insight teaches us that interesting intermediate mechanisms can indeed be obtained via mixing. 

Our result has important consequences for the Pareto frontier: 
it implies that the first unit of manipulability that we sacrifice yields the greatest return in terms of a reduction of deficit. 
Furthermore, the marginal return on further relaxing incentive requirements decreases as the mechanisms become more manipulable. 
This is good news for mechanism designers because it means that the most substantial improvements already arise by relaxing strategyproofness just ``a little bit.'' 
\paragraph{Structural Characterization of the Pareto Frontier}
To fully understand the possible and necessary trade-offs between manipulability and deficit, we need to identify the whole Pareto frontier across \emph{all} manipulability bounds. 
Our main result in this paper is a structural characterization of this Pareto frontier.
We show that there exists a finite set $\varepsilon_0 < \ldots < \varepsilon_K$ of \emph{supporting manipulability bounds}, such that between any two of them ($\varepsilon_{k-1}$ and $\varepsilon_{k}$, say) the Pareto frontier consists precisely of the hybrids of two mechanisms that are optimal at $\varepsilon_{k-1}$ and $\varepsilon_{k}$, respectively. 
Consequently, the two building blocks of the Pareto frontier are, 
first, the optimal mechanisms at the supporting manipulability bounds $\varepsilon_k$ 
and, second, the hybrids of optimal mechanisms at adjacent supporting manipulability bounds for any intermediate $\varepsilon\neq \varepsilon_k$. 
Thus, the Pareto frontier can be represented concisely in terms of a finite number of optimal mechanisms and their hybrids. 
In combination with the linear program \FindOpt, we can exploit this characterization to compute the whole Pareto frontier algorithmically.

\smallskip
In summary, we provide a novel perspective on the possible and necessary trade-offs between incentives and other desiderata. 
Our results unlock the Pareto frontier of random mechanisms to analytic, axiomatic, and algorithmic exploration. 
\section{Related Work}
\label{SEC:RELATED}
Severe impossibility results restrict the design of strategyproof ordinal mechanisms. 
The seminal Gibbard-Satterthwaite Theorem \citep{Gibbard1973ManipulationOfVotingSchemes,Satterthwaite1975StrategyproofnessAndArrowsConditions} established that if all strict preferences over at least 3 alternatives are possible, then the only unanimous, strategyproof, and deterministic mechanisms are dictatorial, and \citet{Gibbard1977ManipulationOfVotingWithChance} extended this result to random mechanisms.
Thus, many important desiderata are incompatible with strategyproofness, such as selecting a Condorcet winner or maximizing Borda count \citep{Pacuit2012VotingMethodsStanfordEncyclopedia}. 
Similar restrictions persist in other domains: 
in the random assignment problem, strategyproofness is incompatible with rank efficiency \citep{Featherstone2011RankBasedRefinementWP}, 
and in the two-sided matching problem, strategyproofness is incompatible with stability \citep{Roth1982MatchingStabilityIncentives}. 

Many research efforts have been made to circumvent these impossibility results and to obtain better performance on other dimensions. 
One way to reconcile strategyproofness with other desiderata is to consider restricted domains. 
\citet{Moulin1980OnStrategyproofnessAndSinglePeakedness} showed that in the single-peaked domain, all strategyproof, anonymous, and efficient mechanisms are variants of the Median mechanism with additional virtual agents, and \citet{EhlersPetersStrocken2002SPProbDecisionSchemesForOneDimSinglePeakedPrefs} extended this result to random mechanisms. 
\citet{ChatterjiSen2013WellBehavedAdmittingDomains} showed that a \emph{semi-single-peaked} structure is \emph{almost} the defining characteristic of domains that admit the design of strategyproof deterministic mechanisms with appealing properties; an analogous result for random mechanisms is outstanding. 

An alternative way to circumvent impossibility results is to continue working in full domains but to relax the strategyproofness requirement ``a little bit.'' 
This can enable the design of mechanisms that come closer to achieving a given desideratum but still have appealing (albeit imperfect) incentive properties. 
\citet{MennleSeuken2017PSP_WP} introduced \emph{partial strategyproofness}, a relaxation of strategyproofness that has particular appeal in the assignment domain. 
\citet{AzevedoBudish2015SPL} proposed \emph{strategyproofness in the large}, which requires that the incentives for any individual agent to misreport should vanish in large markets. 
However, strategyproofness in the large is unsuited for the exact analysis of finite settings which we perform in this paper. 
Instead, we follow \citet{BirrellPass2011ApproximatelyStrategyproofVoting} and \citet{Carroll2013AQuantitativeApproachToIncentives}, who used \emph{approximate strategyproofness for agents with bounded vNM utility functions} to quantify manipulability of non-strategyproof ordinal mechanisms and derived asymptotic results. 
We also use approximate strategyproofness but we give exact results for finite settings. 

Some prior work has considered trade-offs explicitly. 
Using efficiency notions based on dominance relations, \citet{AzizBrandBrill2013OnTradeoffEffAndSPInRandSocialChoice} and \citet{AzizBrandlBrandt2014IncompatibilityOfEffAndSPInRandSocialChoice} proved compatibility and incompatibility of various combinations of incentive and efficiency requirements. 
\citet{Procaccia2010ApproxCircumventGS} considered an \emph{approximation ratio} based on positional scoring and gave bounds on how well strategyproof random mechanisms can approximate optimal positional scores as markets get large. 
While he found most of these to be inapproximable by strategyproof mechanisms, \citet{BirrellPass2011ApproximatelyStrategyproofVoting} obtained positive limit results for the approximation of deterministic target mechanisms via \emph{approximately} strategyproof random mechanisms. 
In \citep{MennleSeuken2017HYB_WP} we studied how hybrid mechanisms make trade-offs between startegyproofness (in terms of the \emph{degree of strategyproofness}) and efficiency (in terms of \emph{dominance}) in random assignment.  
While hybrid mechanisms also play a central role in our present paper, we consider general ordinal mechanisms and different measures (i.e., \emph{manipulability} and \emph{deficit}). 
\section{Formal Model}
\label{SEC:MODEL}
Let $N$ be a set of $n$ \emph{agents} and $M$ be a set of $m$ \emph{alternatives}, where the tuple $(N,M)$ is called a \emph{setting}. 
Each agent $i \in N$ has a \emph{preference order} $P_i$ over alternatives, where 
$P_i: a \succeq b$, 
$P_i: a \succ b$, and
$P_i: a \sim b$ 
denote \emph{weak} preference, \emph{strict} preference, and \emph{indifference}, respectively, 
and $\mathcal{P}$ denotes the set of all preference orders. 
For agent $i$'s preference order $P_i$, the \emph{rank} of alternative $j$ under $P_i$ is the number of alternatives that $i$ strictly prefers to $j$ plus 1, denoted $\text{rank}_{P_i}(j)$.\footnote{1 is added to ensure that first choice alternatives have rank 1, not 0.} 
A \emph{preference profile} $\bm P = (P_i,P_{-i})$ is a collection of preference orders from all agents, and $P_{-i}$ are the preference orders of all other agents, except $i$.
A \emph{(random) mechanism} is a mapping 
$\varphi: \mathcal{P}^N \rightarrow \Delta (M)$. 
Here $\Delta (M)$ is the space of lotteries over alternatives, and any $x \in \Delta (M)$ is called an \emph{outcome}.

We extend agents' preferences over alternatives to preferences over lotteries via von Neumann-Morgenstern utility functions:
each agent $i\in N$ has a \emph{utility function} $u_i:M\rightarrow [0,1]$ that represents their preference order, i.e., $u_i(a) \geq u_i(b)$ holds whenever $P_i: a \succeq b$. 
Note that utility functions are bounded between 0 and 1, so that the model admits a non-degenerate notion of approximate strategyproofness (see Remark \ref{REM:NORMALIZE}).
We denote the set of all utility functions that represent the preference order $P_i$ by $U_{P_i}$. 
\begin{remark}
\label{REMARK:STRICT_PREFERENCES} 
We formulate our results for the full domain but they naturally extend to a variety of domains, including 
the domain of strict preferences, 
the assignment domain, 
and the two-sided matching domain (see Appendix \ref{APP:GENERAL}).
\end{remark}
\section{Approximate Strategyproofness and Manipulability}
\label{SC:SP}
Our goal in this paper is to study mechanisms that trade off manipulability and other desiderata optimally. 
For this purpose we need measures for the performance of different mechanisms with respect to the two dimensions of this trade-off.
In this section, we review approximate strategyproofness, derive a measure for the incentive properties of non-strategyproof mechanisms, and present our first main result.
\subsection{Strategyproofness and Approximate Strategyproofness}
\label{SEC:SP:SP_ASP}
The most demanding incentive concept is strategyproofness. 
It requires that truthful reporting is a dominant strategy for all agents. 
For random mechanisms, this means that truthful reporting always maximizes any agent's \emph{expected utility}.  
\begin{definition}[Strategyproofness]
\label{DEF:STRATEGYPROOF} 
Given a setting $(N,M)$, 
a mechanism $\varphi$ is \emph{strategyproof} 
if for any agent $i \in N$, 
any preference profile $(P_i, P_{-i}) \in \mathcal{P}^N$, 
any utility $u_i \in U_{P_i}$ that represents $P_i$, 
and any misreport $P_i' \in \mathcal{P}$, we have
\begin{equation}
	\sum_{j \in M} u_i(j) \cdot \left( \varphi_j(P_i',P_{-i}) - \varphi_j(P_i,P_{-i}) \right) \leq 0.
\label{EQ:SP_BOUND_INEQUALITY}
\end{equation}
\end{definition}
The left side of (\ref{EQ:SP_BOUND_INEQUALITY}) is the change in its own expected utility that $i$ can affect by falsely reporting $P_i'$ instead of reporting $P_i$. 
For later use, we denote this difference by
\begin{equation}
	\varepsilon(u_i,(P_i,P_{-i}),P_i',\varphi) = \sum_{j \in M} u_i(j) \cdot \left( \varphi_j(P_i',P_{-i}) - \varphi_j(P_i,P_{-i}) \right).
\label{EQ:MANIPULABILITY_IN_SITUATION}
\end{equation}
The fact that $\varepsilon(u_i,(P_i,P_{-i}),P_i',\varphi)$ is upper bounded by 0 for strategyproof mechanisms means that deviating from the true preference report weakly decreases expected utility for any agent in any situation, independent of the other agents' reports. 
Conversely, if a mechanism $\varphi$ is \emph{not} strategyproof, there necessarily exists at least one situation in which $\varepsilon(u_i,(P_i,P_{-i}),P_i',\varphi)$ is strictly positive. 
Imposing a different bound leads to the notion of approximate strategyproofness \citep{BirrellPass2011ApproximatelyStrategyproofVoting,Carroll2013AQuantitativeApproachToIncentives}. 
\begin{definition}[$\varepsilon$-Approximate Strategyproofness]
\label{DEF:ASP}
Given a setting $(N,M)$ and a bound $\varepsilon \in [0,1]$, 
a mechanism $\varphi$ is \emph{$\varepsilon$-approximately strategyproof} 
if for any agent $i \in N$, 
any preference profile $(P_i, P_{-i}) \in \mathcal{P}^N$, 
any utility $u_i \in U_{P_i}$ that represents $P_i$, 
and any misreport $P_i' \in \mathcal{P}$, we have
\begin{equation}
	\varepsilon(u_i,(P_i,P_{-i}),P_i',\varphi) = \sum_{j \in M} u_i(j) \cdot \left( \varphi_j(P_i',P_{-i}) - \varphi_j(P_i,P_{-i}) \right) \leq \varepsilon.
\label{EQ:INCENTIVE_CONSTRAINTS_APPROXIMATE_STRATEGYPROOFNESS}
\end{equation}
\end{definition}
This definition is analogous to Definition \ref{DEF:STRATEGYPROOF} of strategyproofness, except that the upper bound in (\ref{EQ:INCENTIVE_CONSTRAINTS_APPROXIMATE_STRATEGYPROOFNESS}) is $\varepsilon$ instead of 0. 
Thus, $0$-approximate strategyproofness coincides with strategyproofness.
Furthermore, the gain never exceeds 1, which makes $1$-approximate strategyproofness a void constraint that is trivially satisfied by any mechanism.

The interpretation of intermediate values of $\varepsilon \in (0,1)$ is more challenging. 
Unlike utilities in quasi-linear domains, vNM utilities are not comparable across agents. 
Thus, we cannot simply think of $\varepsilon$ as the ``value'' (e.g., in dollars) that an agent can gain by misreporting. 
Instead, $\varepsilon$ is a \emph{relative} bound: 
since $u_i$ is between 0 and 1, a change of magnitude $1$ in expected utility corresponds to the selection of an agent's first choice instead of that agent's last choice. 
Thus, ``1'' is the \emph{maximal gain} from misreporting that \emph{any} agent could obtain under an \emph{arbitrary} mechanism. 
The bound $\varepsilon$ is the \emph{share of this maximal gain} by which any agent can at most improve its expected utility under an $\varepsilon$-approximately strategyproof mechanism. 
\begin{remark}
\label{REM:NORMALIZE}
The bounds on utilities are essential for $\varepsilon$-approximate strategyproofness to be a useful relaxation of strategyproofness for ordinal mechanisms.  
Suppose that a non-strategyproof mechanism $\varphi$ allows a gain of $\varepsilon(u_i,(P_i,P_{-i}),P_i',\varphi) > 0$. 
Then scaling the utility function $u_i$ by a factor $\alpha > 1$ results in a linear increase of this gain. 
Thus, $\varepsilon$-approximate strategyproofness for unbounded utilities would imply strategyproofness.%
\end{remark}
\subsection{Manipulability}
\label{SEC:SP:MANIPULABILITY}
If $\varphi$ is $\varepsilon$-approximately strategyproof, then it is also $\varepsilon'$-approximately strategyproof for any $\varepsilon' \geq \varepsilon$. 
Thus, lower values of $\varepsilon$ correspond to stronger incentives. 
With this in mind, we define the manipulability of a mechanism. 
\begin{definition}[Manipulability]
\label{DEF:MANIPULABILITY}
Given a setting $(N,M)$ and mechanism $\varphi$, the \emph{manipulability of $\varphi$} (\emph{in the setting $(N,M)$}) is given by 
\begin{equation}
	\varepsilon(\varphi) = \min\{\varepsilon' \in [0,1] : \varphi \text{ is }\varepsilon'\text{-approximately strategyproof in }(N,M)\}.
\label{EQ:DEF_MANIPULABILITY}
\end{equation}
\end{definition}
Intuitively, $\varepsilon(\varphi)$ is the lowest bound $\varepsilon'$ for which $\varphi$ is $\varepsilon'$-approximately strategyproof. 
This minimum is in fact attained because all inequalities from (\ref{EQ:INCENTIVE_CONSTRAINTS_APPROXIMATE_STRATEGYPROOFNESS}) are weak. 
Note that in a different setting $(N',M')$, the manipulability may vary. 
However, for all statements in this paper, a setting is held fixed and the value $\varepsilon(\varphi)$ should be understood as the manipulability of the mechanism $\varphi$ in the fixed setting from the respective context. 
\subsection{An Equivalent Condition for Approximate Strategyproofness}
\label{SEC:SP:EQUIVALENT}
Recall that the definition of $\varepsilon$-approximate strategyproofness imposes an upper bound on the gain that agents can obtain by misreporting. 
In particular, inequality (\ref{EQ:INCENTIVE_CONSTRAINTS_APPROXIMATE_STRATEGYPROOFNESS}) must hold for all utility functions that represent the agent's preference order. 
Since there are infinitely many such utility functions, a na\"ive approach to verifying $\varepsilon$-approximate strategyproofness of a given mechanism would involve checking an infinite number of constraints. 
This is somewhat unattractive from an axiomatic perspective and even prohibitive from an algorithmic perspective. 
Fortunately, we can bypass this issue, as the next Theorem \ref{THM:ASP_EQUIV_LINEAR_CONSTRAINTS} shows. 
\begin{theorem}
\label{THM:ASP_EQUIV_LINEAR_CONSTRAINTS}
\ThmASPEquivLinConStatement
\end{theorem}
\begin{proof}[Proof Outline (formal proof in Appendix \ref{APP:PROOFS:ASP_EQUIV_LINEAR_CONSTRAINTS})]
The key idea is to represent any utility function as an element of the convex hull of a certain set of \emph{extreme utility functions}.
For any combination $(i,(P_i,P_{-i}),P_i',k)$ the inequality in statement (2) is just the $\varepsilon$-approximate strategyproofness constraints for one extreme utility function. 
\end{proof}
Theorem \ref{THM:ASP_EQUIV_LINEAR_CONSTRAINTS} yields that $\varepsilon$-approximate strategyproofness can be equivalently expressed as a \emph{finite} set of \emph{weak}, \emph{linear} inequalities. 
This has far-reaching consequences. 
In general, it unlocks approximate strategyproofness for use under the automated mechanism design paradigm \citep{Sandholm2003AutomatedMechanismDesign}. 
Specifically, it enables our identification of optimal mechanisms as solutions to a particular linear program (Section \ref{SEC:OPTIMAL}). 
\section{Desideratum Functions and Deficit}
\label{SEC:VALUE}
While it is important to elicit truthful preferences, good incentives alone do not make a mechanism attractive. 
Instead, it should ultimately select \emph{desirable} alternatives, where desirability depends on the agents' preferences. 
In this section, we introduce a formal method to quantify the ability of mechanisms to achieve a given desideratum. 
\subsection{Desideratum Functions}
\label{SEC:VALUE:DESIDERATUM_FCT}
To express a desideratum formally, we define \emph{desideratum functions}.
These reflect the value that society derives from selecting a particular alternative when the agents have a particular preference profile. 
\begin{definition}[Desideratum Function]
\label{DEF:DESIDERATUM_FCT}
A \emph{desideratum function} is a mapping \hbox{$\desifct:M\times \mathcal{P}^N \rightarrow [0,1]$}, where $\desifct(j,\bm P)$ is the \emph{$\desifct$-value} associated with selecting alternative $j$ when the agents have preferences $\bm P$. 
\end{definition}
We illustrate how desideratum functions represent desiderata with two examples. 
\begin{example}
\label{EX:DEFICIT_PLURALITY}
Suppose that our goal is to select an alternative that is the first choice of as many agents as possible. 
We can define the corresponding desideratum function by setting $\desifct^{\textsc{Plu}}(j,\bm P)=n_j^1/n$, where $n_j^1$ is the number of agents whose first choice under $\bm P$ is $j$. 
Note that $\desifct^{\textsc{Plu}}(j,\bm P)$ is proportional to the \emph{Plurality score} of $j$ under $\bm P$.
\end{example}
\begin{example}
\label{EX:DEFICIT_CONDORCET}
Alternatively, we can consider a binary desideratum. 
An alternative $j$ is a \emph{Condorcet winner at $\bm P$} if it dominates all other alternatives in a pairwise majority comparison.  
A mechanism is \emph{Condorcet consistent} if it selects Condorcet winners whenever they exist. 
We can reflect this desideratum by setting $\desifct^{\textsc{Con}}(j,\bm P)=1$ for any $j$ that is a Condorcet winner at $\bm P$, and $\desifct^{\textsc{Con}}(j,\bm P)=0$ otherwise. 
\end{example}
Desideratum functions are extended to (random) outcomes by taking expectations. 
\begin{definition}[Expected $\desifct$-value]
\label{DEF:EXPECTED_VALUE}
Given a desideratum function $\desifct$, a preference profile $\bm P \in \mathcal{P}^N$, and an outcome 
$x \in \Delta(M)$, 
the \emph{(expected) $\desifct$-value of $x$ at $\bm P$} is given by 
\begin{equation}
	\desifct(x,\bm P) = \sum_{j \in M} x_j \cdot \desifct(j,\bm P).
\end{equation}
\end{definition}
The interpretation of $d(x,\bm P)$ is straightforward: 
if $\desifct$ \emph{quantifies} the value of alternatives (as in Example \ref{EX:DEFICIT_PLURALITY}), then $\desifct(x,\bm P)$ is the \emph{expectation} of the societal value from selecting an alternative according to $x$. 
If $\desifct$ reflects a \emph{binary desideratum} (as in Example \ref{EX:DEFICIT_CONDORCET}), then $\desifct(x,\bm P)$ is the \emph{probability} of selecting an alternative with the desirable property. 
\begin{remark}
\label{REM:RESTRICTIONS}
By taking expectations, the $\desifct$-value of random outcomes is fully determined by the $\desifct$-values of the alternatives. 
This linear structure is a key ingredient to our results. 
In Appendix \ref{APP:GENERAL}, we show that many (but not all) popular desiderata admit such a representation, and we also discuss the limitations. 
\end{remark}
Ideally, mechanisms would always select outcomes that maximize the $d$-value. 
\begin{definition}[$\desifct$-maximizing]
\label{DEF:MAXIMIZING}
Given a desideratum function $\desifct$ and a preference profile $\bm P \in \mathcal{P}^N$, an outcome $x \in \Delta(M)$ is \emph{$\desifct$-maximizing at $\bm P$} if 
$\desifct(x,\bm P) = \max_{j\in M} \desifct(j,\bm P)$. 
A mechanism $\varphi$ is \emph{$\desifct$-maximizing} if, for any $\bm P \in \mathcal{P}^N$, $\varphi(\bm P)$ is $\desifct$-maximizing at $\bm P$.  
\end{definition}
Note that for any preference profile $\bm P$ there always exists at least one alternative that is $\desifct$-maximizing at $\bm P$ by construction. 
Furthermore, any $\desifct$-maximizing random outcome must be a lottery over alternatives that are all $\desifct$-maximizing at $\bm P$. 

\medskip
\textbf{Example \ref{EX:DEFICIT_PLURALITY}, continued.} 
Recall the desideratum function $\desifct^{\textsc{Plu}}(j,\bm P) = n_j^1/n$, where $n_j^1$ is the number of agents who ranked $j$ first under $\bm P$. 
A mechanism is $\desifct^{\textsc{Plu}}$-maximizing if and only if it is a Plurality mechanism (i.e., a mechanism that selects only alternatives which are the first choice of a maximum number of agents). 
\medskip

\textbf{Example \ref{EX:DEFICIT_CONDORCET}, continued.}
Recall that 
$ \desifct^{\textsc{Con}}(j,\bm P) = \mathds{1}_{\{j\text{ Condorcet winners at }\bm P\}}$ 
expresses the desideratum to select Condorcet winners when they exist. 
Indeed, any $\desifct^{\textsc{Con}}$-maximizing mechanism is Condorcet consistent. 
Moreover, at any preference profile $\bm P$ where no Condorcet winner exists, the minimal and maximal achievable $\desifct$-values are both zero. 
At these $\bm P$, the $\desifct^{\textsc{Con}}$-maximizing mechanisms are therefore free to choose any outcome.  
Consequently, maximizing $\desifct^{\textsc{Con}}$-value is \emph{equivalent} to Condorcet consistency. 
\subsection{Deficit of Mechanisms}
\label{SEC:VALUE:MEASURES_MECHANISMS}
Intuitively, the deficit of an outcome is the loss that society incurs from choosing that outcome instead of a $d$-maximizing outcome. 
\begin{definition}[Deficit of Outcomes]
\label{DEF:DEFICIT}
Given a desideratum function $\desifct$, 
a preference profile $\bm P\in \mathcal{P}^N$, 
and an outcome $x \in \Delta(M)$, 
the \emph{$\desifct$-deficit of $x$ at $\bm P$} is 
\begin{equation}
	\deficit_\desifct(x,\bm P) = \max_{j \in M}\desifct(j,\bm P) - \desifct(x,\bm P).
\label{EQ:ABSOLUTE_DEFICIT_DIFFERENCE}
\end{equation}
\end{definition}
\begin{remark}[Relative Deficit]
\label{REM:RELATIVE_DEFICIT}
The difference in (\ref{EQ:ABSOLUTE_DEFICIT_DIFFERENCE}) is \emph{absolute}; 
however, in some situations, it may be more natural to consider \emph{relative} differences, such as the ratio between the achieved and the maximal achievable $\desifct$-value. 
As we show in Appendix \ref{APP:DEFICITS:RELATIVE}, it is without loss of generality that we restrict our attention to absolute differences.  
\end{remark}
Equipped with the deficit of outcomes, we define the deficit of mechanisms. 
This measure is the ``desideratum counterpart'' to the measure $\varepsilon(\varphi)$ for incentive properties.
\begin{definition}[Deficit of Mechanisms]
\label{DEF:WORST_CASE_DEFICIT}
Given a setting $(N,M)$, a desideratum function $\desifct$, and a mechanism $\varphi$, the \emph{(worst-case) $\desifct$-deficit of $\varphi$} (\emph{in $(N,M)$}) is the highest $\desifct$-deficit incurred by $\varphi$ across all preference profiles; 
formally, 
\begin{equation}
	\deficit_{\desifct}^{\max}(\varphi) = \max_{\bm P \in \mathcal{P}^N} \deficit_\desifct(\varphi(\bm P),\bm P).
\end{equation}
\end{definition}
Intuitively, the deficit of the mechanism is determined by the most severe violation of the desideratum across all preference profiles. 
Thus, a mechanism with low deficit violates the desideratum only ``a little bit'' across all possible preference profiles. 
\begin{remark} 
Another way to define the deficit of mechanisms arises when agents' preference profiles are drawn from a known distribution. 
In Appendix \ref{APP:DEFICITS:EXANTE} we define this \emph{ex-ante deficit} formally. 
All results in this paper hold for both notions of deficit. 
\end{remark}
For the remainder of this paper, we fix 
an arbitrary setting $(N,M)$, 
a desideratum expressed via a desideratum function $\desifct$, 
and a notion of deficit $\deficit$ derived from $\desifct$. 
The triple \NMO\ is called a \emph{problem}. 
Unless explicitly stated otherwise, our results are understood to hold for any fixed problem \NMO\ separately. 
\section{Optimal Mechanisms}
\label{SEC:OPTIMAL}

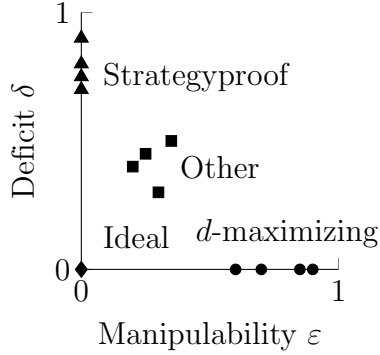
\begin{figure}
\begin{center}
\begin{tikzpicture}[scale=1] 
\begin{axis}[ 
	xlabel={Manipulability $\varepsilon$}, 
	ylabel={Deficit $\deficit$},
	ymin=0,
	ymax=1,
	xmin=0,
	xmax=1, 
	height=5cm,
	width=5cm,
	xtick={0,1},
	ytick={0,1},
	axis lines*=left, 
	clip=false,
	] 
\addplot+[
	mark size=3.0pt,
	mark=triangle*,
	mark options={solid,black},
	only marks,
	]
	coordinates {
	(0,.9) 
	(0,.8) 
	(0,.75)
	(0,.7) 
	}; 
\addplot+[
	mark size=2.0pt,
	mark=*,
	mark options={solid,black},
	only marks,
	]
	coordinates {
	(0.7,0) 
	(0.9,0) 
	(0.85,0)
	(0.6,0) 
	}; 
\addplot+[
	mark size=2.0pt,
	mark=square*,
	mark options={solid,black},
	only marks,
	]
	coordinates {
	(0.3,0.3) 
	(0.2,0.4) 
	(0.25,0.45)
	(0.35,0.5)
	}; 
\addplot+[
	mark size=3.0pt,
	mark=diamond*,
	mark options={solid,black},
	only marks,
	]
	coordinates {
	(0,0) 
	}; 
	\node[label={0:{Strategyproof}}] at (axis cs:0,.75) {};
	\node[label={45:{Ideal}}] at (axis cs:0,0) {};
	\node[label={90:{$d$-maximizing}}] at (axis cs:.8,0) {};
	\node[label={0:{Other}}] at (axis cs:.3,.4) {};
\end{axis} 
\end{tikzpicture}
\end{center}
\caption{Manipulability-deficit-plot with mechanism signatures: strategyproof (triangles), $\desifct$-maximizing (circles), ideal (diamond), others (squares).}
\label{FIG:EXAMPLE_SIGN}
\end{figure}

We are now in a position to formalize and study optimal mechanisms, which are mechanisms that trade off manipulability and deficit optimally. 
\subsection{Signatures of Mechanisms}
\label{SEC:OPTIMAL:SIGNATURE}
To compare different mechanisms, we introduce signatures. 
\begin{definition}[Signature]
\label{DEF:SIGNATURE}
Given a problem \NMO\ and a mechanism $\varphi$, the tuple $\left(\varepsilon(\varphi),\deficit(\varphi)\right) \in [0,1] \times [0,1]$ is called the \emph{signature of $\varphi$} (\emph{in the problem \NMO}).\footnote{Since we fix a problem \NMO, we refer to the tuple $(\varepsilon(\varphi),\deficit(\varphi))$ simply as the \emph{signature of $\varphi$}, keeping in mind that a mechanism's signature may be different for different problems. }
\end{definition}
Signatures allow a convenient graphical representation of the performance of any mechanism in terms of manipulability and deficit.
Figure \ref{FIG:EXAMPLE_SIGN} gives examples of such signatures: 
since $0$-approximate strategyproofness is equivalent to strategyproofness, the signature of any strategyproof mechanism  have an $\varepsilon$-component of 0. 
On the other hand, any $\desifct$-maximizing mechanism have a signature with a $\deficit$-component of 0. 
If an ideal mechanism exists, it has a signature in the origin $(0,0)$. 
Mechanisms that are neither strategyproof nor $d$-maximizing have signatures in the half-open unit square $(0,1]\times (0,1]$. 
\subsection{Definition and Existence of Optimal Mechanisms}
\label{SEC:OPTIMAL:DEFINITION}
\noindent
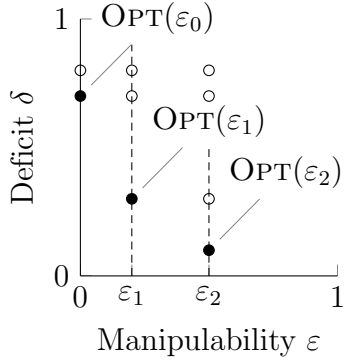
\begin{figure}
\begin{center}
\begin{tikzpicture}[scale=1] 
\begin{axis}[ 
	xlabel={Manipulability $\varepsilon$}, 
	ylabel={Deficit $\deficit$},
	ymin=0,
	ymax=1,
	xmin=0,
	xmax=1, 
	height=5cm,
	width=5cm,
	xtick={0,.2,.5,1},
	xticklabels={$0$,$\varepsilon_1$,$\varepsilon_2$,$1$},
	ytick={0,1},
	axis lines*=left, 
	clip=false,
	] 
\draw[densely dashed] (axis cs:.2,0) -- (axis cs:.2,.9);
\draw[densely dashed] (axis cs:.5,0) -- (axis cs:.5,.5);
\addplot+[
	mark size=2.0pt,
	mark=*,
	mark options={solid,black},
	only marks,
	]
	coordinates {
	(0,0.7) 
	(0.2,.3)
	(0.5,.1)
	}; 
\addplot+[
	mark size=2.0pt,
	mark=o,
	mark options={black},
	only marks,
	]
	coordinates {
	(0.2,.7) 
	(0.5,.7)
	(0.5,.3)
	(0,.8)
	(0.2,.8)
	(0.5,.8)
	}; 
	\node[pin={80:{$\OPT(\varepsilon_0)$}}] at (axis cs:0,.7) {};
	\node[pin={80:{$\OPT(\varepsilon_1)$}}] at (axis cs:0.2,.3) {};
	\node[pin={80:{$\OPT(\varepsilon_2)$}}] at (axis cs:0.5,.1) {};
\end{axis} 
\end{tikzpicture}
\end{center}
\caption{Example signatures of optimal mechanisms (filled circles) and non-optimal mechanisms (empty circles) at manipulability bounds $\varepsilon_0,\varepsilon_1,\varepsilon_2$.}
\label{FIG:EXAMPLE_OPTIMAL}
\end{figure}

When impossibility results prohibit the design of ideal mechanisms, the decision in favor of any mechanism necessarily involves a trade-off between manipulability and deficit. 
To make an informed decision about this trade-off, a mechanism designer must be aware of the different design options. 
One straightforward approach to this problem is to decide on a maximal acceptable manipulability $\varepsilon \in [0,1]$ up front and use a mechanism that minimizes the deficit among all $\varepsilon$-approximately strategyproof mechanisms. 
We define 
$\deficit(\varepsilon) = \min \left\{ \deficit(\varphi)\ |\ \varphi\ \varepsilon\text{-approximately SP} \right\}$ 
to be the lowest deficit that is achievable by any $\varepsilon$-approximately strategyproof mechanism. 
\begin{definition}[Optimal Mechanism]
\label{DEF:OPTIMAL_MECHANISM}
Given a problem \NMO\ and a bound $\varepsilon \in [0,1]$, a mechanism $\varphi$ is \emph{optimal at $\varepsilon$} if $\varphi$ is $\varepsilon$-approximately strategyproof and $\deficit(\varphi) = \deficit(\varepsilon)$.
Denote by $\OPT(\varepsilon)$ the set of all mechanisms that are optimal at $\varepsilon$. 
Any optimal mechanism $\varphi \in \OPT(\varepsilon)$ is called a \emph{representative of $\OPT(\varepsilon)$}.%
\end{definition}
Proposition \ref{PROP:OPTIMAL_NON_EMPTY} shows that optimal mechanisms always exist. 
\begin{proposition}
\label{PROP:OPTIMAL_NON_EMPTY}
\PropOptimalMechanismsNotEmptyStatement
\end{proposition}
Existence follows via a compactness argument (see Appendix \ref{APP:PROOFS:OPTIMAL_NON_EMPTY} for a proof).

Proposition \ref{PROP:OPTIMAL_NON_EMPTY} justifies the use of the minimum (rather than the infimum) in the definition of $\deficit(\varepsilon)$, since the deficit $\deficit(\varepsilon) = \deficit(\varphi)$ is in fact attained by some mechanism.
Figure \ref{FIG:EXAMPLE_OPTIMAL} illustrates signatures of optimal and non-optimal mechanisms. 
On the vertical lines at each of the manipulability bounds $\varepsilon_0 = 0, \varepsilon_1, \varepsilon_2$, the 
circles (empty circles) correspond to signatures of non-optimal mechanisms. 
The signatures of optimal mechanisms (filled circles) from $\OPT(\varepsilon_k), k=0,1,2$ take the lowest positions on the vertical lines. 
\subsection{Identifying Optimal Mechanisms}
\label{SEC:OPTIMAL:COMPUTE}
The existence proof for optimal mechanisms is implicit and does not provide a way of actually determining them. 
Our next result establishes a correspondence between the set $\OPT(\varepsilon)$ and the set of solutions to a linear program. 
We can solve this program algorithmically to find representatives of $\OPT(\varepsilon)$ and to compute $\deficit(\varepsilon)$.
\begin{lp}[\FindOpt]
\label{LP:FINDOPT}
\renewcommand\arraystretch{1.4}
\upshape
\begin{equation*}
\begin{array}{llr}
	\text{minimize } & \variable{d} & \text{\emph{(Objective)}} \\
	\text{subject to } & \displaystyle\sum\limits_{j \in M: \text{rank}_{P_i}(j) \leq k} \variable{f_j(P_i',P_{-i})} - \variable{f_j(P_i,P_{-i})} \leq \varepsilon,  & \text{\emph{($\varepsilon$-approximate SP)}} \\
	& \hspace{2em}\forall i \in N, (P_i,P_{-i}) \in \mathcal{P}^N, P_i' \in \mathcal{P}, k \in \{1,\ldots,m\}  &
\\

	& \variable{d} \geq \max_{j\in M}\desifct(j,\bm P) - \sum_{j \in M} \variable{f_j(\bm P)} \cdot \desifct(j,\bm P), \hspace{1em} \forall \bm P \in \mathcal{P}^N & \text{\emph{(Deficit)}} \\	
	& \displaystyle\sum\limits_{j \in M} \variable{f_j(\bm P)} = 1, \hspace{3em} \forall \bm P \in \mathcal{P}^N & \text{\emph{(Probability)}} \\
	&	\variable{f_j(\bm P)} \in [0,1], \hspace{3em} \forall \bm P \in \mathcal{P}^N, j \in M & \text{\emph{(Outcome variables)}} \\
	&	\variable{d} \in [0,1], & \text{\emph{(Deficit variable)}} 
\end{array}
\end{equation*}
\end{lp}
Each variable $\variable{f_j(\bm P)}$ corresponds to the probability with which the mechanism $\varphi$ selects alternative $j$ if the agents report preference profile $\bm P$.
Consequently, any assignment of the variables $\{\variable{f_j(\bm P)} : j \in M, \bm P \in \mathcal{P}^N\}$ corresponds to a unique mapping $\varphi: M \times \mathcal{P}^N \rightarrow \mathds{R}^M$. 
The two constraints labeled \emph{(Probability)} and \emph{(Outcome variables)} ensure that the variable assignment in fact corresponds to a \emph{mechanism} (rather than just a mapping from preference profiles to $\mathds{R}^M$). 
The additional variable $\variable{d}$ represents this mechanism's deficit and the objective is to minimize its value. 
 
The constraints labeled \emph{($\varepsilon$-approximate SP)} reflect the equivalent condition for $\varepsilon$-approximate strategyproofness that we obtained from Theorem \ref{THM:ASP_EQUIV_LINEAR_CONSTRAINTS}. 
In combination, these constraints ensure that the mechanisms corresponding to the feasible variable assignments of \FindOpt\ are \emph{exactly} the $\varepsilon$-approximately strategyproof mechanisms. 
The constraint labeled \emph{(Deficit)} 
makes $\variable{d}$ an upper bound for the deficit of $\varphi$.\footnote{Alternatively, we can implement the \emph{ex-ante deficit} by replacing the \emph{(Deficit)}-constraint (see Appendix \ref{APP:DEFICITS:EXANTE}).}

The following Proposition \ref{PROP:OPT_EQUIVALENT_LP} formalizes the correspondence between optimal mechanisms and  solutions of the linear program \FindOpt. 
\begin{proposition}
\label{PROP:OPT_EQUIVALENT_LP}
Given a problem \NMO\ and a bound $\varepsilon \in [0,1]$, a variable assignment $\{\variable{f_j(\bm P)} : j \in M, \bm P \in \mathcal{P}^N\}$ is a solution of \FindOpt\ at $\varepsilon$ if and only if the mechanism $\varphi$ defined by $\varphi_j(\bm P) = \variable{f_j(\bm P)}$ for all $j\in M, \bm P \in \mathcal{P}^N$ is optimal at $\varepsilon$. 
\end{proposition}
The proof follows directly from the discussion above.
One important consequence of Proposition \ref{PROP:OPT_EQUIVALENT_LP} is that we can compute optimal mechanisms for any given problem \NMO\ and any manipulability bound $\varepsilon \in [0,1]$. 
Going back to the mechanism designer's problem of trading off manipulability and deficit, we now have a way of determining optimal mechanisms for particular manipulability bounds $\varepsilon$. 
With \FindOpt\ we can evaluate algorithmically what deficit we must accept when manipulability must not exceed some fixed limit $\varepsilon$. 

Shifting the burden of design to a computer by encoding good mechanisms in 
optimization problems is the central idea of \emph{automated mechanism design} \citep{Sandholm2003AutomatedMechanismDesign}. 
A common challenge with this approach is that 
the optimization problem can become large and difficult to solve; 
and na\"ive implementations of \FindOpt\ will face this issue as well. 
Substantial run-time improvements are possible, e.g., by exploiting additional axioms such as anonymity and neutrality \citep{MennleAbaecherliSeuken2015ComputingFrontiers}. 
Nonetheless, determining optimal mechanisms remains a computationally expensive operation. 

Computational considerations aside, Proposition \ref{PROP:OPT_EQUIVALENT_LP} provides a new understanding of optimal mechanisms: 
since $\OPT(\varepsilon)$ corresponds to the solutions of the linear program \FindOpt, it can be interpreted as a convex polytope. 
In Section \ref{SEC:FRONTIER} we use methods from convex geometry to derive our structural characterization of the Pareto frontier. 
The representation of optimal mechanisms as solutions to \FindOpt\ constitutes the first building block of this characterization. 
\section{Hybrid Mechanisms}
\label{SEC:HYBRIDS}
In this section, we introduce \emph{hybrid mechanisms}, which are convex combinations of two \emph{component mechanisms}. 
Intuitively, by mixing one mechanism with low manipulability and another mechanism with low deficit, we may hope to obtain new mechanisms with intermediate signatures. 
Initially, the construction of hybrids is independent of the study of optimal mechanisms. 
However, in Section \ref{SEC:FRONTIER}, they will constitute the second building block for our structural characterization of the Pareto frontier. 
\begin{definition}[Hybrid]
\label{DEF:HYBRIDS}
For $\beta \in [0,1]$ and mechanisms $\varphi$ and $\psi$, the \emph{$\beta$-hybrid} $h_{\beta}$ is given by
$h_{\beta}(\bm P) = (1-\beta) \varphi(\bm P) + \beta \psi(\bm P) $ 
for any preference profile $\bm P \in \mathcal{P}^N$. 
\end{definition}
In practice, ``running'' a hybrid mechanism is straightforward: first, collect the preference reports. 
Second, toss a $\beta$-coin to determine whether to use $\psi$ (probability $\beta$) or $\varphi$ (probability $1-\beta$). 
Third, apply this mechanism to the reported preference profile. 
Our next result formalizes the intuition that hybrids have at least intermediate signatures. 
\begin{theorem}
\label{THM:HYBRID_GUARANTEES}
\ThmHybridGuaranteesStatement
\end{theorem}
\begin{proof}[Proof Outline (formal proof in Appendix \ref{APP:PROOFS:HYBRID_GUARANTEES})]
We write out the definitions of $\varepsilon(h_{\beta})$ and $\deficit(h_{\beta})$, each of which may involve taking a maximum.
The two inequalities are then obtained with the help of the triangle inequality. 
\end{proof}
In words, the signatures of $\beta$-hybrids are always weakly better than the $\beta$-convex combination of the signatures of the two component mechanisms. 

There are two important takeaways from Theorem \ref{THM:HYBRID_GUARANTEES}.
First, it yields a strong argument in favor of randomization: 
given two mechanisms with attractive manipulability and deficit, randomizing between the two always yields a mechanism with a signature that is at least as attractive as the $\beta$-convex combination of the signatures of the two mechanisms. 
As Example \ref{EX:STRICT_IMPROVEMENTS_HYBRIDS} in Appendix \ref{APP:HYBRIDS_STRICT_IMPROVEMENT} shows, randomizing in this way can yield \emph{strictly} preferable mechanisms that have \emph{strictly} lower manipulability and \emph{strictly} lower deficit than either of the component mechanisms. 

The second takeaway is that the common fairness requirement of \emph{anonymity} comes ``for free'' in terms of manipulability and deficit (provided that the deficit measure $\deficit$ is itself anonymous): 
given any mechanism $\varphi$, an anonymous mechanism can be constructed by randomly assigning the agents to new roles. 
This yields a hybrid mechanism with many components, each of which corresponds to the original mechanism with agents assigned to different roles. 
Under an anonymous deficit notion, every component will have the same signature as $\varphi$. 
If follows from Theorem \ref{THM:HYBRID_GUARANTEES} that this new anonymous mechanism has a weakly better signature than $\varphi$. 
Similarly, we can impose \emph{neutrality} without having to accept higher manipulability or more deficit (if $\deficit$ is itself neutral). 
We formalize these insights in Appendix \ref{APP:ANON_NEUTR_SYMM_DECOMP}. 
\section{The Pareto Frontier}
\label{SEC:FRONTIER}
Recall that optimal mechanisms are those mechanisms that trade off manipulability and deficit optimally. 
They form the Pareto frontier because we cannot achieve a strictly lower deficit without accepting strictly higher manipulability. 
\begin{definition}[Pareto Frontier]
\label{DEF:PARETO_FRONTIER}
Given a problem \NMO, let $\bar{\varepsilon}$ be the smallest manipulability bound that is compatible with $\desifct$-maximization; 
formally, 
\begin{equation}
	\bar{\varepsilon} = \min\{\varepsilon \in [0,1]~|~\exists \varphi: \varphi\text{ $\desifct$-maximizing \& }\varepsilon\text{-approximately SP}\}.
\end{equation}
The \emph{Pareto frontier} is the union of all mechanisms that are optimal for some manipulability bound $\varepsilon \in [0,\bar{\varepsilon}]$; 
formally, 
\begin{equation}
	\PF = \bigcup_{\varepsilon \in [0,\bar{\varepsilon}]} \OPT(\varepsilon).
\end{equation}
\end{definition}
The special manipulability bound $\bar\varepsilon$ is chosen such that maximal $\desifct$-value can be achieved with an $\bar\varepsilon$-approximately strategyproof mechanisms ($\bar\varphi$, say) but not with any mechanism that has strictly lower manipulability. 
Since $\bar\varphi$ has deficit 0, any mechanism $\varphi$ that is optimal at some larger manipulability bound $\varepsilon > \bar\varepsilon$ may be more manipulable than $\bar\varphi$ but $\varphi$ cannot have a strictly lower deficit. 
Thus, we can restrict attention to manipulability bounds between 0 and $\bar\varepsilon$ (instead of 0 and 1). 
From the perspective of the mechanism designer, mechanisms on the Pareto frontier are the only mechanisms that should be considered; 
if a mechanism is \emph{not} on the Pareto frontier, we can find another mechanism that is a Pareto-improvement in the sense that it has strictly lower manipulability, strictly lower deficit, or both. 
\subsection{Monotonicity and Convexity}
\label{SEC:FRONTIER:MONOTONIC_CONVEX}
Recall that we have defined $\deficit(\varepsilon)$ 
as the smallest deficit that can be attained by any $\varepsilon$-approximately strategyproof mechanism. 
Thus, the signatures of mechanisms on the Pareto frontier are described by the mapping $\varepsilon \mapsto \deficit(\varepsilon)$ that associates each manipulability bound $\varepsilon \in [0,\bar{\varepsilon}]$ with the deficit $\delta(\varepsilon)$ of optimal mechanisms at this manipulability bound.
Based on our results so far, we can already make interesting observations about the Pareto frontier by analyzing this mapping. 
\begin{corollary}
\label{COR:PF_MONOTONIC_CONVEX}
\CorMonotonicConvexStatement
\end{corollary}
Monotonicity follows from the definition of optimal mechanisms, and convexity is a consequence of Theorem \ref{THM:HYBRID_GUARANTEES} (see Appendix \ref{APP:PROOFS:PF_MONOTONIC_CONVEX} for a formal proof).
While monotonicity is merely reassuring, 
convexity is non-trivial and very important.
It means that when we relax strategyproofness, the first unit of manipulability that we give up allows the largest reduction of deficit.
For any additional unit of manipulability that we sacrifice, the deficit will be reduced at a lower rate, which means \emph{decreasing marginal returns}. 
Thus, we can expect to capture most gains in $\desifct$-value from relaxing strategyproofness early on. 
Moreover, convexity and monotonicity together imply \emph{continuity}. 
This means that trade-offs along the Pareto frontier are \emph{smooth} in the sense that a tiny reduction of the manipulability bound $\varepsilon$ does not require accepting a vastly higher deficit.  

For mechanism designers, these observations deliver an important lesson: 
the most substantial gains (per unit of manipulability) arise from relaxing strategyproofness just ``a little bit.'' 
This provides encouragement to investigate the gains from accepting even small amounts of manipulability. 
On the other hand, if the initial gains are not worth the sacrifice, then gains from accepting more manipulability will not be ``surprisingly'' attractive either. 
\subsection{A Structural Characterization of the Pareto Frontier}
\label{SEC:FRONTIER:STRUCTURE}
In Section \ref{SEC:OPTIMAL}, we have shown that we can identify optimal mechanisms for individual manipulability bounds 
by solving the linear program \FindOpt. 
In Section \ref{SEC:HYBRIDS}, we have introduced hybrids, and we have shown how mixing two mechanisms results in new mechanisms with intermediate or even superior signatures. 
We now give our main result, a structural characterization of the Pareto frontier in terms of these two building blocks, namely \emph{optimal mechanisms} and \emph{hybrids}. 
\begin{theorem}
\label{THM:FRONTIER_STRUCTURE}
\ThmFrontierStructureStatement
\end{theorem}
\begin{proof}[Proof Outline (formal proof in Appendix \ref{APP:PROOFS:FRONTIER_STRUCTURE})]
Our proof exploits that $\OPT(\varepsilon)$ corresponds to the solutions of the linear program \FindOpt\ (Section \ref{SEC:OPTIMAL:COMPUTE}) with feasible set $F_{\varepsilon} = \left\{x \left| Dx\leq d, Ax \leq \varepsilon \right.\right\}$, where neither $D$, nor $d$, nor $A$ depend on $\varepsilon$.
First, we show that if a set of constraints is binding for $F_{\varepsilon}$, then it is binding for $F_{\varepsilon'}$ for $\varepsilon'$ within a compact interval $[\varepsilon^-,\varepsilon^+]$ that contains $\varepsilon$ and not binding for any $\varepsilon'' \notin [\varepsilon^-,\varepsilon^+]$. 
With finiteness of the number of constraints of the LP, this yields a finite segmentation of $[0,\bar{\varepsilon}]$.
The vertex-representations \citep{Gruenbaum2003ConvexPolytopesTEXTBOOK} can then be used to show that on each segment $[\varepsilon_{k-1},\varepsilon_k]$, the \emph{solution sets} $ S_{\varepsilon} = \argmin_{F_{\varepsilon}} \variable{d}$ are exactly the $\beta$-convex combinations of $S_{\varepsilon_{k-1}}$ and $S_{\varepsilon_{k}}$ with $\beta = \frac{\varepsilon-\varepsilon_{k-1}}{\varepsilon_{k}-\varepsilon_{k-1}}$.
\end{proof}

\begin{figure}
\begin{center}
\begin{tikzpicture}[scale=.9]
\begin{axis}[ 
	xlabel={Manipulability $\varepsilon$}, 
	ylabel={Deficit $\deficit^{\textsc{Plu}}$},
	ymin=0,
	ymax=1,
	xmin=0,
	xmax=1, 
	height=5cm,
	width=5cm,
	xtick={0,0.2,0.5,1},
	xticklabels={$0$,$\varepsilon_1$,$\varepsilon_2$,$1$},
	ytick={0,1},
	axis lines*=left, 
	clip=false,
	] 
\draw[densely dashed] (axis cs:.35,0) -- (axis cs:.35,.3);
\draw[densely dashed] (axis cs:.5,0) -- (axis cs:.5,.2);
\draw[densely dashed] (axis cs:.2,0) -- (axis cs:.2,.4);
\addplot+[
	mark size=2.0pt,
	mark=*,
	mark options={solid,black},
	only marks,
	]
	coordinates {
	(0.35,0.2) 
	}; 
\addplot+[
	mark size=2.0pt,
	mark=o,
	mark options={black},
	black,
	]
	coordinates {
	(0,.7) 
	(.2,.3)
	(.5,.1)
	(.9,0)
	}; 
	\node[pin={87:{$\varphi_1\in\OPT(\varepsilon_1)$}}] at (axis cs:0.2,.3) {};
	\node[pin={40:{$h_{\beta}\in\OPT((1-\beta)\varepsilon_1 + \beta \varepsilon_2)$}}] at (axis cs:0.35,.2) {};
	\node[pin={5:{$\varphi_2\in\OPT(\varepsilon_2)$}}] at (axis cs:0.5,.1) {};
\end{axis} 
\end{tikzpicture}
\end{center}
\caption{Illustrative example of the signatures of the Pareto frontier.}
\label{FIG:ILLUSTRATIVE_FRONTIER}
\end{figure}
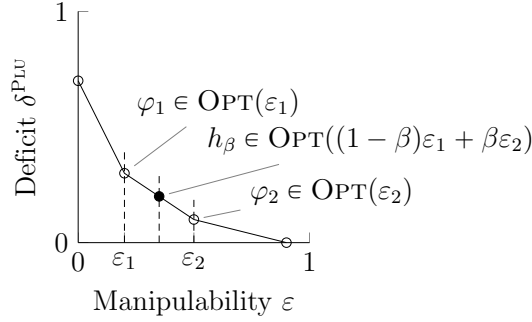

It would be particularly simple if the optimal mechanisms at some manipulability bound $\varepsilon$ 
were just the $\beta$-hybrids of optimal mechanisms at $0$ and $\bar{\varepsilon}$ with $\beta = \varepsilon/\bar{\varepsilon}$. 
While this is not true in general, Theorem \ref{THM:FRONTIER_STRUCTURE} shows that the Pareto frontier has a linear structure over each interval $[\varepsilon_{k-1},\varepsilon_k]$. 
Thus, it is completely identified by two building blocks: 
(1) the \emph{sets of optimal mechanisms} $\OPT(\varepsilon_k)$ for finitely many $\varepsilon_k, k =0,\ldots,K$, 
and 
(2) \emph{hybrid mechanisms}, which provide the missing link for $\varepsilon \neq \varepsilon_k$. 
Representatives of $\OPT(\varepsilon_k)$ can be obtained by solving the linear program \FindOpt\ at the supporting manipulability bounds $\varepsilon_k$; 
and for intermediate bounds $\varepsilon = (1-\beta)\varepsilon_{k-1} + \beta \varepsilon_k$, a mechanism $\varphi$ is optimal at $\varepsilon$ \emph{if and only if} it is a $\beta$-hybrid of two mechanisms $\varphi_{k-1} \in \OPT(\varepsilon_{k-1})$ and $\varphi_k \in \OPT(\varepsilon_k)$.

Theorem \ref{THM:FRONTIER_STRUCTURE} allows an additional insight about the mapping $\varepsilon \mapsto \deficit(\varepsilon)$ (where $\deficit(\varepsilon) = \min\{\deficit(\varphi)~|~\varphi~\varepsilon\text{-approximately SP}\}$). 
We have already observed that this mapping is monotonic, decreasing, convex, and continuous (see Corollary \ref{COR:PF_MONOTONIC_CONVEX}). 
In addition, we obtain \emph{piecewise linearity}.
\begin{corollary}
\label{COR:FRONTIER_PIECEWISE_LINEAR}
Given a problem \NMO, $	\varepsilon \mapsto \deficit(\varepsilon)$ is piecewise linear. 
\end{corollary}
Figure \ref{FIG:ILLUSTRATIVE_FRONTIER} illustrates the results of Theorem \ref{THM:FRONTIER_STRUCTURE} and Corollary \ref{COR:FRONTIER_PIECEWISE_LINEAR}: 
first, the mapping $\varepsilon \mapsto \deficit(\varepsilon)$ is monotonic, decreasing, convex, continuous, and piecewise linear, and, second, optimal mechanisms with intermediate manipulability can be obtained by mixing two mechanisms that are optimal at the two adjacent supporting manipulability bounds, respectively. 
\section{Computing the Pareto Frontier: Algorithms and Examples}
\label{SEC:COMPUTE}
In this section we discuss an algorithm that exploits the structural characterization from Theorem \ref{THM:FRONTIER_STRUCTURE} to compute all supporting manipulability bounds. 
Then we derive the Pareto frontier for two example problems. 
\subsection{Algorithm \FindBounds}
\label{SEC:COMPUTE:ALGORITHM}
Recall that the linear program \FindOpt\ can be used to determine a representative of the set of optimal mechanisms $\OPT(\varepsilon)$. 
One na\"ive approach to identifying the Pareto frontier would be to run \FindOpt\ for many different manipulability bounds to obtain optimal mechanisms at each of these bounds, and then consider these mechanisms and their hybrids. 
However, this method has two drawbacks: 
first, and most importantly, it would not yield the correct Pareto frontier. 
The result can, at best, be viewed as a \emph{conservative estimate}, since choosing fixed manipulability bounds is not guaranteed to identify any actual supporting manipulability bounds exactly. 
Second, from a computational perspective, executing \FindOpt\ is expensive, which is why we would like to keep the number of executions as low as possible. 

Knowing the structure of the Pareto frontier, we can use the information obtained from iterated applications of \FindOpt\ to interpolate the most promising candidates for supporting manipulability bounds. 
In Appendix \ref{APP:ALGORITHM} we provide the algorithms \FindBounds\ (Algorithm \ref{ALG:FIND_SB}) and \FindLower\ (Algorithm \ref{ALG:FIND_DELTA}) that do this. 
Proposition \ref{PROP:RUNTIME} summarizes their properties. 
\begin{proposition}
\label{PROP:RUNTIME}
\PropComputableStatement
\end{proposition}
Due to space constraints we delegate the detailed discussion of the algorithms and the proof of Proposition \ref{PROP:RUNTIME} to Appendix \ref{APP:ALGORITHM}. 
\subsection{Examples: Plurality Scoring and Veto Scoring}
\label{SEC:COMPUTE:ANALYTIC_EXAMPLES}
In this section, we consider two concrete problems and derive the respective Pareto frontiers.  
The two examples highlight the different shapes that the Pareto frontier can take. 
Both settings contain 3 alternatives and 3 agents with strict preferences over the alternatives, but they differ in the desideratum.
\paragraph{Plurality and Random Dictatorships}
Consider a setting with agents $N = \{1,2,3\}$, alternatives $M = \{a,b,c\}$, and where agents' preferences are strict.
Suppose further that the desideratum is to select alternatives that are the first choice of as many agents as possible. 
From Example \ref{EX:DEFICIT_PLURALITY} we know that the corresponding desideratum function is $\desifct^{\textsc{Plu}}(j,\bm P) = n_j^1/n$, where $n_j^1$ is the number of agents whose first choice under $\bm P$ is $j$. 
Finally, suppose that the deficit $\deficit^{\textsc{Plu}}$ is the worst-case $\desifct^{\textsc{Plu}}$-deficit; 
formally, 
\begin{equation}
	\deficit^{\textsc{Plu}}(\varphi) = \max_{\bm P \in \mathcal{P}^N} \left( \max_{j \in M} \left(\desifct^{\textsc{Plu}}(j,\bm P)\right) - \desifct^{\textsc{Plu}}(\varphi(\bm P),\bm P) \right).
\end{equation}
Using \FindBounds, we find that in this problem, the Pareto frontier has only two supporting manipulability bounds $\varepsilon_0 = 0$ and $\bar{\varepsilon} = \varepsilon_1 = 1/3$. 
Thus, it consists precisely of the hybrids of mechanisms that are optimal at $\varepsilon_0$ and $\varepsilon_1$ by Theorem \ref{THM:FRONTIER_STRUCTURE}. 
Moreover, we can use \FindOpt\ to obtain representatives of $\OPT(0)$ and $\OPT(1/3)$. 
Interestingly, the representatives determined by the algorithm have a familiar structure: 
first, the representative of $\OPT(0)$ corresponds to Random Dictatorship with a uniform choice of the dictator. 
Thus, in this problem, no strategyproof mechanism has lower deficit than Random Dictatorship. 
Second, the representative of $\OPT(1/3)$ corresponds to \emph{Uniform Plurality}, a mechanism that determines the set of $\desifct^{\textsc{Plu}}$-maximizing alternatives and selects one of these uniformly at random. 
Thus, in this problem, no $\desifct$-maximizing mechanism has lower manipulability than Uniform Plurality. 
Figure \ref{FIG:FRONTIER_PLOT_PLURALITY_VETO_SMALL} (a) depicts the signatures of the Pareto frontier. 

\begin{figure}%
\begin{center}
\begin{tikzpicture}[scale=1] 
\begin{groupplot}[group style={group name=my plots, group size = 2 by 1, horizontal sep=3cm, }]
	\nextgroupplot[
		xlabel={Manipulability $\varepsilon$}, 
		ylabel={Deficit $\deficit^{\textsc{Plu}}$},
		ymin=0,
		ymax=0.222222,
		xmin=0,
		xmax=1, 
		height=5cm,
		width=5cm,
		xtick={0,0.3333333333333333,1},
		xticklabels={$0$,$1/3$,$1$},
		ytick={0,.11111111111,.2222222222222222},
		yticklabels={$0$,$1/9$,$2/9$},
		axis lines*=left, 
		clip=false,
		]
		\addplot+[
			mark size=2.0pt,
			mark=o,
			mark options={black},
			black,
			]
			coordinates {
			(0,0.1111111) 
			(0.3333333,0)
		}; 
		\node[pin={[align=left]45:{$\OPT(0)$, e.g., Random \\ Dictatorship}}] at (axis cs:0,0.222222) {};
		\node[pin={[align=left]45:{$\OPT(1/3)$, e.g., \\ Uniform Duple}}] at (axis cs:0.3333333,0) {};	
	\nextgroupplot[ 
		xlabel={Manipulability $\varepsilon$}, 
		ylabel={Deficit $\deficit^{\textsc{Veto}}$},
		ymin=0,
		ymax=0.2222222,
		xmin=0,
		xmax=1, 
		height=5cm,
		width=5cm,
		xtick={0,0.0476190476190476,0.08333333333333,1},
		xticklabels={$0$,,,$1$},
		ytick={0,.11111111111,.2222222222222222},
		yticklabels={$0$,$1/9$,$2/9$},
		axis lines*=left, 
		clip=false,
		] 
	\draw[densely dotted] (axis cs:0,0.2222222) -- (axis cs:0.5,0);
	\addplot[
		mark size=2.0pt,
		mark=o,
		mark options={black},
		black,
		]
		coordinates {
		(0,0.222222) 
		(0.04761904761,0.15873015873)
		(0.08333333333,0.13888888888)
		(.5,0)
		}; 
		\node[pin={[align=left]0:{$\OPT(0)$, e.g., Random Duple}}] at (axis cs:0,0.222222) {};
		\node[pin={[align=left]0:{$\OPT(1/21)$ with $\deficit(1/21) = 10/63$}}] at (axis cs:0.04761904761,.15873015873) {};
		\node[pin={[align=left]0:{$\OPT(1/12)$ with $\deficit(1/12) = 5/36$}}] at (axis cs:0.08333333333,0.1388888888) {};
		\node[pin={[align=left]10:{$\OPT(1/3)$, e.g., Uniform Veto}}] at (axis cs:.5,0) {};
\end{groupplot}
\node[below = 1.0cm of my plots c1r1.south] {(a)};
\node[below = 1.0cm of my plots c2r1.south] {(b)};
\end{tikzpicture} 
\end{center}
\caption{Plot of signatures of the Pareto frontier for $n=3$ agents, $m=3$ alternatives, strict preferences, (a) worst-case $\desifct^{\textsc{Plu}}$-deficit $\deficit^{\textsc{Plu}}$ and (b) worst-case $\desifct^{\textsc{Veto}}$-deficit $\deficit^{\textsc{Veto}}$.}%
\label{FIG:FRONTIER_PLOT_PLURALITY_VETO_SMALL}%
\end{figure}
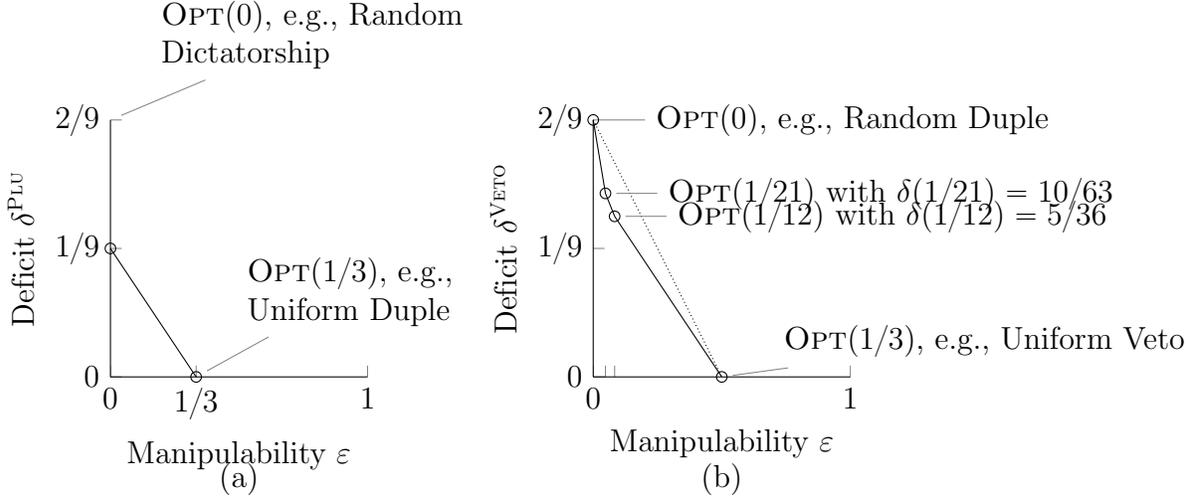

In addition, we prove these insights analytically. 
\begin{proposition}
\label{PROP:PLURALITY_PARETO_FRONTIER}
\PropPluralityParetoFrontierStatement
\end{proposition}
The formal proof is given in Appendix \ref{APP:PROOFS:PLURALITY_PARETO_FRONTIER}.
\paragraph{Veto and Random Duple}
The second problem again involves a setting with three agents, three alternatives, and agents' preferences are strict.
The difference to the previous problem is the different desideratum: 
this time, our goal is to select an alternative that is the \emph{last} choice of as \emph{few} agents as possible. 
This desideratum is reflected by the desideratum function
$\desifct^{\textsc{Veto}}(j,\bm P) = (n - n_j^m)/n$, 
where $n_j^m$ is the number of agents whose last choice under $\bm P$ was $j$. 
The worst-case $\desifct^{\textsc{Veto}}$-deficit $\deficit^{\textsc{Veto}}$ is given by 
\begin{equation}
	\deficit^{\textsc{Veto}}(\varphi) = \max_{\bm P \in \mathcal{P}^N} \left( \max_{j \in M} \left(\desifct^{\textsc{Veto}}(j,\bm P)\right) - \desifct^{\textsc{Veto}}(\varphi(\bm P),\bm P) \right).
\end{equation}
Again, we use the algorithm \FindBounds\ to determine the supporting manipulability bounds of the Pareto frontier in this problem. 
These are 
$\varepsilon_0 = 0, \varepsilon_1 = 1/21, \varepsilon_2 = 1/12, \bar\varepsilon=\varepsilon_3 = 1/2$. 
As in the previous example, we can compute representatives of the optimal mechanisms at each of these bounds. 
For the extreme bounds $\varepsilon_0=0$ and $\bar\varepsilon=1/2$, these representatives again have a familiar structure: 
first, the representative of $\OPT(0)$ corresponds to \emph{Random Duple}, a mechanism that picks two alternatives uniformly at random and then selects the one that is preferred to the other by a majority of the agents (breaking ties randomly). 
Second, the representative of $\OPT(1/2)$ corresponds to \emph{Uniform Veto}, a mechanism that finds all alternatives that are the last choice of a minimal number of agents and selects one of these alternatives uniformly at random. 
Figure \ref{FIG:FRONTIER_PLOT_PLURALITY_VETO_SMALL} (b) depicts the signatures of the Pareto frontier in this problem. 

We also created representatives of $\OPT(1/21)$ and $\OPT(1/12)$ via \FindOpt\ (they are given in Appendix \ref{APP:PROOFS:VETO_PARETO_FRONTIER} in their numerical form). 
While the interpretation of these two mechanisms is not straightforward, it is clear that neither of them is a hybrid Random Duple and Uniform Veto. 
Indeed, in order to generate mechanisms on the Pareto frontier in this problem, we cannot simply consider hybrids of optimal mechanisms from the extreme supporting manipulability bounds. 
Instead, we can (and have to) exploit the additional design freedom in the particular problem by separately designing optimal mechanisms for the two intermediate supporting manipulability bounds $\varepsilon_1=1/21$ and $\varepsilon_2=1/12$ specifically. 

As in the previous example, we convince ourselves of the correctness of these assertions by proving them analytically.
\begin{proposition}
\label{PROP:VETO_PARETO_FRONTIER}
\PropVetoParetoFrontierStatement
\end{proposition}
The formal proof is given in Appendix \ref{APP:PROOFS:VETO_PARETO_FRONTIER}.
\begin{remark} 
The fact that we obtain Random Dictatorship and Random Duple as the optimal strategyproof mechanisms in the two problems, respectively, highlights a connection between our results and the characterization of strategyproof random mechanisms by \citet{Gibbard1977ManipulationOfVotingWithChance}. 
We discuss this connection in Appendix \ref{APP:ANON_NEUTR_SYMM_DECOMP} and we present a new symmetric decomposition result that may be of independent interest. 
\end{remark}
\paragraph{Pareto Frontiers in Larger Settings}
For larger settings analytic insights about the Pareto frontier are presently not available. 
However, we can apply the algorithm \FindBounds\ to identify the Pareto frontier algorithmically.\footnote{This can be achieved by solving reduced linear programs, see \citep{MennleAbaecherliSeuken2015ComputingFrontiers}.} 
Figure \ref{FIG:FRONTIER_PLOT_PLURALITY_VETO_LARGE} shows signatures of Pareto frontiers for $m=3$ alternatives, up to $n =18$ agents, and the worst-case $d^{\textsc{Plu}}$-deficit (a) and $d^{\textsc{Veto}}$-deficit (b), respectively. 

The plots suggest that the Pareto frontier converges to an L-shape as the number of agents grows. 
Based on this observation, we conjecture that for these two desiderata near-ideal mechanisms can be found in large markets. 
If true, we could find mechanisms that are \emph{almost} strategyproof and \emph{almost} achieve the desideratum perfectly. 
Proving these limit results will be an interesting topic for future research. 

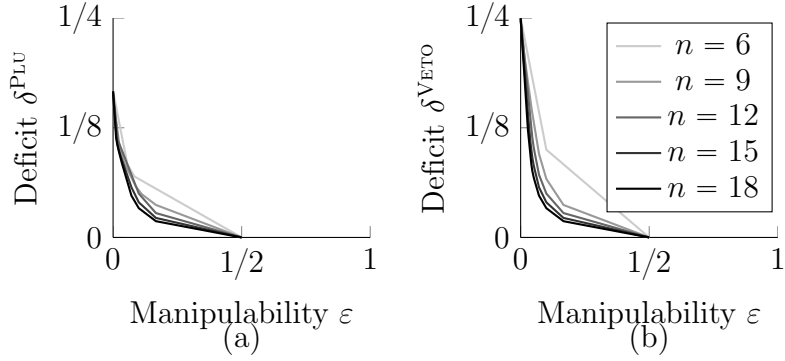
\begin{figure}[t]%
\begin{center}
\begin{tikzpicture}[scale=1]

\begin{groupplot}[
	group style={group name=my plots, group size = 2 by 1, horizontal sep=2cm, },	
	]
	  \nextgroupplot[
			xlabel={Manipulability $\varepsilon$}, 
			ylabel={Deficit $\deficit^{\textsc{Plu}}$},
			ymin=0,
			ymax=.25,
			xmin=0,
			xmax=1, 
			height=4.5cm,
			width=5cm,
			xtick={0,.5,1},
			xticklabels={$0$,$1/2$,$1$},
			ytick={0,.125,.25},
			yticklabels={$0$,$1/8$,$1/4$},
			axis lines*=left, 
			clip=false,
			]
			\addplot[thick,color=Black!20] coordinates { 
				(0,0.166666666666666) (0.0555555555555553,0.0833333333333335) (0.0833333333333329,0.0694444444444445) (0.5,0) 
			};
			\addplot[thick,color=Black!40] coordinates { 
				(0,0.160493827160494) (0.0123456790123429,0.13168724279836) (0.0242656449552998,0.105576841209025) (0.0244755244755153,0.105283605283616) (0.0307692307692302,0.0974358974358979) (0.0342465753424677,0.0936073059360712) (0.0438596491228049,0.0847953216374281) (0.0877192982456125,0.0584795321637433) (0.11111111111111,0.0493827160493828) (0.166666666666666,0.0370370370370372) (0.499999999999999,0) 				
			};
			\addplot[thick,color=Black!60] coordinates { 
				(0,0.166666666666668) (0.0183150183150154,0.11172161172162) (0.0187032418952617,0.110972568578553) (0.0999999999999999,0.05) (0.166666666666666,0.0277777777777777) (0.5,0) 
			};
			\addplot[thick,color=Black!80] coordinates { 
				(0,0.164444444444434) (0.00182149362525058,0.157377049178499) (0.00238095238095272,0.155238095238093) (0.0112994350282485,0.122033898305084) (0.0139552801174379,0.112565233204496) (0.0140753677697606,0.112320859372948) (0.0141997957408191,0.112069195950117) (0.0144888789870446,0.111496788088069) (0.0145900277483926,0.111298917931587) (0.0147343795570679,0.111020472984778) (0.015916088295098,0.108751239367632) (0.0173393973848762,0.106196702671975) 
(0.0174789915966376,0.105949579831934) (0.0185469305315489,0.104134703152007) (0.0222222222222217,0.1) (0.0277777777777776,0.0944444444444445) (0.0322580645161282,0.090322580645162) (0.071428571428571,0.0571428571428574) (0.1,0.0399999999999999) (0.166666666666666,0.0222222222222222) (0.5,0) 
			};
			\addplot[thick,color=Black!100] coordinates { 
				(0,0.166666666666788) (0.0117845117844223,0.113636363636766) (0.0132275132275303,0.111111111111081) (0.0157232704402512,0.106918238993711) (0.0163098878695223,0.106014271151884) (0.0225988700564977,0.0979284369114871) (0.034013605442177,0.085034013605442) (0.0714285714285714,0.0476190476190475) (0.0999999999999999,0.0333333333333333) (0.166666666666666,0.0185185185185185) (0.5,0) 
			};
		\nextgroupplot[
			xlabel={Manipulability $\varepsilon$}, 
			ylabel={Deficit $\deficit^{\textsc{Veto}}$},
			ymin=0,
			ymax=0.25,
			xmin=0,
			xmax=1, 
			height=4.5cm,
			width=5cm,
			xtick={0,.5,1},
			xticklabels={$0$,$1/2$,$1$},
			ytick={0,.125,.25},
			yticklabels={$0$,$1/8$,$1/4$},
			axis lines*=left, 
			clip=false,
			]
			\addplot[thick,color=Black!20] coordinates { 
				(0,0.25) (0.1,0.1) (0.499999999999994,0) 
			};
			\addlegendentry{$n=6$};
			\addplot[thick,color=Black!40] coordinates { 
				(0,0.246913580246913) (0.0134680134680074,0.210250654695114) (0.016032064128256,0.20374081496326) (0.043478260869565,0.144927536231884) (0.0714285714285714,0.0952380952380951) (0.1,0.0666666666666664) (0.166666666666667,0.0370370370370369) (0.5,0) 	
			};
			\addlegendentry{$n=9$};
			\addplot[thick,color=Black!60] coordinates { 
				(0,0.25) (0.0454545454545453,0.113636363636363) (0.0499999999999996,0.104166666666667) (0.0535714285714285,0.0982142857142856) (0.0714285714285714,0.0714285714285714) (0.0999999999999997,0.05) (0.166666666666666,0.0277777777777778) (0.5,0) 
			};
			\addlegendentry{$n=12$};
			\addplot[thick,color=Black!80] coordinates { 
				(0,0.248888888888888) (0.00600231907782644,0.224159334288244) (0.00660066006600677,0.221782178217821) (0.0301204819277105,0.134939759036145) (0.0384615384615387,0.107692307692306) (0.0454545454545454,0.0909090909090909) (0.0500000000000005,0.0833333333333325) (0.0535714285714288,0.0785714285714283) (0.0714285714285711,0.0571428571428574) (0.0999999999999999,0.04) (0.166666666666666,0.0222222222222222) (0.499999999999999,0) 
			};
			\addlegendentry{$n=15$};
			\addplot[thick,color=Black!100] coordinates { 
				(0,0.249999999999943) (0.028301886792554,0.122641509433507) (0.0306122448979518,0.115646258503423) (0.0416666666666666,0.0833333333333333) (0.0454545454545456,0.0757575757575755) (0.05,0.0694444444444442) (0.0535714285714294,0.0654761904761896) (0.0714285714285717,0.0476190476190473) (0.0999999999999998,0.0333333333333334) (0.166666666666666,0.0185185185185185) (0.499999999999996,0) 
			};
			\addlegendentry{$n=18$};
\end{groupplot}
\node[below = 1.0cm of my plots c1r1.south] {(a)};
\node[below = 1.0cm of my plots c2r1.south] {(b)};
\end{tikzpicture}
\end{center}
\caption{Plots of signatures of the Pareto frontier for $n\in\{6,9,12,15,18\}$ agents, $m=3$ alternatives, strict preferences, and worst-case $\desifct^{\textsc{Plu}}$-deficit (left) and worst-case $\desifct^{\textsc{Veto}}$-deficit (right).}
\label{FIG:FRONTIER_PLOT_PLURALITY_VETO_LARGE}
\end{figure}

\section{Conclusion}
\label{SEC:CONCLUSION}
In this paper, we have presented a structural characterization of the Pareto frontier for random ordinal mechanisms. 
Loosely speaking, the \emph{Pareto frontier} consists of those mechanisms that make optimal trade-offs between incentive properties (measured by manipulability $\varepsilon(\varphi)$) and other desiderata (measured by deficit $\deficit(\varphi)$)), such as Condorcet consistency or the goal of choosing an alternative that is the first choice of as many agents as possible. 

We have achieved our main result in three distinct steps: 
first, we have shown that $\varepsilon$-approximate strategyproofness can be equivalently expressed by a finite set of linear constraints. 
This has enabled us to define the linear program \FindOpt\ to identify the set of all mechanisms that have minimal deficit but satisfy $\varepsilon$-approximate strategyproofness for a given manipulability bound $\varepsilon$. 
Second, we have shown how hybrid mechanisms, which are convex combinations of two component mechanisms, trade off manipulability and deficit. 
In particular, we have given a guarantee that the signature of a $\beta$-hybrid $h_{\beta}$ is always at least as good as the $\beta$-convex combination of the signatures of the two component mechanisms. 
Third, we have shown that the Pareto frontier consists of two building blocks: 
(1) there exists a finite set of \emph{supporting manipulability bounds} $\varepsilon_0,\ldots,\varepsilon_K$ such that we can characterize the set of optimal mechanisms at each of the bounds $\varepsilon_k$ as the set of solutions to the linear program \FindOpt\ at $\varepsilon_k$, and 
(2) for any intermediate manipulability bound $\varepsilon = (1-\beta)\varepsilon_{k-1} + \beta \varepsilon_k$, the set of optimal mechanisms at $\varepsilon$ is precisely the set of $\beta$-hybrids of optimal mechanisms at each of the two adjacent supporting manipulability bounds $\varepsilon_{k-1}$ and $\varepsilon_k$. 

Our results have a number of interesting consequences (beyond their relevance in this paper): 
first, Theorem \ref{THM:ASP_EQUIV_LINEAR_CONSTRAINTS} gives a finite set of linear constraints that is equivalent to $\varepsilon$-approximate strategyproofness. 
This makes $\varepsilon$-approximate strategyproofness accessible to algorithmic analysis. 
In particular, it enables the use of this incentive requirement under the automated mechanism design paradigm. 

Second, the performance guarantees for hybrid mechanisms from Theorem \ref{THM:HYBRID_GUARANTEES} 
yield convincing arguments in favor of randomization. 
In particular, we learn that the important requirements of anonymity and neutrality come ``for free;'' 
mechanism designers do not have to accept a less desirable signature when imposing either or both (provided that the deficit measure is anonymous, neutral, or both).

Third, our main result, Theorem \ref{THM:FRONTIER_STRUCTURE}, has provided a \emph{structural understanding} of the whole Pareto frontier. 
Knowledge of the Pareto frontier enables mechanism designers to make a completely informed decision about trade-offs between manipulability and deficit. 
In particular, we now have a way to determine precisely by how much the performance of mechanisms (with respect to a given desideratum) can be improved when allowing additional manipulability. 
An important learning is that the mapping $\varepsilon\mapsto\deficit(\varepsilon)$, which associates each manipulability bound with the lowest achievable deficit at this manipulability bound, is \emph{monotonic}, \emph{decreasing}, \emph{convex}, \emph{continuous}, and \emph{piecewise linear}. 
This means that when trading off manipulability and deficit along the Pareto frontier, the trade-offs are smooth, and
the earliest sacrifices yield the greatest gains. 
Thus, it can be worthwhile to consider even small bounds $\varepsilon >0$ in order to obtain substantial improvements. 

Finally, we have illustrated our results by considering two concrete problems. 
In both problems, three agents had strict preferences over three alternatives. 
In the first problem, the desideratum was to choose an alternative that is the \emph{first choice of as many agents as possible} (i.e., Plurality scoring), and in the second problem, the desideratum was to choose an alternative that is the \emph{last choice of as few agents as possible} (i.e., Veto scoring).  
In both problems, we have computed the Pareto frontier and verified the resulting structure analytically. 
The examples have shown that the Pareto frontier may be completely linear (first problem) or truly non-linear (second problem). 
For the same desiderata and up to $n=18$ agents we have determined the Pareto frontier algorithmically and formulated a conjecture about its limit behavior. 

In summary, we have given novel insights about the Pareto frontier for random ordinal mechanisms. 
We have proven our results for the full ordinal domain that includes indifferences, but they continue to hold for many other interesting domains that arise by restricting the space of preference profiles, such as the assignment domain and the two-sided matching domain.
When impossibility results restrict the design of strategyproof mechanisms, we have provided a new perspective on the unavoidable trade-off between incentives and other desiderata along this Pareto frontier. 
%
%


\newpage
\appendix
\section*{APPENDIX}
\setcounter{section}{0}

\section{Relative and Ex-ante Deficit}
\label{APP:DEFICITS}
We have defined the deficit of outcomes (Definition \ref{DEF:DEFICIT}) as the absolute difference between the achievable and the achieved $\desifct$-value; 
and we have defined the deficit of mechanisms (Definition \ref{DEF:WORST_CASE_DEFICIT}) as the worst-case deficit across all possible preference profiles. 
In this section we present two variations of these definitions, the \emph{relative deficit} of outcomes and the \emph{ex-ante deficit} of mechanisms. 
Our results in this paper hold for any of these variations. 
\subsection{Relative Deficit of Outcomes}
\label{APP:DEFICITS:RELATIVE}
In some situations, it may be more natural to consider a \emph{relative} difference, e.g., the ratio between the achieved and the maximal achievable $\desifct$-value. 
As we show next, it is without loss of generality that in this paper we have restricted our attention to absolute differences. 
This is because the relative $\desifct$-deficit can always be expressed as an absolute $\tilde{\desifct}$-deficit, where the adjusted desideratum function $\tilde{\desifct}$ is obtained from $\desifct$ by scaling. 
Proposition \ref{PROP:RELATIVE_ABSOLUTE_EQUIVALENCE} makes this argument precise. 

To state this equivalence formally, we need to define the relative deficit: 
for any preference profile $\bm P\in \mathcal{P}^N$, the \emph{$\desifct$-value margin at $\bm P$} is the difference between the highest and the lowest $\desifct$-value achievable by any alternative at $\bm P$. 
We set $\desifct^{\max}(\bm P) = \max_{j\in M} \desifct(j,\bm P)$, $ \desifct^{\min}(\bm P) = \min_{j\in M} \desifct(j,\bm P)$, and $\desifct^{\text{margin}}(\bm P) = \desifct^{\max}(\bm P) - \desifct^{\min}(\bm P)$. 
Note that for the special case where $\desifct^{\text{margin}}(\bm P) = 0$, all alternatives (and therefore all outcomes) have the same $\desifct$-value. 
In this case, \emph{any} alternative is $\desifct$-maximizing at $\bm P$. 
For an outcome $x \in \Delta(M)$, the \emph{relative $\desifct$-deficit of $x$ at $\bm P$} is the $\desifct$-deficit of $x$ at $\bm P$, normalized by the $\desifct$-value margin at $\bm P$; 
formally,  
\begin{equation}
	\deficit_\desifct^{\text{relative}}(x,\bm P) = \left\{
		\begin{array}{ll}
			\frac{\desifct^{\max}(\bm P) - \desifct(x,\bm P)}{\desifct^{\text{margin}}(\bm P)}, & \text{ if }\desifct^{\text{margin}}(\bm P) > 0, \\
			0, & \text{ else}.
		\end{array}
	\right.
\end{equation}
\begin{proposition}
\label{PROP:RELATIVE_ABSOLUTE_EQUIVALENCE}
For any desideratum function $\desifct$, there exists a desideratum function $\tilde{\desifct}$ such that the relative $\desifct$-deficit coincides with the absolute $\tilde{\desifct}$-deficit, such that for all outcomes $x \in \Delta(M)$ and all preference profiles $\bm P \in \mathcal{P}^N$, we have
\begin{equation}
	\deficit_{\desifct}^{\text{relative}}(x,\bm P) = \deficit_{\tilde{\desifct}}(x,\bm P).
\end{equation}
\end{proposition}
The proof follows immediately by setting $\tilde{\desifct}(j,\bm P) = \frac{\desifct^{\max}(\bm P) - \desifct(j,\bm P)}{\desifct^{\text{margin}}(\bm P)}$, whenever $\desifct^{\text{margin}}(\bm P) > 0$, and $\tilde{\desifct}(j,\bm P) = 0$ otherwise. 
Proposition \ref{PROP:RELATIVE_ABSOLUTE_EQUIVALENCE} shows that including relative deficit does not enrich the space of possible criteria but that the space of desideratum functions is rich enough to cover relative deficits implicitly. 
Thus, it is without loss of generality that we have restricted attention to absolute deficits in this paper. 
\subsection{Ex-ante Deficit of Mechanisms}
\label{APP:DEFICITS:EXANTE}
Recall that we defined the deficit of mechanisms as the worst-case loss that society incurs from using a mechanism that is not $d$-maximizing. 
This is most meaningful when we have no prior knowledge of the agents' preferences and therefore need to design mechanisms that achieve the desideratum as well as possible across all possible preference profiles. 

However, in some situations, we may have probabilistic information about the agents' preferences, which we would like to exploit to design better mechanisms. 
Suppose that the agents' preference profiles are drawn from a known distribution $\mathds{P}$. 
In this case, we may prefer a mechanism that induces high \emph{expected} $\desifct$-value under $\mathds{P}$. 
\begin{definition}[Ex-ante Deficit]
\label{DEF:EX_ANTE_DEFICIT}
Given a setting $(N,M)$, a desideratum function $\desifct$, a probability distribution $\mathds{P}$ over preference profiles, and a mechanism $\varphi$, the \emph{ex-ante $\desifct$-deficit of $\varphi$ with respect to $\mathds{P}$ (in $(N,M)$)} is 
\begin{equation}
	\deficit_{\desifct}^{\mathds{P}}(\varphi) = \sum_{\bm P \in \mathcal{P}^N} \mathds{P}[\bm P] \cdot \deficit_{\desifct}(\varphi(\bm P),\bm P).
\end{equation}
\end{definition}
Minimizing $\deficit_{\desifct}^{\mathds{P}}(\varphi)$ corresponds to minimizing the expected $\desifct$-deficit from applying $\varphi$ \emph{ex-ante} (i.e., before the agents' preferences are instantiated from $\mathds{P}$). 
This approach is attractive in situations where the same mechanism is applied repeatedly for different groups of agents, so that the outcomes are attractive \emph{on average} across all the repetitions. 

We can incorporate the ex-ante deficit into the linear program \FindOpt\ (Linear Program \ref{LP:FINDOPT}) by exchanging the constraint labeled \emph{(Deficit)}.
Instead of making the variable $\variable{d}$ an upper bound for the deficit of the mechanism across all preference profiles, we need to make $\variable{d}$ an upper bound for the expected deficit of the mechanism, where this expectation is taken with respect to the distribution $\mathds{P}$ over preference profiles. 
We achieve this by including the following constraint.
\begin{equation}
	\variable{d} \geq 
		\sum_{\bm P \in \mathcal{P}^N}
		\mathds{P}[\bm P] \cdot \left( 
			\max_{j\in M}\desifct(j,\bm P) - \sum_{j \in M} \variable{f_j(\bm P)} \cdot \desifct(j,\bm P)
		\right).
	\tag{\emph{Ex-ante deficit}}
\end{equation}
\section{Strict Improvements from Hybrids in Theorem \ref{THM:HYBRID_GUARANTEES}}
\label{APP:HYBRIDS_STRICT_IMPROVEMENT}
Theorem \ref{THM:HYBRID_GUARANTEES} showed two weak inequalities for hybrid mechanisms. 
We now give an example that shows that a hybrid of two mechanisms can in fact have a strictly lower manipulability and a strictly lower deficit than both of its component mechanisms. 
\begin{example}
\label{EX:STRICT_IMPROVEMENTS_HYBRIDS} 
Consider a problem with one agent and three alternatives $a,b,c$, where $\deficit$ is the worst-case deficit that arises from Plurality scoring. 
Let $\varphi$ and $\psi$ be two mechanisms whose outcomes depend only on the agent's relative ranking of $b$ and $c$. 
\begin{center}
\begin{tabular}{|c|c|c|c|c|c|c|}
	\hline
 	& \multicolumn{3}{|c|}{$\bm \varphi$} & \multicolumn{3}{|c|}{$\bm \psi$} \\
	\hline
	\textbf{Report} & \hspace{1em}$\bm a$\hspace{1em} & \hspace{1em}$\bm b$\hspace{1em} & \hspace{1em}$\bm c$\hspace{1em} & \hspace{1em}$\bm a$\hspace{1em} & \hspace{1em}$\bm b$\hspace{1em} & \hspace{1em}$\bm c$\hspace{1em} \\
	\hline\hline
	If $P:b \succeq c$ & 0 & 2/3 & 1/3 & 5/9 & 1/9 & 1/3 \\
	\hline
	If $P:c \succ b$ & 1/3 & 1/3 & 1/3 & 1/9 & 5/9 & 1/3 \\
	\hline
\end{tabular}
\end{center}
It is a simple exercise to verify that the mechanisms' signatures are $(\varepsilon(\varphi),\deficit(\varphi)) = (1/3,1)$ and $(\varepsilon(\psi),\deficit(\psi)) = (4/9,8/9)$. 
Furthermore, for $\beta = 3/7$, the hybrid $h_{\beta}$ is constant and therefore strategyproof, and it has a deficit of $\deficit(h_{\beta}) = 16/21$. 
Figure \ref{FIG:EXAMPLE_HYB_GUARANTEE} illustrates these signatures. 
The guarantees from Theorem \ref{THM:HYBRID_GUARANTEES} are represented by the shaded area, where the signature of $h_{\beta}$ can lie anywhere below and to the left of the $\beta$-convex combination of the signatures of $\varphi$ and $\psi$. 
\end{example}
\begin{figure}[t]
\centering
\begin{tikzpicture}[scale=4]
\draw[-] (-0.05,0) -- (1.0,0) coordinate (x axis);
\draw[-] (0,-0.05) -- (0,1.0) coordinate (y axis);
\draw[-] (1,-0.05) -- (1,0);
\draw[-] (-0.05,1) -- (0,1);
\draw (-0.15,1) node {1};
\draw (1.0, -0.15) node {1};
\draw (0, -0.15) node {0};
\draw (-0.15,0) node {0};
\draw (0.5,-0.3) node { Manipulability};
\draw (0.5,-0.15) node { $\varepsilon$};
\draw (-.4,0.5) node {Deficit};
\draw (-.15,0.5) node {$\deficit$};
\newcommand{\manipphi}{1/3}
\newcommand{\deficphi}{1}
\newcommand{\manippsi}{4/9}
\newcommand{\deficpsi}{8/9}
\newcommand{\maniphyb}{0}
\newcommand{\defichyb}{16/21}
\newcommand{\hybsigmanip}{8/21}
\newcommand{\hybsigdefic}{20/21}
\fill [fill=LightGray] (0.002,0.002) rectangle (\hybsigmanip,\hybsigdefic);
\draw[-,dashed] (\hybsigmanip,0) -- (\hybsigmanip,\hybsigdefic);
\draw[-,dashed] (0,\hybsigdefic) -- (\hybsigmanip,\hybsigdefic);
\draw[-] (\manipphi,\deficphi) -- (\manippsi,\deficpsi);
\draw [fill=White](\manipphi,\deficphi) circle [radius=0.02];
\draw [fill=White](\manippsi,\deficpsi) circle [radius=0.02];
\draw [fill=Black](\maniphyb,\defichyb) circle [radius=0.02];
\node[label=right:$\varphi$] at (\manipphi,\deficphi) {};
\node[label=right:$\psi$] at (\manippsi,\deficpsi) {};
\node[label=right:$h_{\beta}$] at (\maniphyb,\defichyb) {};
\end{tikzpicture}
\caption{Example signature of hybrid $h_{\beta}$ (\textbullet) must lie in the shaded area, weakly to the left and below the $\beta$-convex combinations of the signatures of $\varphi$ and $\psi$ ($\circ$).}
\label{FIG:EXAMPLE_HYB_GUARANTEE}
\end{figure}
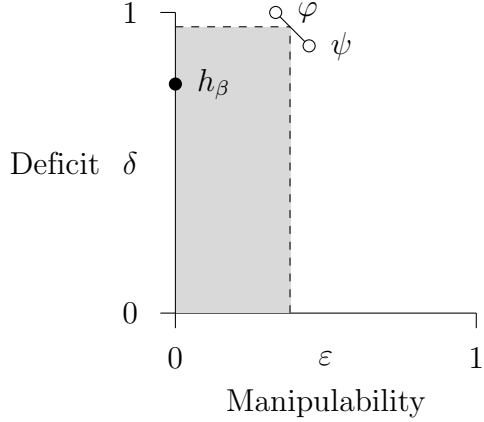

\section{Anonymity, Neutrality, Symmetric Decomposition}
\label{APP:ANON_NEUTR_SYMM_DECOMP}
Two common requirements in the design of ordinal mechanisms are anonymity and neutrality. 
Anonymity captures the intuition that a mechanism should not discriminate between agents; 
instead, the influence of any agent on the outcome should be the same and independent of the agent's name. 
Neutrality requires that the mechanism is not biased towards particular alternatives; 
the decision should depend only on the agents' preferences but not on the names of the alternatives. 

In this section we describe the implications of these two requirements for the design of optimal mechanisms: 
after providing definitions, we show how they can be incorporated in the linear program \FindOpt. 
Then we formally prove the second takeaway from Theorem \ref{THM:HYBRID_GUARANTEES} that (under certain conditions) both requirements come for free in terms of the signatures of optimal mechanisms. 
Finally, we prove a new symmetric decomposition result for strategyproof, anonymous, neutral mechanisms, that extends the characterization of strategyproof random mechanisms in \citep{Gibbard1977ManipulationOfVotingWithChance}. 
\subsection{Definition of Anonymity and Neutrality}
\label{APP:ANON_NEUTR_SYMM_DECOMP:DEFINITION}
First, we define anonymity: for any \emph{renaming of the agents} (i.e., any bijection $\pi:N\rightarrow N$) and any preference profile $\bm P = (P_1,\ldots,P_n) \in \mathcal{P}^N$, let $\bm P^{\pi} =(P_{\pi(1)},\ldots,P_{\pi(n)})$ be the preference profile where the agents have exchanged their roles according to $\pi$. 
Agent $i$ is now reporting the preference order $P_{\pi(i)}$ that was reported by agent $\pi(i)$ under the original preference profile $\bm P$. 
For any mechanism $\varphi$ let $\varphi^{\pi}$ be the mechanism under which the agents trade roles according to $\pi$; 
formally, let $\varphi^{\pi}(\bm P) = \varphi(\bm P^{\pi})$ for any preference profile $\bm P \in \mathcal{P}^N$.
\begin{definition}[Anonymity] 
\begin{itemize}
	\setlength{\itemsep}{0pt}
	\item A desideratum function $\desifct$ is \emph{anonymous} if for all renamings $\pi: N \rightarrow N$, preference profiles $\bm P \in \mathcal{P}^N$, and alternatives $j \in M$, we have $\desifct(j,\bm P) = \desifct(j,\bm P^{\pi})$ (i.e., the $\desifct$-value is independent of the order in which the agents submit their preferences).
	\item A probability distribution $\mathds{P}$ over preference profiles is \emph{anonymous} if for all renamings $\pi: N \rightarrow N$ and preference profiles $\bm P \in \mathcal{P}^N$, we have $\mathds{P}[\bm P] = \mathds{P}[\bm P^{\pi}]$ (i.e., the probability of a preference profile does not depend on the order in which the different preference orders appear). 
	\item The worst-case deficit $\deficit_\desifct$ is \emph{anonymous} if the underlying desideratum function $\desifct$ is anonymous. 
		The ex-ante deficit $\deficit_\desifct^{\mathds{P}}$ is \emph{anonymous} if the underlying desideratum function and the probability distribution $\mathds{P}$ are anonymous. 
	\item A mechanism $\varphi$ is \emph{anonymous} if for all renamings $\pi: N \rightarrow N$ and preference profiles $\bm P \in \mathcal{P}^N$, we have $\varphi(\bm P) = \varphi^{\pi}(\bm P)$ (i.e., the outcome of the mechanism is independent of the order in which the agents submit their preferences). 
\end{itemize}
\end{definition}
Next, we define neutrality: 
for any \emph{renaming of the alternatives} (i.e., any bijection $\varpi:M\rightarrow M$) and any preference order $P_i \in \mathcal{P}$, let $P_i^{\varpi}$ be the preference order under which $P_i^{\varpi}: \varpi(j) \succeq \varpi(j')$ whenever $P_i: j \succeq j'$ for any alternatives $j,j' \in M$. 
This means that $P_i^{\varpi}$ corresponds to $P_i$, except that the all alternatives have been renamed according to $\varpi$. 
For any preference profile $\bm P = (P_1,\ldots,P_n) \in \mathcal{P}^N$, let $\bm P^{\varpi} =(P_1^{\varpi},\ldots,P_{n}^{\varpi})$ be the preference profile where the alternatives inside the agents' preference orders have been renamed according to $\varpi$. 
For any mechanism $\varphi$ let $\varphi^{\varpi}$ be the mechanism under which the alternatives are renamed according to $\varpi$; 
formally, let $\varphi^{\varpi}_j(\bm P) = \varphi_{\varpi(j)}(\bm P^{\varpi})$ for all preference profile $\bm P \in \mathcal{P}^N$ and alternative $j \in M$.
\begin{definition}[Neutrality] 
We define the following: 
\begin{itemize}
	\setlength{\itemsep}{0pt}
	\item A desideratum function $\desifct$ is \emph{neutral} if for all renamings $\varpi:M\rightarrow M$, preference profiles $\bm P \in \mathcal{P}^N$, and alternatives $j \in M$, we have $\desifct(j,\bm P) = \desifct(\varpi(j),\bm P^{\varpi})$ (i.e., the $\desifct$-value of any alternative is independent of its name).
	\item A probability distribution $\mathds{P}$ over preference profiles is \emph{neutral} if for all renamings $\varpi:M\rightarrow M$ and preference profiles $\bm P \in \mathcal{P}^N$, we have $\mathds{P}[\bm P] = \mathds{P}[\bm P^{\varpi}]$ (i.e., the probability of a preference profile does not depend on the names of the alternatives).  
	\item The worst-case deficit $\deficit_\desifct$ is \emph{neutral} if the underlying desideratum function $\desifct$ is neutral. 
		The ex-ante deficit $\deficit_\desifct^{\mathds{P}}$ is \emph{neutral} if the underlying desideratum function and the probability distribution $\mathds{P}$ are neutral. 
	\item A mechanism $\varphi$ is \emph{neutral} if for all renamings $\varpi:M\rightarrow M$, preference profiles $\bm P \in \mathcal{P}^N$, and alternatives $j \in M$, we have $\varphi_j(\bm P) = \varphi_{j}^{\varpi}(\bm P) = \varphi_{\varpi(j)}(\bm P^{\varpi})$ (i.e., the outcomes of the mechanism are independent of the names of the alternatives). 
\end{itemize}
\end{definition}
Incorporating anonymity and neutrality as additional constraints in the linear program \FindOpt\ is straightforward and can be done as follows.

\medskip
\textsc{Linear Program \ref{LP:FINDOPT} (\FindOpt), extended.}
\begin{equation*}
\begin{array}{lllr}
	 & \ldots &  \\
	& \variable{f_j(\bm P)} = \variable{f_j(\bm P^{\pi})}, & \forall \bm P \in \mathcal{P}^N, j \in M, \pi:N \rightarrow N \text{ bijection}& \text{\emph{(Anonymity)}} \\
	& \variable{f_j(\bm P)} = \variable{f_{\varpi(j)}(\bm P^{\varpi})}, & \forall \bm P \in \mathcal{P}^N, j \in M, \varpi:M \rightarrow M \text{ bijection}& \text{\emph{(Neutrality)}} 
\end{array}
\end{equation*}
\medskip

\subsection{Costlessness of Requiring Anonymity and Neutrality}
\label{APP:ANON_NEUTR_SYMM_DECOMP:COSTLESS}
With these notions of anonymity and neutrality in mind, observe that for any given mechanism it is possible to construct an anonymous mechanism or a neutral mechanism by randomizing over the roles of the agents or the roles of the alternatives in the mechanism, respectively: let
\begin{equation}
	\varphi^{\text{anon}}_j(\bm P) = \frac{1}{n!} \sum_{\scriptsize{\begin{array}{c}
			\pi:N \rightarrow N \\
			\text{bijection}\end{array}}} \varphi_j^{\pi}(\bm P)
\end{equation}
and 
\begin{equation}
	\varphi_j^{\text{neut}}(\bm P) = \frac{1}{m!} \sum_{\scriptsize{\begin{array}{c}
			\varpi:M \rightarrow M \\
			\text{bijection}\end{array}}} \varphi_{j}^{\varpi}(\bm P)
\end{equation}
for all preference profiles $\bm P \in \mathcal{P}^N$ and alternatives $j \in M$. 
Observe that $\varphi^{\text{anon}}$ and $\varphi^{\text{neut}}$ are simply hybrid mechanisms with many components, where each component is used with a probability of $\frac{1}{n!}$ or $\frac{1}{m!}$, respectively. 
With these definitions we can formally derive the second main takeaway from Theorem \ref{THM:HYBRID_GUARANTEES}. 
\begin{corollary}
\label{COR:ANON_NEUTR_FREE}
Given a problem \NMO, where the deficit is anonymous/neutral/both, for any mechanism $\varphi$ there exists a mechanism $\tilde{\varphi}$ that is anonymous/neutral/both and has a weakly better signature than $\varphi$; 
formally, $\varepsilon(\tilde{\varphi}) \leq \varepsilon(\varphi)$ and $\deficit(\tilde{\varphi}) \leq \deficit(\varphi)$
\end{corollary}
\begin{proof}
Observe that if the deficit $\deficit$ is anonymous, then $\varphi$ and $\varphi^{\pi}$ have the same signature for any renaming of the agents $\pi: N \rightarrow N$. 
Consequently, $\varphi^{\text{anon}}$ is a hybrid of $n!$ mechanisms which all have the same signature, and therefore it has a weakly better signature by Theorem \ref{THM:HYBRID_GUARANTEES}. 
Similarly, if $\deficit$ is neutral, then $\varphi$ and $\varphi^{\varpi}$ have the same signature for any renaming of the alternatives $\varpi: M \rightarrow M$, and the result follows analogously. 
\end{proof}

Intuitively, Corollary \ref{COR:ANON_NEUTR_FREE} means that, given the right desideratum, the two requirements anonymity and neutrality are ``free'' in terms of manipulability and deficit. 
We do not have to accept a less attractive signatures in order to achieve either property. 
\subsection{Non-trivial Signature Cost of Other Desiderata}
\label{APP:ANON_NEUTR_SYMM_DECOMP:OTHERS}
In contrast to anonymity and neutrality, other common desiderata do not come for free, such as Condorcet consistency, Pareto optimality, or even the rather weak requirement of unanimity. 
\begin{definition}[Unanimity, Pareto Optimality, Condorcet Consistency]
For a given preference profile $\bm P \in \mathcal{P}^N$, define the following:
\begin{itemize}
	\setlength{\itemsep}{0pt}
	\item An alternative $j \in M$ is a \emph{unanimity winner at $\bm P$} if for all agents $i\in N$ and all other alternatives $j' \in M$ we have $P_i: j \succeq j'$.
		Let $M^{\text{unan}}(\bm P)$ be the set of all unanimity winners at $\bm P$ and $M^{\neg\text{unan}}(\bm P) = M \backslash M^{\text{unan}}(\bm P)$ the set of non-winners.
	\item An alternative $j \in M$ \emph{Pareto dominates} another alternative $j'\in M$ at $\bm P$ if for all agents $i \in N$ we have $P_i: j \succeq j'$ and for some agent $i'$ we have $P_{i'}: j \succ j'$. 
		$j$ is \emph{Pareto optimal at $\bm P$} if there exists no other alternative $j'$ that Pareto dominates $j$.
		Let $M^{\text{Pareto}}(\bm P)$ be the set of Pareto optimal alternatives at $\bm P$, and let $M^{\neg\text{Pareto}}(\bm P) = M \backslash M^{\text{Pareto}}(\bm P)$ be the set of alternatives that are Pareto dominated by another alternative at $\bm P$.
	\item For any two alternatives $a,b \in M$ let $ n_{a \succ b}(\bm P) = \#\left\{ i \in N~\left|~P_i: a \succ b \right. \right\}$ be the number of agents who strictly prefer $a$ to $b$ at $\bm P$.
		An alternative $j\in M$ is a \emph{Condorcet winner at $\bm P$} if for all other alternatives $j' \in M$ we have $n_{j \succ j'}(\bm P) \geq n_{j' \succ j}(\bm P)$. 		
		$M^{\text{Condorcet}}(\bm P)$ is the set of Condorcet winners at $\bm P$, and the set of non-Condorcet winners is $M^{\neg\text{Condorcet}}(\bm P) = M \backslash M^{\text{Condorcet}}(\bm P)$.
\end{itemize}
\end{definition}
A mechanism $\varphi$ satisfies unanimity, Pareto optimality, or Condorcet consistency, if it only selects alternatives that have the respective property whenever they exist. 
One way to incorporate these desiderata into the linear program \FindOpt\ is to include them in the objective function by using a desideratum function that assigns higher value to alternatives that have the respective property. 

However, this is not sufficient if the optimal mechanisms must satisfy the requirement \emph{completely}, independent of the resulting increase in manipulability. 
For example, the goal could be to find a unanimous, strategyproof mechanism that minimizes the worst-case deficit based on Veto scoring. 
In this case, it would not suffice to include unanimity in the desideratum function, because the only way to then guarantee unanimity would be to impose full value maximization. 
Alternatively, we can incorporate the property as additional constraints in the linear program \FindOpt\ directly. 
The following linear constraints can be used to require unanimity, Pareto optimality, and Condorcet consistency, respectively. 

\medskip
\textsc{Linear Program \ref{LP:FINDOPT} (\FindOpt), extended.}
\begin{equation*}
\begin{array}{llr}
	 & \ldots &  \\
	& \variable{f_j(\bm P)} = 0, & \text{(Unanimity)} \\
	& \hspace{2em} \forall \bm P \in \mathcal{P}^N \text{ such that } M^{\text{unan}}(\bm P) \neq \emptyset \text{ and } j \in M^{\neg\text{unan}}(\bm P) \\
	& \variable{f_j(\bm P)} = 0,  & \text{(Pareto)}  \\
	& \hspace{2em} \forall \bm P \in \mathcal{P}^N, j \in M^{\neg\text{Pareto}}(\bm P) \\	
	& \variable{f_j(\bm P)} = 0,  & \text{(Condorcet)}  \\	
	& \hspace{2em} \forall \bm P \in \mathcal{P}^N \text{ such that } M^{\text{Condorcet}}(\bm P) \neq \emptyset \text{ and } j \in M^{\neg\text{Condorcet}}(\bm P) 
\end{array}
\end{equation*}
\medskip

The next example illustrates that, unlike anonymity and neutrality, unanimity does not come for free in terms of the mechanisms' signatures. 
Since Pareto optimality and Condorcet consistency imply unanimity, the same is true for both other desiderata. 
\begin{example}
\label{EX:UNANIMITY_NOT_FREE}
Consider the same problem as in Section \ref{SEC:COMPUTE:ANALYTIC_EXAMPLES}, where $n = m = 3$, agents have strict preferences over alternatives $a,b,c$, $\deficit$ is the worst-case deficit based on Veto scoring. 
Let $\varphi$ be a strategyproof mechanism that is also unanimous. 

By the characterization in \citep{Gibbard1977ManipulationOfVotingWithChance}, $\varphi$ must be a hybrid of strategyproof unilateral and strategyproof duple mechanisms. 
However, for $m \geq 3$ alternatives, no duple mechanism is unanimous, and unilateral mechanisms are unanimous only if they always select the first choice of the dictating agent. 
But as soon as a single component of a hybrid is not unanimous, the hybrid is not unanimous either. 
Consequently, $\varphi$ must be a hybrid of dictatorships. 
Since unanimity is an anonymous and neutral constraint (i.e., the constraint is invariant to renamings of agents or alternatives), we obtain from Corollary \ref{COR:ANON_NEUTR_FREE} that the signature of $\varphi$ is at most as good as the signature of Random Dictatorship, where the dictating agent is chosen \emph{uniformly} at random. 
However, Random Dictatorship has a deficit of $\frac{4}{9}$ at the preference profile
\begin{eqnarray}
	P_1,P_2 &: & a \succ b \succ c, \\
	P_3 &: & c \succ b \succ a.
\end{eqnarray}
But we already observed that Random Duple is strategyproof and optimal (but not unanimous) in this problem with strictly lower deficit of $\frac{2}{9}$. 
This means that requiring unanimity in this problem leads to a strict increase in the lowest deficit that is achievable with strategyproof, optimal mechanisms.
\end{example}
\subsection{Symmetric Decomposition of Strategyproof, Anonymous, and Neutral Mechanisms} 
\label{APP:ANON_NEUTR_SYMM_DECOMP:SYMM_DECOMP}
We present a refinement of Gibbard's strong characterization of strategyproof mechanisms \citep{Gibbard1977ManipulationOfVotingWithChance}.
Our \emph{symmetric decomposition} characterizes mechanisms that are \emph{strategyproof}, \emph{anonymous}, and \emph{neutral}.
We use this result to establish the shapes of Pareto frontiers in Sections \ref{SEC:COMPUTE:ANALYTIC_EXAMPLES} analytically. 
Beyond this application, the symmetric decomposition may be of independent interest.

In the full domain of \emph{strict} preferences, \citet{Gibbard1977ManipulationOfVotingWithChance} showed that any strategyproof mechanism is a hybrid of multiple ``simple'' mechanisms, namely strategyproof \emph{unilateral} and \emph{duple} mechanisms.
\begin{definition}[\citeauthor{Gibbard1977ManipulationOfVotingWithChance}, \citeyear{Gibbard1977ManipulationOfVotingWithChance}]
\label{DEF:UNILATERAL}
A mechanism $\text{uni}$ is \emph{unilateral} if the outcome only depends on the report of a single agent; 
formally, there exists $i \in N$ such that for all preference profiles $\bm P, \bm P' \in \mathcal{P}^N$ we have that $P_i = P_i'$ implies $\text{uni}(\bm P) = \text{uni}(\bm P')$. 
\end{definition}
\begin{definition}[\citeauthor{Gibbard1977ManipulationOfVotingWithChance}, \citeyear{Gibbard1977ManipulationOfVotingWithChance}]
\label{DEF:DUPLE}
A mechanism $\text{dup}$ is \emph{duple} if only two alternatives are possible; 
formally, there exist $a,b \in M$ such that for all preference profiles $\bm P \in \mathcal{P}^N$ we have $ \text{dup}_j(\bm P) = 0 $ for all $j \neq a,b$.
\end{definition}
The strong characterization result is the following.
\begin{fact}[\citeauthor{Gibbard1977ManipulationOfVotingWithChance}, \citeyear{Gibbard1977ManipulationOfVotingWithChance}] 
\label{FACT:GIBBARD_DECOMPOSITION}
A mechanism $\varphi$ is strategyproof \emph{if and only if} it can be written as a hybrid of mechanisms $\varphi^1,\ldots,\varphi^K$, and each component $\varphi^k$ is strategyproof and either unilateral or duple.\footnote{\citeauthor{Gibbard1977ManipulationOfVotingWithChance} further refined this result by replacing strategyproofness with \emph{localized} and \emph{non-perverse}.} 
\end{fact}
Obviously, duple mechanisms cannot satisfy neutrality (unless $m=2$) and unilateral mechanisms cannot satisfy anonymity (unless the mechanism is constant or $n=1$).
This means that anonymity and neutrality of strategyproof mechanisms are the result of mixing the unilateral and duple components ``correctly.''
This intuition gives rise to the following more refined decomposition of strategyproof, anonymous, neutral mechanism.
\begin{theorem}[Symmetric Decomposition]
\label{THM:SYMMETRIC_DECOMP} 
A mechanism $\varphi$ is strategyproof, anonymous, and neutral if and only if there exist
\begin{enumerate}
	\setlength{\itemsep}{0pt}
	\item strategyproof, neutral, unilateral mechanisms ${\text{uni}}_k,k \in \{1,\ldots,K^{\text{uni}}\}$,
	\item strategyproof, anonymous, duple mechanisms ${\text{dup}}_k,k \in \{K^{\text{uni}}+1,\ldots,K^{\text{uni}} + K^{\text{dup}}\}$,
	\item coefficients $\beta_k \geq 0, k \in \{1,\ldots,K^{\text{uni}} + K^{\text{dup}}\}$ with $\sum_{k=1}^{K^{\text{uni}} + K^{\text{dup}}}\beta_k = 1$,
\end{enumerate}
such that
\begin{equation}
	\varphi = \sum_{\scriptsize{\begin{array}{c}
			\pi:N \rightarrow N \\
			\text{bijection}		
		\end{array}}}
		\sum_{k=1}^{K^{\text{uni}}} \left(\frac{\beta_k}{n!}\right){\text{uni}}_{k}^{\pi} +
			\sum_{\scriptsize{\begin{array}{c}
			\varpi:M \rightarrow M \\
			\text{bijection}		
		\end{array}}} \sum_{k=K^{\text{uni}} + 1}^{K^{\text{uni}} + K^{\text{dup}}} \left(\frac{\beta_k}{m!}\right){\text{dup}}_{k}^{\varpi}.
\label{EQ:SYMMETRIC_DECOMPOSITION}
\end{equation}
\end{theorem}
\begin{proof}
By anonymity and neutrality of $\varphi$ we get that 
\begin{equation}
	\varphi = \left(\varphi^{\pi}\right)^{\varpi} = \left(\varphi^{\varpi}\right)^{\pi} = \varphi^{\pi,\varpi}
\end{equation}
for all bijections $\pi: N\rightarrow N$ and $\varpi: M \rightarrow M$, which implies
\begin{equation}
	\varphi = \sum_{\pi,\varpi} \frac{1}{n!m!} \varphi^{\pi,\varpi}.
\end{equation}
Since $\varphi^{\pi,\varpi}$ is strategyproof, we can use Fact \ref{FACT:GIBBARD_DECOMPOSITION} to decompose it into $K^{\pi,\varpi}$ strategyproof unilateral and duple mechanisms, i.e.,
\begin{equation}
	\varphi^{\pi,\varpi} = \sum_{k=1}^{K^{\pi,\varpi}} \beta_k^{\pi,\varpi} \varphi^{k,\pi,\varpi}.
\end{equation}
By symmetry, the decomposition can be chosen such that for any pair of renamings $\pi,\varpi$ and $k \in \{1,\ldots,K^{\pi,\varpi}\}$ we have
\begin{itemize}\setlength{\itemsep}{0pt}
	\item $K^{\text{uni}} + K^{\text{dup}} = K^{\pi,\varpi}$, 
	\item $\beta_k = \beta_k^{\pi,\varpi}$, 
	\item if $\varphi^k = \varphi^{k,\text{id},\text{id}}$ is a strategyproof unilateral  (or duple) mechanism, than $\varphi^{k,\pi,\varpi}$ is a strategyproof unilateral (or duple) mechanism with $\varphi^k_{\varpi(j)}(\bm P^{\pi,\varpi}) = \varphi_{j}^{k,\pi,\varpi}(\bm P)$,
	\item without loss of generality, $\varphi^k$ is unilateral for $k \in \{1,\ldots,K^{\text{uni}}\}$ and duple for $k \in \{K^{\text{uni}}+1,\ldots,K^{\text{uni}}+K^{\text{dup}}\}$. 
\end{itemize} 
Averaging over all renamings $\pi,\varpi$ we get
\begin{equation}
	\varphi 
		= \sum_{\pi,\varpi} \left(\frac{1}{n!m!}\right) \sum_{k=1}^{K^{\pi,\varpi}} \beta_k^{\pi,\varpi} \varphi^{k,\pi,\varpi} 
		= \sum_{k=1}^{K^{\text{uni}} + K^{\text{dup}}} \sum_{\pi,\varpi} \left(\frac{\beta_k}{n!m!}\right)\varphi^{k,\pi,\varpi}.
	\label{EQ:DECOMPOSITION_INTO_ATOMS}
\end{equation}
If $\varphi^{k,\text{id},\varpi}$ is strategyproof and duple, then 
\begin{equation}
	{\text{dup}}^{k,\varpi} = \sum_{\pi} \left(\frac{1}{n!}\right) \varphi^{k,\pi,\varpi}
\end{equation}
is strategyproof and duple as well. 
Similarly, if $\varphi^{k,\pi,\text{id}}$ is strategyproof and unilateral, then 
\begin{equation}
	{\text{uni}}^{k,\pi} = \sum_{\varpi} \left(\frac{1}{m!}\right) \varphi^{k,\pi,\varpi}
\end{equation}
is strategyproof and unilateral as well. 
With this we can rewrite (\ref{EQ:DECOMPOSITION_INTO_ATOMS}) as 
\begin{equation}
	\varphi 
		= \sum_{k=1}^{K^{\text{uni}}} \sum_{\pi} \frac{\beta_k}{n!}{\text{uni}}^{k,\pi}
		+ \sum_{k=K^{\text{uni}} + 1}^{K^{\text{uni}} + K^{\text{dup}}} \sum_{\varpi} \frac{\beta_k}{m!}{\text{dup}}^{k,\varpi},
\end{equation}
which is precisely the symmetric decomposition (\ref{EQ:SYMMETRIC_DECOMPOSITION}). 
\end{proof}
The symmetric decomposition (\ref{EQ:SYMMETRIC_DECOMPOSITION}) is a consequence of \citeauthor{Gibbard1977ManipulationOfVotingWithChance}'s strong characterization and the fact that for any anonymous, neutral mechanism we have $\varphi = \varphi^{\pi,\varpi}$. 
It is \emph{symmetric} in the sense that for any component ${\text{uni}}^k$ (or ${\text{dup}}^k$) that occur with coefficient $\beta_k$, the corresponding component ${\text{uni}}^{k,\pi}$ (or ${\text{dup}}^{k,\varpi}$) occur with the same coefficient. 
(\ref{EQ:SYMMETRIC_DECOMPOSITION}) decomposes $\varphi$ into two parts: 
a neutral part on the left, that gets ``anonymized'' by randomization, 
and an anonymous part on the right that gets ``neutralized'' by randomization. 
\subsection{Connection to Gibbard's Strong Characterization of Strategyproof Random Mechanisms}
\label{APP:ANON_NEUTR_SYMM_DECOMP:GIBBARD_CONNECTION}
In Section \ref{SEC:COMPUTE:ANALYTIC_EXAMPLES} we have studied the Pareto frontier for two different desiderata, namely Plurality and Veto scoring. 
Juxtaposing the two examples shows an interesting connection to earlier work that characterized strategyproof random mechanisms: 
\citet{Gibbard1977ManipulationOfVotingWithChance} showed that in a domain with three or more alternatives and strict preferences, any mechanism that is strategyproof must be a probability mixture (i.e., a hybrid with an arbitrary number of components) of strategyproof \emph{unilateral} mechanisms (i.e., whose random outcomes depend only on the preference order of a single agent) and strategyproof \emph{duple} mechanisms (i.e., which assign positive probability to at most two alternatives). 

Our examples show that the optimal strategyproof mechanisms for the particular problems are exactly the two extremes of this representation: 
when the desideratum is based on Plurality scoring, optimal strategyproof mechanism arise by randomizing over unilateral mechanisms only, namely dictatorships.  
Conversely, when the desideratum is based on Veto scoring, optimal strategyproof mechanisms arise by randomizing over duple mechanisms only, namely the majority vote between two alternatives. 
Thus, our two examples teach us in what sense unilateral and duple mechanisms can be understood to reside on ``opposite ends'' of the spectrum of strategyproof mechanisms. 
\section{The Algorithm \FindBounds}
\label{APP:ALGORITHM}
Knowing the structure of the Pareto frontier from Theorem \ref{THM:FRONTIER_STRUCTURE}, we can take a more refined approach to the problem of identifying the Pareto frontier (than applying \FindOpt\ to various manipulability bounds). 
The algorithm \FindBounds\ (Algorithm \ref{ALG:FIND_SB}) applies \FindOpt\ sequentially to determine the signatures of optimal mechanisms at different manipulability bounds. 
It uses the information obtained in each step to interpolate the most promising candidates for supporting manipulability bounds. 

\begin{algorithm}[t]
\caption{$\FindBounds$}
\textbf{Input}: bound $\underline{\varepsilon} > 0$ \\
\textbf{Variables}:
	set of supporting manipulability bounds $\mathtt{supp.bnds}$,
	stacks of \emph{unverified}, \emph{verified}, and \emph{outer} segments $\mathtt{segments.u}$, $\mathtt{segments.v}$, $\mathtt{segments.o}$ \\
\Begin{
	$\mathtt{supp.bnds} \leftarrow \left\{ 0 \right\}$  \\
	$\mathtt{segments.u} \leftarrow \left\{ \left(\textit{signature}(\underline{\varepsilon}),\textit{signature}(1)\right) \right\}$  \\
	$\mathtt{segments.v} \leftarrow \emptyset$  \\
	$\mathtt{segments.o} \leftarrow \left\{ \left(\textit{signature}(0),\textit{signature}(\underline{\varepsilon})\right),\left(\textit{signature}(1),\textit{signature}(2)\right) \right\}$  \\
	
	\While{$\mathtt{segments.u} \neq \emptyset$}{
	
		$(\mathtt{P}^-,\mathtt{P}^+) \leftarrow \textit{pop}\left(\mathtt{segments.u}\right)$ \\
		
		$(\mathtt{P}^{--},\mathtt{P}^-), (\mathtt{P}^{+},\mathtt{P}^{++}) \in \mathtt{segments.v} \cup \mathtt{segments.u} \cup \mathtt{segments.o}$ \\
		
		$\mathtt{e} \leftarrow \left( \textit{affine.hull}\left(\{\mathtt{P}^{--},\mathtt{P}^-\}\right) \cap \textit{affine.hull}\left(\{\mathtt{P}^{--},\mathtt{P}^-\}\right) \right)_{\varepsilon} $ \\
		
		$\mathtt{P} \leftarrow \textit{signature}(\mathtt{e}) $ \\
		
		\If{$\mathtt{P} \in \left( \textit{affine.hull}\left(\{\mathtt{P}^{--},\mathtt{P}^-\}\right) \cap \textit{affine.hull}\left(\{\mathtt{P}^{+},\mathtt{P}^{++}\}\right) \right)$}
			{
			$\mathtt{supp.bnds} \leftarrow \mathtt{supp.bnds} \cup \left\{ \mathtt{P} \right\}$ \\
			
			$\mathtt{segments.v} \leftarrow \mathtt{segments.v} \cup \left\{ \left(\mathtt{P}^{--},\mathtt{P}^-\right), \left(\mathtt{P}^{+},\mathtt{P}^{++}\right), \left( \mathtt{P}^-,\mathtt{P} \right), \left( \mathtt{P},\mathtt{P}^+ \right) \right\}$ \\
			}
			
		\ElseIf{$\mathtt{P} \in \textit{affine.hull}\left(\mathtt{P}^-,\mathtt{P}^+ \right)$}
			{
			$\mathtt{segments.v} \leftarrow \mathtt{segments.v} \cup \left\{ \left( \mathtt{P}^-,\mathtt{P}^+ \right) \right\}$ \\
			}
			
		\Else
			{
			$\mathtt{segments.u} \leftarrow \mathtt{segments.u} \cup \left\{ \left( \mathtt{P}^-,\mathtt{P} \right), \left( \mathtt{P},\mathtt{P}^+ \right) \right\}$ \\
			}
	}
	\Return $\mathtt{supp.bnds}$
}
\label{ALG:FIND_SB}
\end{algorithm}
\begin{algorithm}[t]
\caption{\FindLower}
\textbf{Variables}: signatures $\mathtt{sign^0},\mathtt{sign},\mathtt{sign^+}$, bound $\underline{\varepsilon}$ \\
\Begin{

	$\underline{\varepsilon} \leftarrow 1/2$ \\
	$\mathtt{sign^0} \leftarrow \textit{signature}(0)$, $\mathtt{sign^+} \leftarrow \textit{signature}(1)$, $\mathtt{sign} \leftarrow \textit{signature}(\underline{\varepsilon})$ \\
	
	\While{$\mathtt{sign} \notin \textit{affine.hull}(\mathtt{sign^0},\mathtt{sign^+})$}{
		$\mathtt{sign^+} \leftarrow \mathtt{sign}$, $\underline{\varepsilon} \leftarrow \underline{\varepsilon}/2$, $\mathtt{sign} \leftarrow \textit{signature}(\underline{\varepsilon})$\\
	}
	\Return $\underline{\varepsilon}$
}
\label{ALG:FIND_DELTA}
\end{algorithm}

Before we state our formal result about the correctness, completeness, and runtime of \FindBounds, we provide an outline of how the algorithm works. 
Observe that by Theorem \ref{THM:FRONTIER_STRUCTURE}, the problem of computing all supporting manipulability bounds is equivalent to identifying the path of the monotonic, convex, piecewise linear function $\varepsilon\mapsto\deficit(\varepsilon)$. 
\FindBounds\ keeps an inventory of ``known signatures'' and ``verified segments'' on this path. It uses these to interpolate and verify new segments. 
\begin{description}
	\setlength{\itemsep}{.5em}
	\item[Interpolation:] Suppose that we know four points $s_0=(\varepsilon_0,\deficit_0),s_1=(\varepsilon_1,\deficit_1),s_2=(\varepsilon_2,\deficit_2),s_3=(\varepsilon_3,\deficit_3)$ with $\varepsilon_0 < \varepsilon_1 < \varepsilon_2 < \varepsilon_3$ on the path. 
		If the two segments $\overline{s_0\ s_1}$ and $\overline{s_2\ s_3}$ are linearly independent, their affine hulls have a unique point of intersection $\varepsilon'$ over the interval $[\varepsilon_1,\varepsilon_2]$. 
		This point of intersection is the candidate for a new supporting manipulability bound that \FindBounds\ considers. 
		The left plot in Figure \ref{FIG:INTERPOLATION} illustrates the geometry of this step. 
	\item[Verification:] Once we have identified the candidate $\varepsilon'$, we use \FindOpt\ to compute the deficit $\deficit' = \deficit(\varepsilon')$.
		The signature $(\varepsilon',\deficit')$ is either on or below the straight line connecting the signatures $(\varepsilon_1,\deficit_2)$ and $(\varepsilon_2,\deficit_2)$. 
		In the first case, if all three signatures lie on a single straight line segment, then we can infer that this line segment must be part of the path of $\varepsilon\mapsto\deficit(\varepsilon)$. 
		We can mark this segment as ``verified,'' because there are no supporting manipulability bounds in the open interval $(\varepsilon_1,\varepsilon_2)$ (but possibly at its limits). 
		In the second case, if $(\varepsilon',\deficit')$ lies strictly below the line segment, then there must exist at least one supporting manipulability bound in the open interval $(\varepsilon_1,\varepsilon_2)$ at which the path has a real ``kink.''  
		The right plot in Figure \ref{FIG:INTERPOLATION} illustrates the geometry of this step. 
		The signature $(\varepsilon',\deficit')$ must lie somewhere on the gray vertical line at $\varepsilon'$.
\end{description}
In this way \FindBounds\ repeatedly identifies and verifies line segments on the path of $\varepsilon\mapsto\deficit(\varepsilon)$. 
If no additional supporting manipulability bounds exist on a segment, then this segment is verified. 
Otherwise, one more point $(\varepsilon',\deficit')$ is added to the collection of known signatures, and the algorithm continues to interpolate.
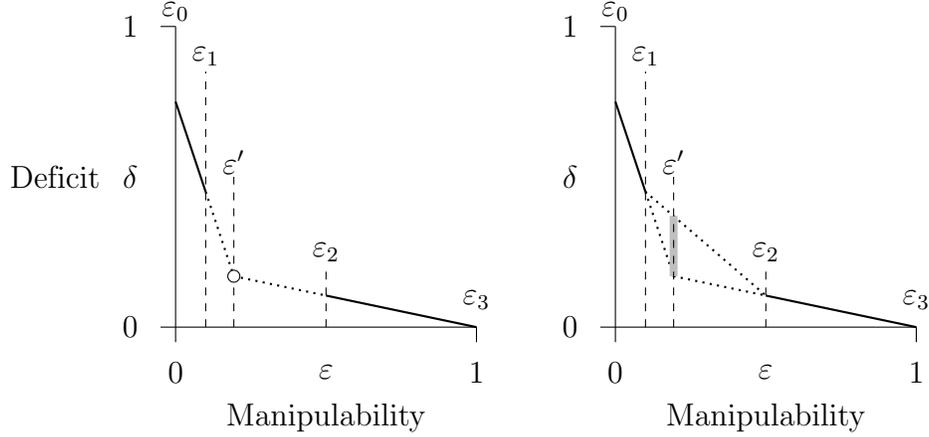
\begin{figure}[t]
\centering
\begin{tikzpicture}[scale=4]
\draw[-] (-0.05,0) -- (1.0,0) coordinate (x axis);
\draw[-] (0,-0.05) -- (0,1.0) coordinate (y axis);
\draw[-] (1,-0.05) -- (1,0);
\draw[-] (-0.05,1) -- (0,1);
\draw (-0.15,1) node {1};
\draw (1.0, -0.15) node {1};
\draw (0, -0.15) node {0};
\draw (-0.15,0) node {0};

\draw (0.5,-0.3) node { Manipulability};
\draw (0.5,-0.15) node { $\varepsilon$};
\draw (-.4,0.5) node {Deficit};
\draw (-.15,0.5) node {$\deficit$};
\draw[-,dashed] (.19354838709,.0) -- (.19354838709,.5);
\draw (6/31,0.55) node {$\varepsilon'$};
\draw[line width=.8pt] (0,.75) -- (.1,.45);
\draw[line width=.8pt] (.5,.105) -- (1,0);
\draw[line width=.8pt,dotted] (.1,.45) -- (.19354838709,0.1693548387) -- (.5,.105);
\draw[-,dashed] (.1,.0) -- (.1,.85);
\draw (.1,0.9) node {$\varepsilon_1$};
\draw[-,dashed] (.5,.0) -- (.5,.2);
\draw (.5,0.25) node {$\varepsilon_2$};
\draw (0.00,1.05) node {$\varepsilon_0$};
\draw (1.0,0.09) node {$\varepsilon_3$};
\draw [fill=White](0.19354838709,0.1693548387) circle [radius=0.02];
\end{tikzpicture}
\hspace{1em}
\begin{tikzpicture}[scale=4]
\draw[-] (-0.05,0) -- (1.0,0) coordinate (x axis);
\draw[-] (0,-0.05) -- (0,1.0) coordinate (y axis);
\draw[-] (1,-0.05) -- (1,0);
\draw[-] (-0.05,1) -- (0,1);
\draw (-0.15,1) node {1};
\draw (1.0, -0.15) node {1};
\draw (0, -0.15) node {0};
\draw (-0.15,0) node {0};

\draw (0.5,-0.3) node { Manipulability};
\draw (0.5,-0.15) node { $\varepsilon$};
\draw (-.15,0.5) node {$\deficit$};

\draw[line width=.8pt] (0,.75) -- (.1,.45);
\draw[line width=.8pt] (.5,.105) -- (1,0);
\draw[line width=.8pt,dotted] (.1,.45) -- (0.19354838709,21/124) -- (.5,.105) -- (.1,.45) ;
\draw[line width=3pt,MedGray] (0.19354838709,0.16935483871) -- (0.19354838709,.37);
\draw[-,dashed] (0.19354838709,.0) -- (0.19354838709,.5);
\draw (0.19354838709,0.55) node {$\varepsilon'$};
\draw[-,dashed] (.1,.0) -- (.1,.85);
\draw (.1,0.9) node {$\varepsilon_1$};
\draw[-,dashed] (.5,.0) -- (.5,.2);
\draw (.5,0.25) node {$\varepsilon_2$};
\draw (0.00,1.05) node {$\varepsilon_0$};
\draw (1.0,0.09) node {$\varepsilon_3$};
\end{tikzpicture}

\caption{Illustration of interpolation (left) and verification step (right).}
\label{FIG:INTERPOLATION}
\end{figure}

Next, we prove Proposition \ref{PROP:RUNTIME}.
\begin{proof}[Proof of Proposition \ref{PROP:RUNTIME}]
\emph{\PropComputableStatement\footnote{For technical reasons, \FindBounds\ assumes knowledge of some value $\underline{\varepsilon} \in (0,\varepsilon_1]$, i.e., some strictly positive lower bound for the smallest non-trivial supporting manipulability bound. 
This is needed to initialize the interpolation process. 
Algorithm \FindLower\ (Algorithm \ref{ALG:FIND_DELTA} identifies a suitable value of $\underline{\varepsilon}$.}}

\medskip
\PropComputableProof%
\end{proof}

The interpolation process in \FindBounds\ is set up in such a way that at most 4 executions of \FindOpt\ are required for finding each supporting manipulability bound. 
Thus, assuming that representatives of $\OPT(\varepsilon_k)$ must be determined by executing \FindOpt\ at each $\varepsilon_k$, \FindBounds\ is at most a constant factor $4$ slower than an algorithm that correctly ``guesses'' all supporting manipulability bounds and uses \FindOpt\ for these only.

\section{Alternative Domains and Desiderata} 
\label{APP:GENERAL}
With the exception of the examples in Sections \ref{SEC:COMPUTE:ANALYTIC_EXAMPLES}, we have formulated our results for the general domain with weak preferences. 
In the following Section \ref{APP:GENERAL:DOMAIN_RESTRICTION}, we explain how our findings continue to hold in many other interesting domains. 
Subsequently, in Section \ref{APP:GENERAL:VALUE}, we discuss the expressiveness of desideratum functions for the purpose of encoding desiderata. 
\subsection{Domain Restrictions}
\label{APP:GENERAL:DOMAIN_RESTRICTION}
We have proven our results (Theorems \ref{THM:ASP_EQUIV_LINEAR_CONSTRAINTS} through \ref{THM:FRONTIER_STRUCTURE} and Proposition \ref{PROP:OPT_EQUIVALENT_LP}) for the unrestricted preference domain with indifferences. 
However, all these results continue to hold when the domain is restricted by only admitting a subset of the possible preference profiles. 
Formally, let $\bm{\mathcal{P}} \subseteq \mathcal{P}^N$ be \emph{any} subset of the set of preference profiles; then all our results still hold, even if we admit only preference profiles from $\bm{\mathcal{P}}$. 
This extension is possible because none of our proofs make use of the richness of the full domain. 

Important domains that can be expressed as domain restrictions $\bm{\mathcal{P}} \subset \mathcal{P}^N$ are: 
\begin{itemize}
	\setlength{\itemsep}{0pt}
	\item The \emph{full domain of strict ordinal preferences}, which is the basis for the classical impossibility results pertaining to strategyproofness of deterministic mechanisms \citep{Gibbard1973ManipulationOfVotingSchemes,Satterthwaite1975StrategyproofnessAndArrowsConditions} and random mechanisms \citep{Gibbard1977ManipulationOfVotingWithChance}. 
	\item The \emph{assignment domain}, where each agent must be assigned to a single, indivisible object (e.g., a seat at a public school) \citep{AbdulSonmez2003SchoolChoiceAMechanismDesignApproach}.
	\item The \emph{two-sided matching domain}, where agents of two different types must be matched to each other to form pairs (e.g., men and women, doctors and hospitals, mentors and mentees), and variations of this domain, such as \emph{matching with couples} \citep{Roth1984EvolutionOfMedicalLaborMarket}.
	\item The \emph{combinatorial assignment domain}, where agents receive bundles of indivisible objects (e.g., schedules with multiple courses) \citep{Budish2011TheCombinatorialAssignmentProblem}
\end{itemize}
Of course, the fact that our results continue to hold in these domains does not mean that the actual Pareto frontiers will be equal across the different domains. 
On the contrary, one would expect optimal mechanisms to be highly adapted to the particular domain in which they are designed to operate. 
\begin{remark}[Full Utility Assumption]
\label{REM:FULL_UTILITIY}
One essential ingredient to our results is the fact that $\varepsilon$-approximate strategyproofness can be equivalently expressed via a finite set of linear constraints. 
Theorem \ref{THM:ASP_EQUIV_LINEAR_CONSTRAINTS}, which yields this equivalence, relies on the assumption of \emph{full utility}: 
given a preference order $P_i$, the agent $i$ can have vNM utility function $u_i \in U_{P_i}$ that represents this preference order (and is bounded between 0 and 1). 
A domain restriction that excludes \emph{subsets of utility functions} would violate this condition. 
For example, suppose that we imposed the additional restriction that the agent's utility are \emph{normalized} (i.e., $\left\|u_i\right\|_2 = \sqrt{\sum_{j\in M} u_i(j)^2} = 1$). 
This restriction limits the gain that an agent can possibly obtain from misreporting in a non-linear way. 
The linear constraints from Theorem \ref{THM:ASP_EQUIV_LINEAR_CONSTRAINTS} would be sufficient for $\varepsilon$-approximate strategyproofness. 
However, they would no longer be necessary, and we would lose \emph{equivalent}. 
\end{remark}
\subsection{Encoding Desiderata via Desideratum Functions}
\label{APP:GENERAL:VALUE}
desideratum functions are a very versatile means to express different design goals that the mechanism designer may have. 
We have already seen how they can be used to reflect desiderata based on Plurality scoring (Example \ref{EX:DEFICIT_PLURALITY}), Condorcet consistency (Examples \ref{EX:DEFICIT_CONDORCET}), or Veto scoring (Section \ref{SEC:COMPUTE:ANALYTIC_EXAMPLES}). 
We now give additional examples that illustrate how desideratum functions can be used to encode common desiderata. 
\paragraph{Binary Desiderata}
Many desirable properties of mechanisms are simply properties of the alternatives they select at any given preference profile. 
For instance, \emph{Pareto optimality} requires the mechanism to choose alternatives that are not Pareto dominated by any other alternative at the respective preference profile.
In general, let $\Theta$ be a desirable property that an alternative can have at some preference profile (e.g., Pareto optimality). 
Desideratum functions can capture $\Theta$ by setting
\begin{equation}
	\desifct(j,\bm P) = \left\{%
		\begin{array}{ll}
			1, & \text{ if alternative }j\text{ has property }\Theta\text{ at }\bm P, \\
			0, & \text{ else.}
		\end{array}
	\right.
\end{equation}
The $\desifct$-deficit of a mechanism at a given preference profile,
\begin{equation}
	\deficit_{\desifct}(\varphi(\bm P), \bm P) = \max_{j \in M} \desifct(j,\bm P) - \sum_{j\in M} \varphi_j(\bm P) \cdot \desifct(j,\bm P),
\end{equation}
has a straightforward interpretation: 
it is simply the ``probability that an alternative selected by $\varphi$ at $\bm P$ violates property $\Theta$ at $\bm P$.'' 
Consequently, if $\deficit$ is the worst-case deficit, then the probability that $\varphi$ selects an alternative that violates $\Theta$ is at most $\deficit(\varphi)$ across all preference profiles. 
Alternatively, if $\deficit^{\mathds{P}}$ is the ex-ante deficit, then $\deficit^{\mathds{P}}(\varphi)$ is the ex-ante probability (i.e., before the preferences have been reported) that $\varphi$ selects an alternative that violates $\Theta$, given the prior distribution $\mathds{P}$ over agents' preference profiles. 

Common binary desiderata include 
\begin{itemize}
	\setlength{\itemsep}{0pt}
	\item \emph{unanimity}:  if all agents agree on a first choice, this alternative should be selected, 
	\item \emph{Condorcet consistency}: if some alternatives weakly dominate any other alternatives in a pairwise majority comparison, one of these alternatives should be selected, 
	\item \emph{Pareto optimality}: only undominated alternatives should be selected,
	\item \emph{egalitarian fairness}: let $R(j) = \max_{i \in N} \text{rank}_{P_i}(j)$ be the rank of alternative $j$ in the preference order of the agent who likes it least, then only alternatives with minimal $R(j)$ should be selected.
\end{itemize}
While not all of these desiderata are in conflict with strategyproofness individually, their combinations may be. 
It is possible to encode combinations of binary properties: 
given two desideratum functions $\desifct_1$ and $\desifct_2$, the minimum $\desifct = \min(\desifct_1, \desifct_2)$ is again a desideratum function, and it expresses the desideratum that both properties (encoded by $\desifct_1$ and $\desifct_2$) should hold simultaneously; 
and the maximum $\desifct = \max(\desifct_1, \desifct_2)$ expresses that at least one of the properties should hold. 

Even in the absence of a clear-cut objective, a desideratum may be \emph{implicitly} specified via a target mechanism. 
A \emph{target mechanism} $\phi$ is a mechanism that one would like to use if there were no concerns about manipulability \citep{BirrellPass2011ApproximatelyStrategyproofVoting}. 
The desideratum function induced by $\phi$ takes value 1 at the alternatives that are selected by $\phi$ at the respective preference profiles, or equivalently, $\desifct(j,\bm P) = \phi_j(\bm P)$. 
An analogous construction is possible for a \emph{target correspondence} $\Phi$, which is a mechanism that selects a \emph{set of alternatives}. 
To reflect a target correspondence, the desideratum function can be chosen as $\desifct(j,\bm P) = 1$ if $ j \in \Phi(\bm P)$ and $\desifct(j,\bm P) = 0$ otherwise, where we denote by $\Phi(\bm P)$ the set of alternatives selected by the correspondence $\Phi$ at $\bm P$. 

In the assignment domain, our notion of Pareto optimality turns into \emph{ex-post efficiency}, which is an important baseline requirement for many assignment mechanisms. 
\citet{Featherstone2011RankBasedRefinementWP} introduced \emph{$v$-rank efficiency}, which is a refinement of ex-post efficiency (and implies ordinal efficiency). 
He showed that for any \emph{rank valuation $v$}, the set of $v$-rank efficient assignments is equal to the convex hull of the set of deterministic $v$-rank efficient assignments. 
Thus, $v$-rank efficiency is representable by a desideratum function by setting $\desifct(a,\bm P) = 1$ for all alternatives $a$ that correspond to a $v$-rank efficient deterministic assignment. 
In two-sided matching, it is often important to select \emph{stable matchings} to prevent blocking pairs from matching outside the mechanism and thereby causing unraveling of the market. 
Since stability is a property of deterministic matchings, we can encode this desideratum by a desideratum function. 

With any such desideratum function in hand, we can then define a notion of deficit and use the results in this paper to understand the respective Pareto frontier. 
\begin{remark}[Linearity Assumption]
\label{REM:LINEARITY}
One example of a desideratum that cannot be represented by a desideratum function is the intermediate efficiency requirement of \emph{ordinal efficiency} for assignment mechanisms. 
To see this, note that ordinal efficiency coincides with ex-post efficiency on deterministic assignments. 
However, the convex combination of two ordinally efficient assignments is not necessarily again ordinally efficient.
In general, to be able to encode a desideratum via a desideratum function, the societal value of any random outcome $x$ must depend in a linear fashion on the societal value of the individual alternatives (see Definition \ref{DEF:EXPECTED_VALUE}). 
Ordinal efficiency violates this requirement. 
\end{remark}
\paragraph{Quantifiable Desiderata}
In some situations, there exists a quantified notion of the desirability of different alternatives at different preference profiles. 
Straightforward examples are desiderata based on positional scoring. 
Formally, a \emph{scoring function} $v : \{1,\ldots,m\} \rightarrow \mathds{R}$ is a mapping that associates a \emph{score} $v(r)$ with selecting some agent's $r$th choice alternative. 
Typical examples include \emph{Borda} (i.e., $v(r) = m-r$ for all $r \in \{1,\ldots,m\}$), \emph{Plurality} (i.e., $v(1) = 1$ and $v(r) = 0$ for $r\neq 1$), and \emph{Veto} (i.e., $v(r) = 0$ for all $r\neq m$ and $v(m) = -1$).
The score of an alternative at a preference profile is then determined by summing the scores of this alternative under the different preference orders of the agents. 
Formally,  
\begin{equation}
	\text{sc}_v(j,\bm P) = \sum_{i \in N} v(\text{rank}_{P_i}(j)),
\end{equation} 
where $\text{rank}_{P_i}(j)$ is the rank of alternative $j$ in preference order $P_i$. 
For any scoring-based desideratum, the respective desideratum function arises by scaling $\text{sc}_v$ as follows: 
\begin{equation*}
	\desifct(j,\bm P) = 
	\left\{
		\begin{array}{ll}
			\frac{\text{sc}_v(j,\bm P) - \min_{j' \in M} \text{sc}_v(j',\bm P)}{\max_{j' \in M} \text{sc}_v(j',\bm P) - \min_{j' \in M} \text{sc}_v(j',\bm P)}, 
			& \text{if }\max_{j' \in M} \text{sc}_v(j',\bm P) > \min_{j' \in M} \text{sc}_v(j',\bm P)\\ 
			0, & \text{else}. 
		\end{array}
	\right.
\end{equation*}
In the assignment domain, the rank valuation $v$ is used to determine the $v$-rank value of any assignment. 
This induces a scoring function on the corresponding alternatives. 
Consequently, instead of viewing $v$-rank efficiency as a binary desideratum, the quantified desideratum could be to maximize the $v$-rank value. 
This natural measure of societal value was put forward by \citet{Featherstone2011RankBasedRefinementWP}. 

In quasi-linear domains, a common objective is to maximize the sum of the agents' cardinal utilities. 
For situations where this aggregate value is interesting, \citet{ProcacciaRosenschein2006DistortionOfCardinalPrefsInVoting} proposed to maximize a \emph{conservative estimate} of this value instead. 
This (slightly unorthodox) desideratum is also representable by a desideratum function. 
Even though the vNM utilities in our model do not have to be comparable across agents, one may choose to maximize the sum of these utilities. 
However, since agents only report ordinal preferences, we cannot maximize their utilities directly. 
Instead, we can maximize a \emph{conservative estimate} for this sum. 
Formally, the respective desideratum would be to maximize 
\begin{equation}
	\desifct(j,\bm P) = \inf_{u_1 \in U_{P_1},\ldots,u_n \in U_{P_n}} \sum_{i \in N} u_i(j).
\end{equation}
Recall that we bounded utilities between 0 and 1, and let us assume that $\max_{j \in M} u_i(j) = 1$. 
Then the aggregate utility from any outcome $x \in \Delta(M)$ is lowest if agents have utility (close to) 0 for all except their first choices. 
Thus, maximizing the conservative lower bound for aggregate utility in our model corresponds to selecting an alternative that is the first choice of as many agents as possible. 
This is equivalent to the quantitative measure of societal value induced by Plurality scoring. 
\section{Omitted Proofs} 
\label{APP:PROOFS}
\subsection{Proof of Theorem \ref{THM:ASP_EQUIV_LINEAR_CONSTRAINTS}}
\label{APP:PROOFS:ASP_EQUIV_LINEAR_CONSTRAINTS}
\begin{proof}[Proof of Theorem \ref{THM:ASP_EQUIV_LINEAR_CONSTRAINTS}]
\emph{\ThmASPEquivLinConStatement}

\medskip
\ThmASPEquivLinConProof%
\end{proof}
\subsection{Proof of Proposition \ref{PROP:OPTIMAL_NON_EMPTY}}
\label{APP:PROOFS:OPTIMAL_NON_EMPTY}
\begin{proof}[Proof of Proposition \ref{PROP:OPTIMAL_NON_EMPTY}]
\emph{\PropOptimalMechanismsNotEmptyStatement}

\medskip
\PropOptimalMechanismsNotEmptyProof
\end{proof}
\subsection{Proof of Theorem \ref{THM:HYBRID_GUARANTEES}}
\label{APP:PROOFS:HYBRID_GUARANTEES}
\begin{proof}[Proof of Theorem \ref{THM:HYBRID_GUARANTEES}]
\emph{\ThmHybridGuaranteesStatement}

\medskip
\ThmHybridGuaranteesProof
\end{proof}
\subsection{Proof of Corollary \ref{COR:PF_MONOTONIC_CONVEX}}
\label{APP:PROOFS:PF_MONOTONIC_CONVEX}
\begin{proof}[Proof of Corollary \ref{COR:PF_MONOTONIC_CONVEX}]
\emph{\CorMonotonicConvexStatement}
\medskip

\CorMonotonicConvexProof
\end{proof}
\subsection{Proof of Theorem \ref{THM:FRONTIER_STRUCTURE}}
\label{APP:PROOFS:FRONTIER_STRUCTURE}
\begin{proof}[Proof of Theorem \ref{THM:FRONTIER_STRUCTURE}]
\emph{\ThmFrontierStructureStatement}

\medskip
\ThmFrontierStructureProof%
\end{proof}
\subsection{Proof of Proposition \ref{PROP:PLURALITY_PARETO_FRONTIER}}
\label{APP:PROOFS:PLURALITY_PARETO_FRONTIER}
\begin{proof}[Proof of Proposition \ref{PROP:PLURALITY_PARETO_FRONTIER}]
\emph{\PropPluralityParetoFrontierStatement}

\medskip
\PropPluralityParetoFrontierProof
\end{proof}
\subsection{Proof of Proposition \ref{PROP:VETO_PARETO_FRONTIER}}
\label{APP:PROOFS:VETO_PARETO_FRONTIER}
\begin{proof}[Proof of Proposition \ref{PROP:VETO_PARETO_FRONTIER}]
\emph{\PropVetoParetoFrontierStatement}

\medskip
\PropVetoParetoFrontierProof
\end{proof}
\end{document}